\documentclass[twoside]{article}
\usepackage{}
\usepackage[utf8]{inputenc}
\usepackage{amsfonts}
\usepackage[left=2.5cm,right=2.5cm,top=2cm,bottom=3cm]{geometry}
\usepackage{epsfig,graphics}
\usepackage{amssymb,amsmath,amsthm,bm}
\usepackage{url}
\usepackage{float}
\usepackage{multirow}
\usepackage{colortbl}  
\newtheorem{theo}{Theorem}
\newtheorem{lem}{Lemma}

\newtheorem{rem}{Remark}
\newtheorem{mdef}{Definition}

\newtheorem{alg}{Algorithm}

\def\donchitre#1#2{\vskip 6.5cm\noindent
\parbox[t]{1in}{\special{eps:#1.eps x=6.5cm y=5.5cm}}
\hbox to 7cm{}\parbox[t]{0.0cm}{\special{eps:#2.eps x=6.5cm y=5.5cm}}}

\def\ZZ{{\mathbb Z}}
\def\RR{{\mathbb R}}

\def\ZZ{{\mathbb Z}}
\def\RR{{\mathbb R}}

\def\fs{\mathfrak{S}}
\def\fr{\mathfrak{R}}

\def\ga{\alpha}

\def\go{\omega}

\def\gs{\sigma}
\def\gl{\lambda}
\def\wt{\widetilde}

\def\gt{\triangle}

\def\gp{{\prime}}
\def\b0{{\bf 0}}
\def\1{{\bf 1}}

\def\cG{\mathcal G}

\def\wh{\widehat}

\def\wt{\widetilde}

\makeatletter
\def\widebreve{\mathpalette\wide@breve}
\def\wide@breve#1#2{\sbox\z@{$#1#2$}%
     \mathop{\vbox{\m@th\ialign{##\crcr
\kern0.08em\brevefill#1{0.8\wd\z@}\crcr\noalign{\nointerlineskip}%
                    $\hss#1#2\hss$\crcr}}}\limits}
\def\brevefill#1#2{$\m@th\sbox\tw@{$#1($}%
  \hss\resizebox{#2}{\wd\tw@}{\rotatebox[origin=c]{90}{\upshape(}}\hss$}
\makeatletter
\def\wb{\widebreve}

\title{A Chirplet Transform-based Mode Retrieval Method for Multicomponent Signals with Crossover Instantaneous Frequencies}
\author{Lin Li\thanks{School of Electronic Engineering, Xidian University, Xi\rq{}an, China.
		e-mail: lilin@xidian.edu.cn.},
Ningning Han\thanks{Shenzhen Key Laboratory of Advanced Machine Learning and Applications, College of Mathematics and Statistics, Shenzhen University, Shenzhen, 518060, China.
e-mail: ningninghan@szu.edu.cn.},
Qingtang  Jiang\thanks{Department of Mathematics and Statistics, University of Missouri-St. Louis, St. Louis,  MO, USA.
	e-mail:jiangq@umsl.edu.},
Charles K.Chui\thanks{Department of Statistics, Stanford
University, Stanford, CA 94305, USA, and 
Department of Mathematics, Hong Kong Baptist
University, Hong Kong. e-mail: ckchui@stanford.edu.}
}
\date{}
\begin{document}
\maketitle
(Note: This is the revised paper of  ``A  separation method for multicomponent non-stationary signals with crossover instantaneous frequencies,\rq\rq{} which was submitted to IEEE Trans IT on Feb 14, 2020,  ms \# IT-20-0113. Compared to the original paper, some references were added. In addition,  mode retrieval comparison results with Fast-IF  and ACMD were provided.)

 \begin{abstract}
In nature and engineering world, the acquired signals are usually affected by multiple complicated factors and appear as multicomponent nonstationary modes. In such and many other situations, it is necessary to separate these signals into a finite number of monocomponents to represent the intrinsic modes and underlying dynamics implicated in the source signals. In this paper, we consider the mode retrieval of a multicomponent signal which has crossing instantaneous frequencies (IFs), meaning that some of the components of the signal overlap in the time-frequency domain. We use the 
chirplet transform (CT) to represent a multicomponent signal in the three-dimensional space of time, frequency and chirp rate and introduce a CT-based signal separation scheme (CT3S) to retrieve modes. In addition, we analyze the error bounds for IF estimation and component recovery with this scheme.  We also propose a matched-filter along certain specific time-frequency lines with respect to the chirp rate to make nonstationary signals be further separated and more concentrated in the three-dimensional space of  CT. Furthermore, based on the approximation of source signals with linear chirps at any local time, we propose an innovative signal reconstruction algorithm, called the group 
 filter-matched CT3S (GFCT3S), which also takes a group of components into consideration simultaneously. GFCT3S is suitable for signals with crossing IFs. It also decreases component recovery errors when the IFs curves of different components are not crossover, but fast-varying and close to one and other. Numerical experiments on synthetic and real signals show our method is more accurate and consistent in signal separation than the empirical mode decomposition, synchrosqueezing transform, and other approaches.
 \end{abstract}


\section{Introduction}
%\IEEEPARstart{F}{or}

This paper studies blind source nonstationary signal separation in which a nonstationary signal is represented as a superposition of Fourier-like oscillatory modes:
\begin{equation}
\label{AHM}
x(t)=A_0(t)+\sum_{k=1}^K x_k(t),  \quad x_k(t)=A_k(t) \cos \big(2\pi \phi_k(t)\big),
\end{equation}
with $A_k(t), \phi_k'(t)>0$ and $A_k(t)$ varying slowly.
%The representation of $x(t)$ in \eqref{AHM}  is called an {adaptive harmonic model (AHM) representation} of $x(t)$,
%where  $A_k(t)$, called the instantaneous amplitudes (IAs)  and $\phi'_k(t)$, called the instantaneous frequencies (IFs)
%can be used to describe the underlying dynamics.
Such a representation, called  an {adaptive harmonic model (AHM) representation},
 is important for extracting information, such as the underlying dynamics, hidden in the nonstationary signal, with  the trend $A_0(t)$, instantaneous amplitudes (IAs) $A_k(t)$  and the instantaneous frequencies (IFs) $\phi'_k(t)$ being used to describe the underlying dynamics.

In nature, many real-world phenomena that can be formulated as signals (or in terms of time series) are often affected by a number of factors and appear as  time-overlapping multicomponent signals in the form of  \eqref{AHM}.
A natural approach to understand and process such phenomena is to decompose, or even better, to separate the multicomponent signals into their basic building blocks $x_k(t)$ (called components, modes or sub-signals) for extracting the necessary features.
% such as volatility, trends, outliers, and the underlying dynamics.
Also, for radar, communications, and other applications, signals often appear in multicomponent modes.
Since these signals are mainly nonstationary, meaning the amplitudes  and/or phases of some or all components change with the time, there have been few effective rigorous methods available for decomposition of them.

The empirical mode decomposition (EMD) algorithm along with the Hilbert spectrum analysis (HSA)
%introduced by Huang et. al.
is a popular method to decompose and analyze nonstationary signals \cite{Huang98}.
%EMD decomposes a nonstationary signal as a superposition of intrinsic mode functions (IMFs) and then the IF of each IMFs is calculated by HSA which results in a representation of the signal as in \eqref{AHM}.
EMD works like a filter bank \cite{Flandrin04,Rehman_Mandic11} to decompose
a nonstationary signal  into a superposition of intrinsic mode functions (IMFs) and a trend, and then the IF of each IMF is calculated by HSA. There are many articles studying the properties of EMD and variants of EMD have been proposed to improve the performance, see e.g. \cite{Flandrin04}-\cite{van20}. %\cite{Li_Jiang19}.
In particular, the separation ability of EMD was discussed in \cite{Rilling08}, which shows that EMD cannot decompose two components when their frequencies are close to each other. The ensemble EMD (EEMD) was proposed to suppress noise interferences \cite{Wu_Huang09}. The original EMD was extended to multivariate signals in \cite{Xu06, Rehman_Mandic10, Rehman_Mandic11}. Alternative sifting algorithms and formulations 
for EMD were introduced in \cite{HM_Zhou09, Oberlin12a, HM_Zhou16}. Similar to the EMD filter bank, the wavelet filter bank for signal decomposition was proposed in \cite{Gilles13}, called empirical wavelet transform. An EMD-like sifting process was recently proposed in \cite{Li_Jiang19} to extract signal components in the linear time-frequency (TF) plane one by one.
EMD is an efficient data-driven approach and no basis of functions is required.
A weakness of EMD or EEMD is that it can easily lead to mode mixtures or artifacts, namely undesirable or false components \cite{Li_Ji09}. In addition, there is no mathematical theorem to guarantee the recovery of the components.

Time-frequency analysis is another method to separate multicomponent signals, which is widely used in engineering fields such as communication, radar and sonar as a powerful tool for analyzing %time-varying nonstationary
signals \cite{Leon_Cohen}. Time-frequency signal analysis and synthesis using the eigenvalue decomposition method have been studied \cite{HMB04,Stank06,Stank18}.
%In particular, an eigenvalue decomposition-based approach  which enables the separation of nonstationary components with overlapped supports in the TF plane has been proposed in \cite{Stank18}.
Recently the synchrosqueezing transform (SST)
%, also called the synchrosqueezed wavelet transform,
was developed %by Daubechies, Lu and Wu
 in \cite{Daub_Lu_Wu11} to provide mathematical theorems to guarantee the component recovery of  nonstationary multicomponent signals. The SST, which was first  introduced
 %introduced by Daubechies and Maes
in 1996, intended for speech signal separation \cite{Daub_Maes96},
is based on the continuous wavelet transform (CWT).
%CWT has scale and time variables. SST re-assigns the scale variable to the frequency variable to sharpen the time-frequency representation of a signal. %as the method of both time and frequency re-assignments studied in \cite{A_Flandrin_reassignment95} (see also \cite{CM_A_Flandrin_reassignment03}.
%In addition, the original signal can be recovered from its SST with its each component being reconstructed by extracting the IF curves in the SST plane. Furthermore,
The short-time Fourier transform (STFT)-based SST
%(denoted by FSST, while the CWT-based SST is usually denoted by WSST)
was also proposed in \cite{Thakur_Wu11,Wu_thesis} and further studied in \cite{MOM14}.
SST provides an alternative to the EMD method and its variants, and it overcomes some limitations of the EMD and EEMD schemes \cite{Meig_Oberlin_M12,Flandrin_Wu_etal_review13}. SST has been used in machine fault diagnosis \cite{Li_Liang_fault12,WCSGTZ18},
crystal image analysis \cite{Yang_Crystal_15,Yang_Crystal_18},
welding crack acoustic emission signal analysis \cite{HLY18},
medical data analysis  %\cite{Wu_breathing14}-\cite{Wu_heartbeat17}.
\cite{Wu_breathing14, Wu_sleep15, Wu_heartbeat17}, etc.
%SST has been used in many applications including machine fault diagnosis \cite{Li_Liang_fault12}, anesthesia evaluation \cite{Wu_anaesthesia14, Chui_Lin_Wu15}, sleep stage assessment \cite{Wu_sleep15} and heart beat classification \cite{Wu_heartbeat17}.

SST works well for sinusoidal signals, but not for broadband time-varying frequency signals.
To provide sharp representations for signals with significantly frequency changes, two methods were proposed.
One is %the generalized SST and
the matching demodulation transform-based SST
(or called the instantaneous frequency-embedded SST) proposed in
 \cite{Wang_etal14, Jiang_Suter17} which changes broadband signals to narrow-band signals (see also \cite{Li_Liang12}). The instantaneous frequency-embedded SST proposed in \cite{Jiang_Suter17} preserves the IFs of the original signals.
%The matching demodulation transform and its STFT-based synchrosqueezing are discussed in \cite{Wang_etal14}.
The other method is the 2nd-order SST introduced in \cite{MOM15} and \cite{OM17}.
%FSST (FSST2) and  the 2nd-order WSST (WSST2) were introduced in \cite{MOM15} and \cite{OM17}
%The 2nd-order SST improves the concentration of the time-frequency representation.
% well on perturbed linear chirps with Gaussian modulated amplitudes.
The higher-order FSST  was presented in \cite{Pham17} and \cite{Lin_Cubic19}, which aims to handle signals containing more general types. Very recently
an adaptive SST %STFT-based SST and CTW-based SST
with a time-varying parameter were introduced in \cite{LCHJJ18, LCJ18}.
%the 2nd-order adaptive FSST and WSST with a time-varying parameter.
They obtained the well-separated condition for multicomponent signals  using the linear frequency modulation to approximate a nonstationary signal at any local time.
In addition, theoretical analysis of adaptive SST %and the 2nd-order adaptive FSST
was obtained in \cite{LJL20,CJLS21}.
SST with a time-varying window width has also been studied in \cite{Wu17, Saito17}.

%To recover individual component $x_k(t)$,  the SST method consists of two steps. First IF $\phi^\gp_k(t)$ 
%$x_k(t)$ is estimated from the SST plane. Secondly, $x_k(t)$ is computed by a definite integral along each estimated IF curve on the SST plane.
%The reconstruction accuracy for $x_k(t)$ depends heavily on the accuracy of the IFs estimation carried out in the first step. On the other hand,  

Recently, a direct time-frequency approach, called signal separation operator (SSO), was introduced in \cite{Chui_Mhaskar15} for multicomponent signal separation. In SSO approach, the components are reconstructed simply  by substituting the time-frequency ridge to SSO. SSO based on CWT was proposed in \cite{Chui_Han20,CJLL20_adpCWT}. 
%  the phase functions of the signal components in the source signal model \eqref{AHM0} are approximated by some linear polynomials at any local time for the purpose of extracting the IFs. 
More recently,  \cite{LCJ20} proposed SSO of ``linear chirp-based model\rq{}\rq{} with 
 the phase functions of the signal components in the source signal model \eqref{AHM} approximated by quadratic polynomials at any local time. This model provides a more accurate component recovery formulae, 
with analysis established in \cite{CJLL20_adpSTFT,CJLL20_adpCWT}. 

 While SST and SSO are mathematically rigorous on IF estimation and component recovery,
 %SSO avoids the second step of the two-step SST method in signal separation.
%Mathematically rigorous signal separation methods SST and SSO
both of them require that the components $x_k(t)$ are well-separated in the time-frequency plane, namely IFs of $x_k(t)$ satisfy
\begin{equation}
 \label{def_sep_cond}
 \phi'_k(t)-\phi'_{k-1}(t)\ge 2\gt, \; 2\le k\le K,
\end{equation}
for some $\gt>0$.  In many applications, multicomponent signals are overlapping in the time-frequency plane, that is the IFs of its components are crossover. For example, in radar signal processing, the micro-Doppler
\begin{figure}[H]
\centering
%\begin{tabular}{ccc}
\begin{tabular}{c@{\hskip -0.2cm}c}
\resizebox {2in}{1.5in} {\includegraphics{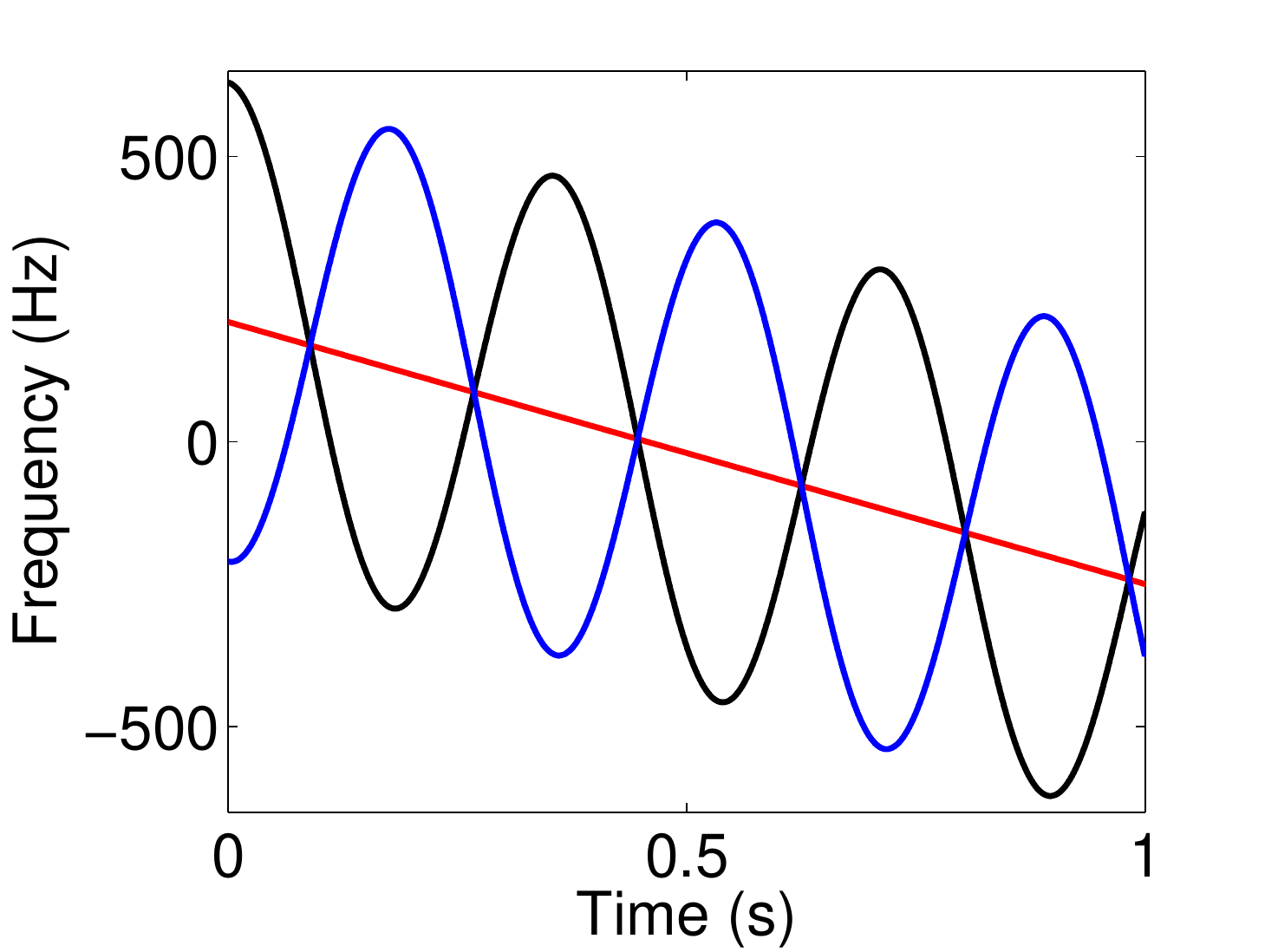}}
\quad & \quad
\resizebox {2in}{1.5in} {\includegraphics{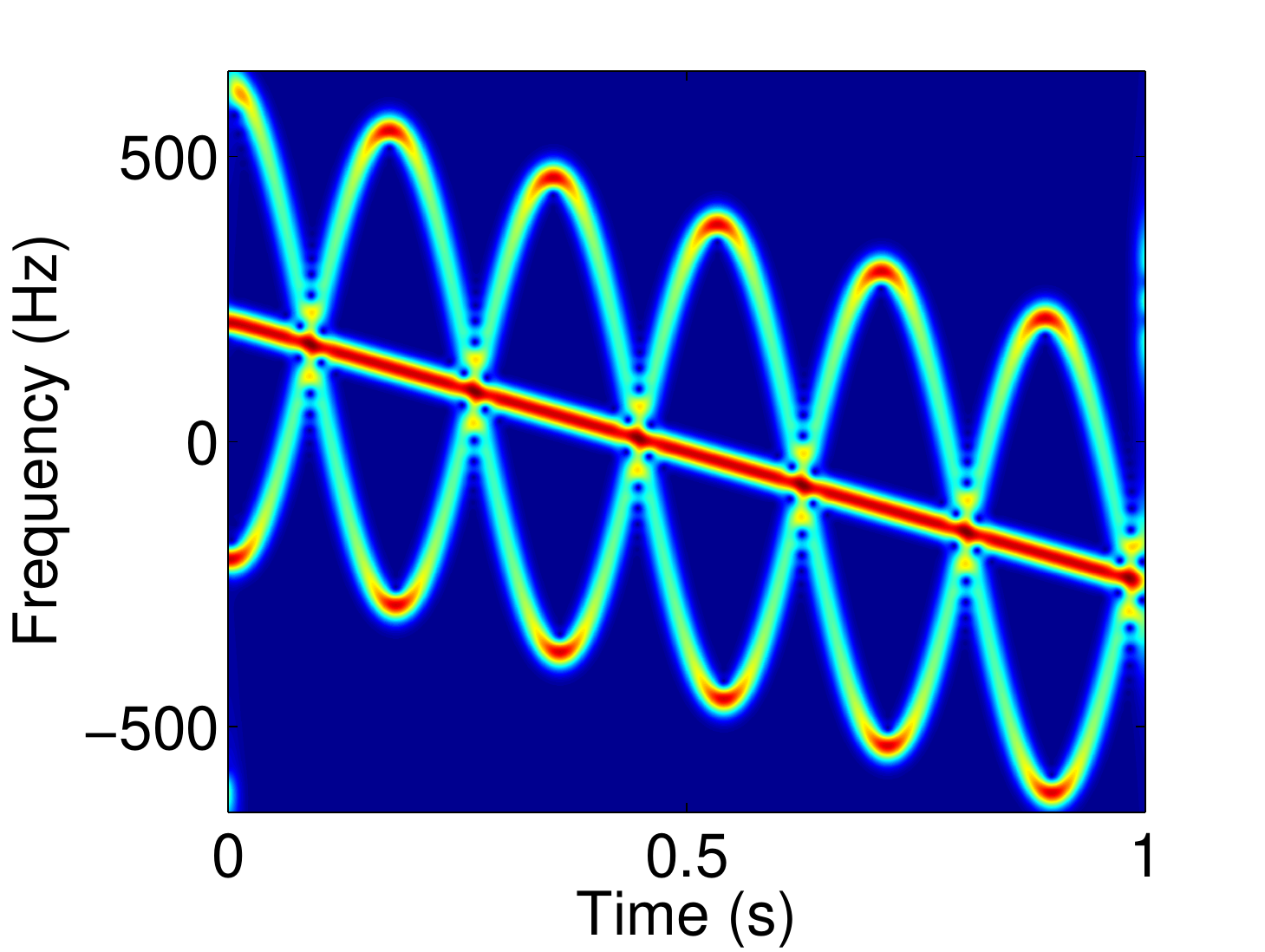}}
\end{tabular}
\vskip -0.3cm
\caption{\small  Micro-Doppler modulations induced by target\rq{}s tumbling (Right) and their %short time Fourier transforms
STFTs (Left).} %from \cite{radar_basic_2016}}
\label{Figure:Micro_Doppler}
\end{figure}
\noindent effects are represented by highly
nonstationary signals, when the target or any structure on the target undergoes micro-motion dynamics, such as mechanical vibrations, rotations, or tumbling and coning motions \cite{radar_basic_2016, Stankovic_compressive_sense_2013}.
Figure \ref{Figure:Micro_Doppler} shows  simulated micro-Doppler modulations (two sinusoidal signals) induced by two target\rq{}s tumbling motions %\cite{radar_basic_2016}
and their short-time Fourier transforms (SFTFs).
 %Similar situation may, for example, arise in communications, when narrowband signals are disturbed by a frequency hopping jammer that is of shorter duration than the considered time-interval, but may also be overlapping with narrowband signals within same intervals. \cite{Stankovic_compressive_sense_2013}
In practice we need to recover each components or at least the main body signature corresponding to the target\rq{}s motion in the radar signal processing.

We say two components $x_{k-1}(t)$ and $x_k(t)$ in \eqref{AHM} are overlapping in time-frequency plane when the IF curves of them are crossing at $t=t_{0}$, i.e., $\phi_{k-1}^{'}(t_{0})=\phi_{k}^{'}(t_{0})$.
In this paper we consider multicomponent signals of the form \eqref{AHM} satisfying
\begin{equation}\label{def_sep_cond_cros}
%\min_{0\leq \ell \neq k\leq K}\sqrt{(\phi'_{k}(t)-\phi'_{\ell}(t))^2+\epsilon|\phi''_{k}(t)-\phi''_{\ell}(t)|} =\eta >0.
|\phi'_{k}(t)-\phi'_{\ell}(t)|+\rho |\phi''_{k}(t)-\phi''_{\ell}(t)| \ge 2 \gt,
\end{equation}
where $\rho\ge 0$ is the regularization parameter and $\gt>0$ is called the separation resolution.
When $\rho=0$, \eqref{def_sep_cond_cros} is reduced to the well-separated condition \eqref{def_sep_cond} required for SST and SSO. Condition
\eqref{def_sep_cond_cros} allows the IFs of some components $x_k(t)$ to cross as long as their chirp rates
$\phi_k''(t)$ are different near the time $t_0$ where IFs crossing occurs.

Recently, many studies have been devoted to estimations of crossing instantaneous frequencies. An algorithm based on linear least square fitting technique was introduced \cite{add1}. 
In \cite{add2}, the authors employed a three-dimensional parameter space (time, frequency, chirp rate) to estimate instantaneous frequencies and chirp rates jointly of signal, which can make the crossing of IFs appear as separated in the time-frequency-chirprate domain. Furthermore, \cite{add2} proposed a three-dimensional ridge detection algorithm.
\cite{add3} used time-frequency analysis and Radon transform  for the unsupervised separation of frequency modulated signals in noisy environment. The chirp rates were also introduced in \cite{add5}, where estimations of IFs were turned into the resolution of a two-dimensional linear system. Effective quasi maximum likelihood–random samples consensus algorithm was also proposed for IF estimation of overlapping signals in the time-frequency (TF) plane \cite{add6}.

Moreover, \cite{Stankovic_compressive_sense_2013} developed a compressive sensing approach to recover stationary narrowband signals contaminated by strong nonstationary signals. An algorithm for estimations of IFs of signals with intersecting signatures in the time-frequency domain is proposed in \cite{add4}. Using the estimated IF, the corresponding component is removed from the mixture signal based on de-chirping method \cite{add4}. \cite{add7} combined parameterized de-chirping and band-pass filter to obtain components of multi-component signal, which avoided dealing with time-frequency representation of the signal and worked well under heavy noise. \cite{IEEE_Sensors_2017} introduced a novel nonparametric algorithm called ridge path regrouping to extract the IFs of the overlapped components and recover the corresponding component by using the intrinsic chirp component decomposition method. The time-frequency distribution (with a cross-term) \cite{TF_distribution_2004}
%and the fractional Fourier transform \cite{IEEE_Ultrason_2010}
was used  to estimate IFs.
%In addition, many efforts have also been made to separate blind sources where these modes are overlapping in the frequency domain. 
In \cite{Over1}, the authors developed a general model to characterize multi-component chirp signals, where IFs and instantaneous amplitudes of the intrinsic chirp components were modeled as Fourier series. 
%The estimation of IF  was obtained using the framework of the general parameterized time-frequency transform and then the signal can be reconstructed by solving a linear system. 
Similarly, a multivariate intrinsic chirp mode was defined in \cite{Over2}. Then IFs can be estimated using the framework of the general parameterized time-frequency transform and the corresponding multivariate intrinsic chirp modes were reconstructed by solving multivariate linear equations through an extended least square method \cite{Over2}. Variational nonlinear chirp mode decomposition was employed to separate close or even crossed modes in \cite{Over3}. In \cite{Over4}, the blind source was decomposed into a series of nonlinear chirp components and the reconstructed nonlinear chirp components can be clustered into corresponding intrinsic chirp sources. %Finally, the sources were reconstructed based on the clustering consequence. To deal with close or overlapped signals with high precision, 
The frequency-domain intrinsic component decomposition method was proposed to characterize nonlinear and non-monotonic nonlinear group delays in frequency domain by using different kernel functions and separate and reconstruct each mode simultaneously as a time-frequency filter \cite{Over5}. 
An adaptive chirp mode pursuit was proposed to capture signal modes one by one in a recursive framework \cite{Over6}. In \cite{Over7}, the authors transformed a two-dimensional sinusoidal pattern into a single point in a two-dimensional plane and reconstructed signals from a reduced set of observations (back-projections).

In this paper, using the idea of SSO, we propose a chirplet transform-based signal separation scheme (CT3S) to retrieve modes of multicomponent signals with crossover IFs.  
Our algorithm addresses the inverse transform and aims to retrieve modes of multicomponent signals with fast varying IFs, even crossing IFs, while the aforementioned papers use  the chirplet transform to obtain IFs.  Most importantly, we provide a mathematically rigorous theorem which guarantees  the recovery of components of multicomponent signals satisfying \eqref{def_sep_cond_cros} with CT3S. To our best knowledge, there is no existing literature that establishes mathematical theorems to guarantee the recovery of the components overlapping in the time-frequency plane.
Furthermore, we use a matched-filter along the specific time-frequency lines respect to the chirp rate to make different %, including crossover
components be further separated and more concentrated in the three dimensional space of the chirplet transform.  Based on the approximation of source signals with linear chirps at any local time, 
we propose an innovative signal reconstruction algorithm, called group 
 filter-matched CT3S (GFCT3S), which takes a group of components into consideration simultaneously.
GFCT3S  is more suitable for signals with crossing IFs.
When the IFs curves of different components are not crossover, but fast-varying and close to each other, our reconstruction algorithm will decrease the recovery errors significantly.

The remainder of this paper is organized as follows.
In Section 2, we first review the signal separation methods by SST and SSO. After that we state the problem of recovering sub-signals with crossing IFs.
In Section 3, we  state our theorem which guarantees  the recovery of components for multicomponent signal satisfying \eqref{def_sep_cond_cros}  with CT3S.   We also introduce  
the filtered chirplet transform and propose GFCT3S for separating multicomponent signals with fast-varying and crossing IFs. We present the numerical experiments in Section 4.
A conclusion is then presented in Section 5. 

\section {Problem statement}
The (modified) STFT of signal $x(t) \in L_2(\RR)$ with a window function $g(t) \in L_2(\RR)$ is defined by,
\begin{equation}
\label{def_STFT}
V_x(t, \eta)=\int_{\RR} x(\tau) g(\tau-t) e^{-i2\pi \eta(\tau-t)}d\tau,
\end{equation}
If $g(0)\not=0$, then the original signal $x(t)$ can be recovered back from its STFT:
\begin{equation}
\label{rec_x_def}
x(t)=\frac 1{g(0)}\int_\RR V_x(t, \eta) d\eta.
%\frac 1{\left\|h\right\|_2^2}  \int_{\RR}\int_{\RR} V_x(t, \eta) \overline {h(\tau-t)} e^{-i2\pi \eta(\tau-t)}d\tau d\eta.
\end{equation}
For multicomponent signal $x(t)$ in \eqref{AHM} satisfying the separation condition \eqref{def_sep_cond}, the sub-signal $x_k(t)$ can be recovered by
\begin{equation}
\label{rec_real_x_def}
x_k(t)\approx \frac 1{g(0)}  \int_{|\eta-\phi_k'(t)|<\Gamma_1} V_x(t, \eta) d\eta,
\end{equation}
for some $\Gamma_1>0$.

To enhance the time-frequency resolution and concentration, the idea of SST is to reassign the frequency variable.
As in \cite{Thakur_Wu11}, denote
\begin{equation}
\label{def_phase}
\go_x(t, \eta)=\frac{\frac {\partial}{\partial t} V_x(t, \eta)}{i 2\pi  V_x(t, \eta)}.
\end{equation}
The quantity $\go_x(t, \eta)$ is called the ``phase transformation"  \cite{Daub_Lu_Wu11}.
The STFT-based SST is to reassign the frequency variable $\eta$ by transforming the STFT $V_x(t, \eta)$ of $x(t)$ to a quantity, denoted by $R_x(t, \eta)$, on the time-frequency plane:
\begin{equation}
\label{def_FSST_simple}
R_x(t, \eta)=\int_{\{\xi: V_x(t, \xi)\not=0\}} V_x(t, \xi) \delta\left(\go_x(t, \xi)-\eta\right) d \xi.
\end{equation}
\begin{figure}[H]
	\centering
	\hspace{-0.5cm}
	%\begin{tabular}{cccc}
	\begin{tabular}{c@{\hskip -0.06cm}c@{\hskip -0.06cm}c@{\hskip -0.06cm}c}
		\resizebox{1.6in}{1.2in}{\includegraphics{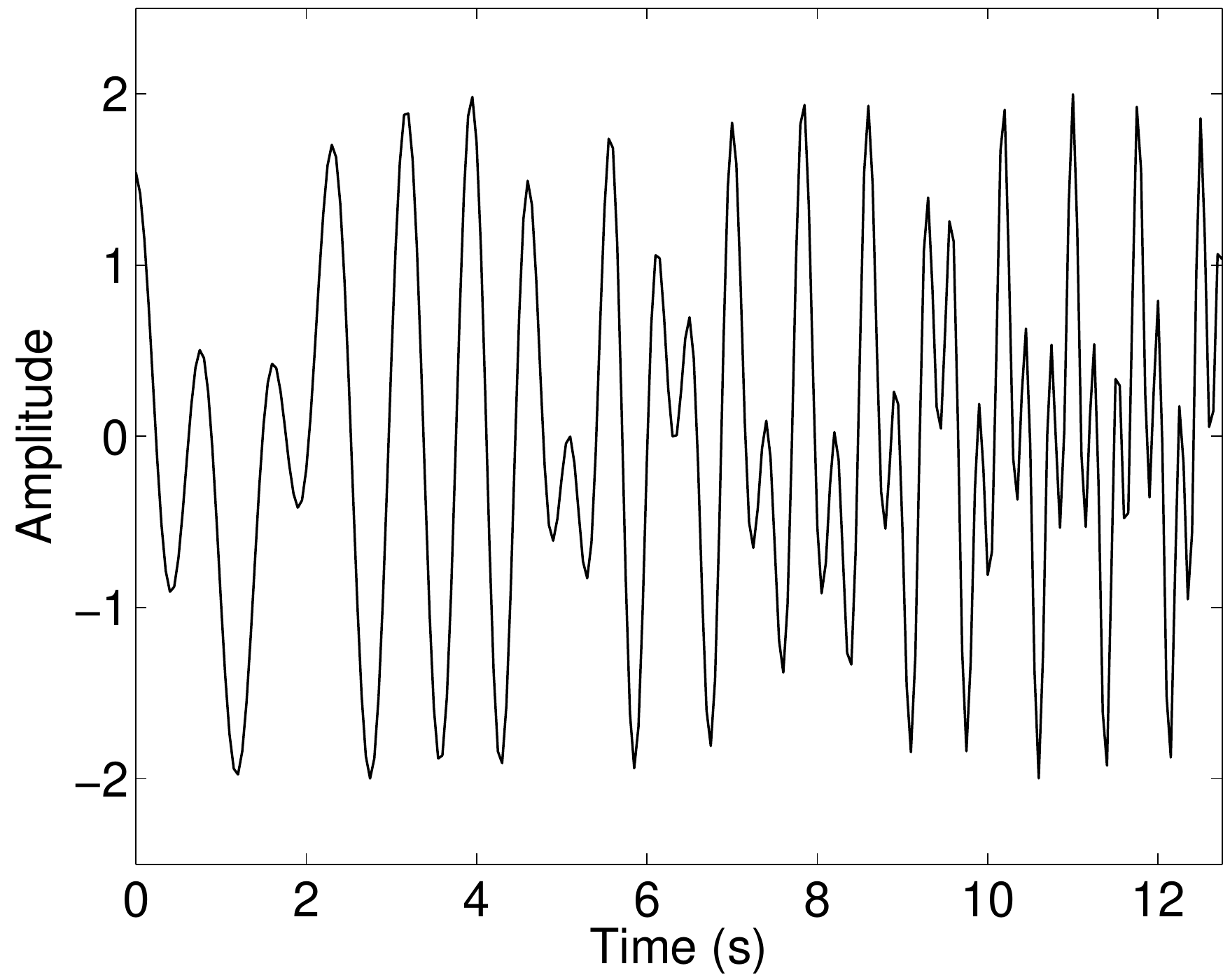}} \quad &
		\resizebox{1.6in}{1.2in}{\includegraphics{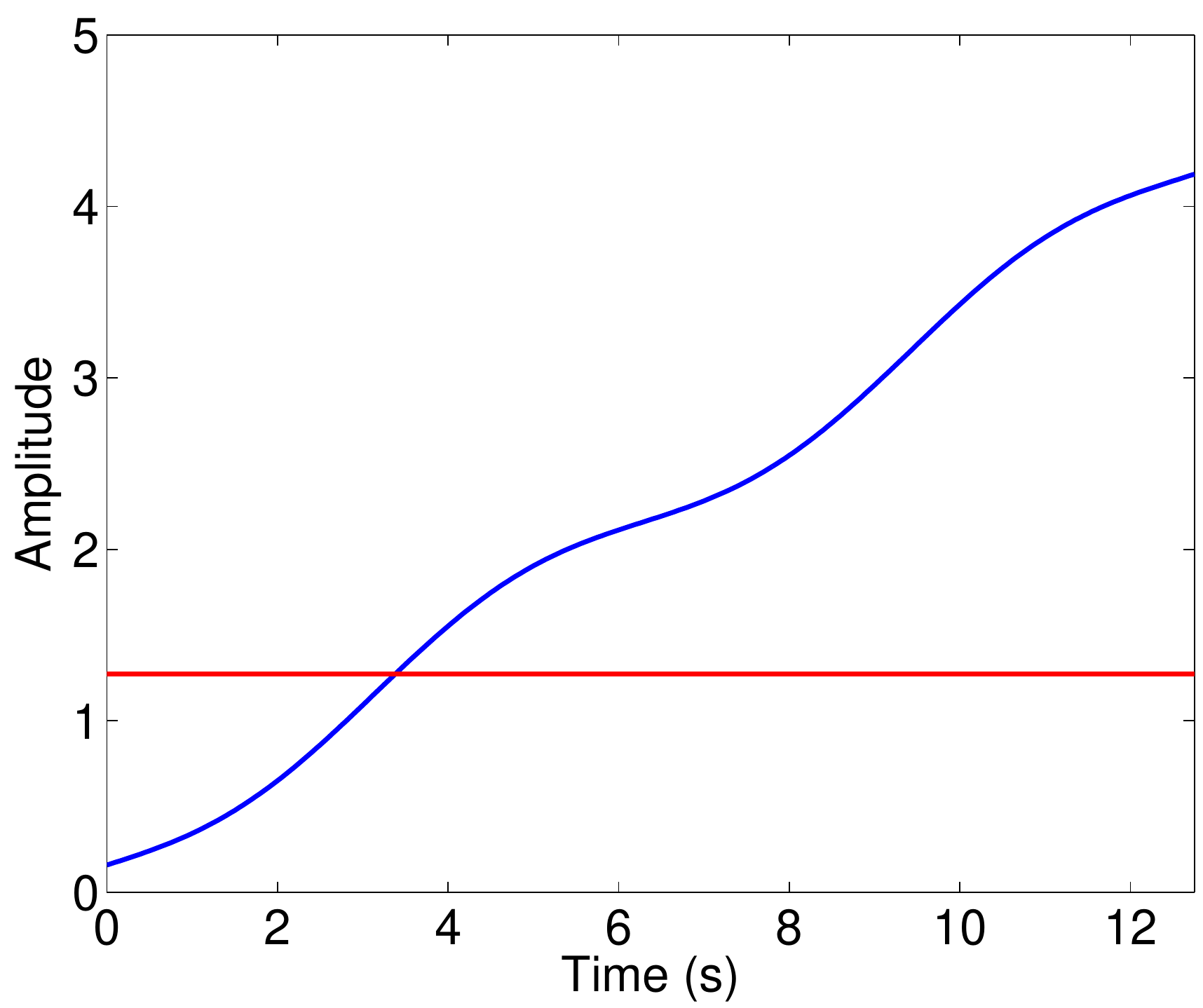}} \\
		\resizebox{1.6in}{1.2in}{\includegraphics{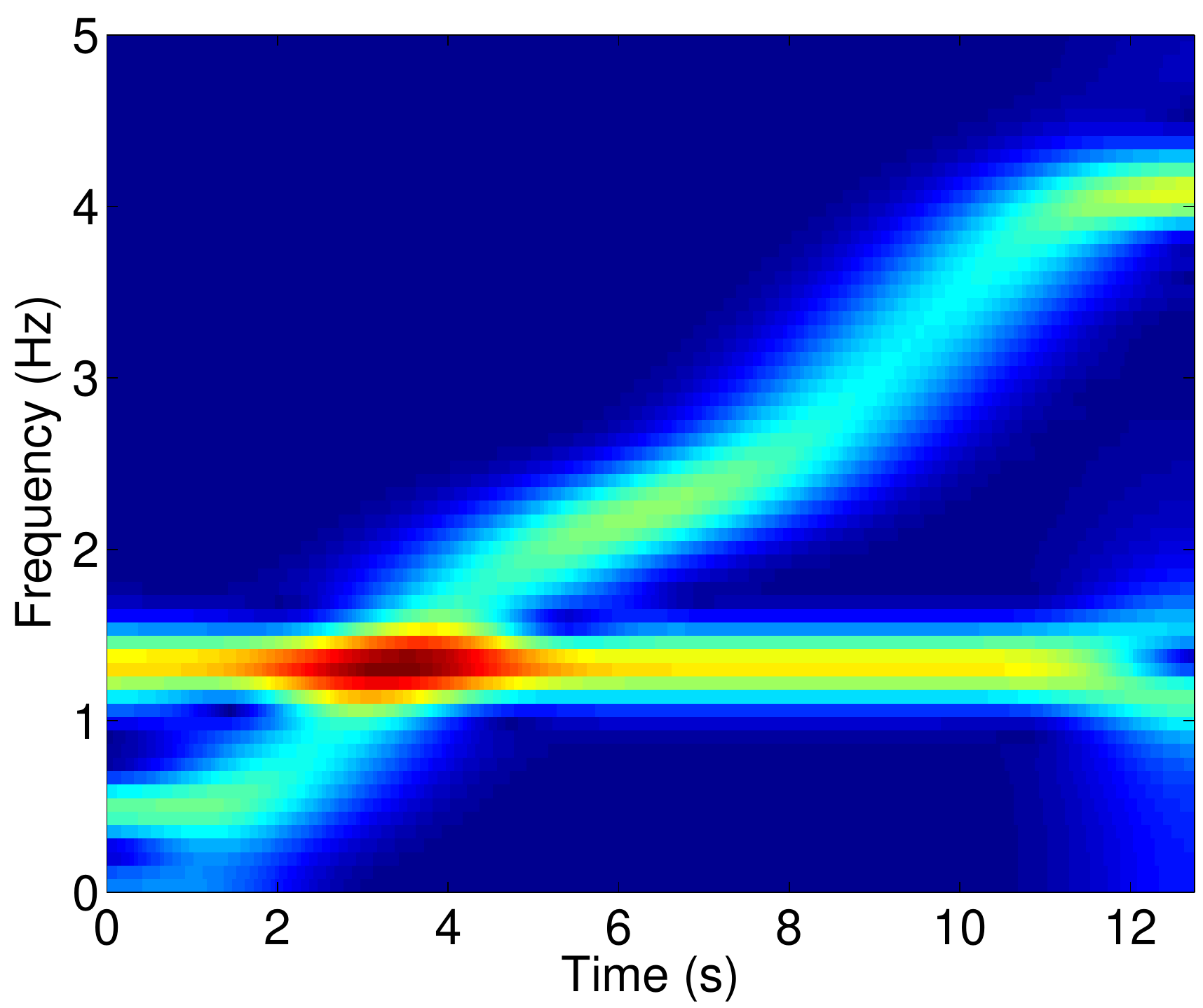}} \quad &
		\resizebox{1.6in}{1.2in}{\includegraphics{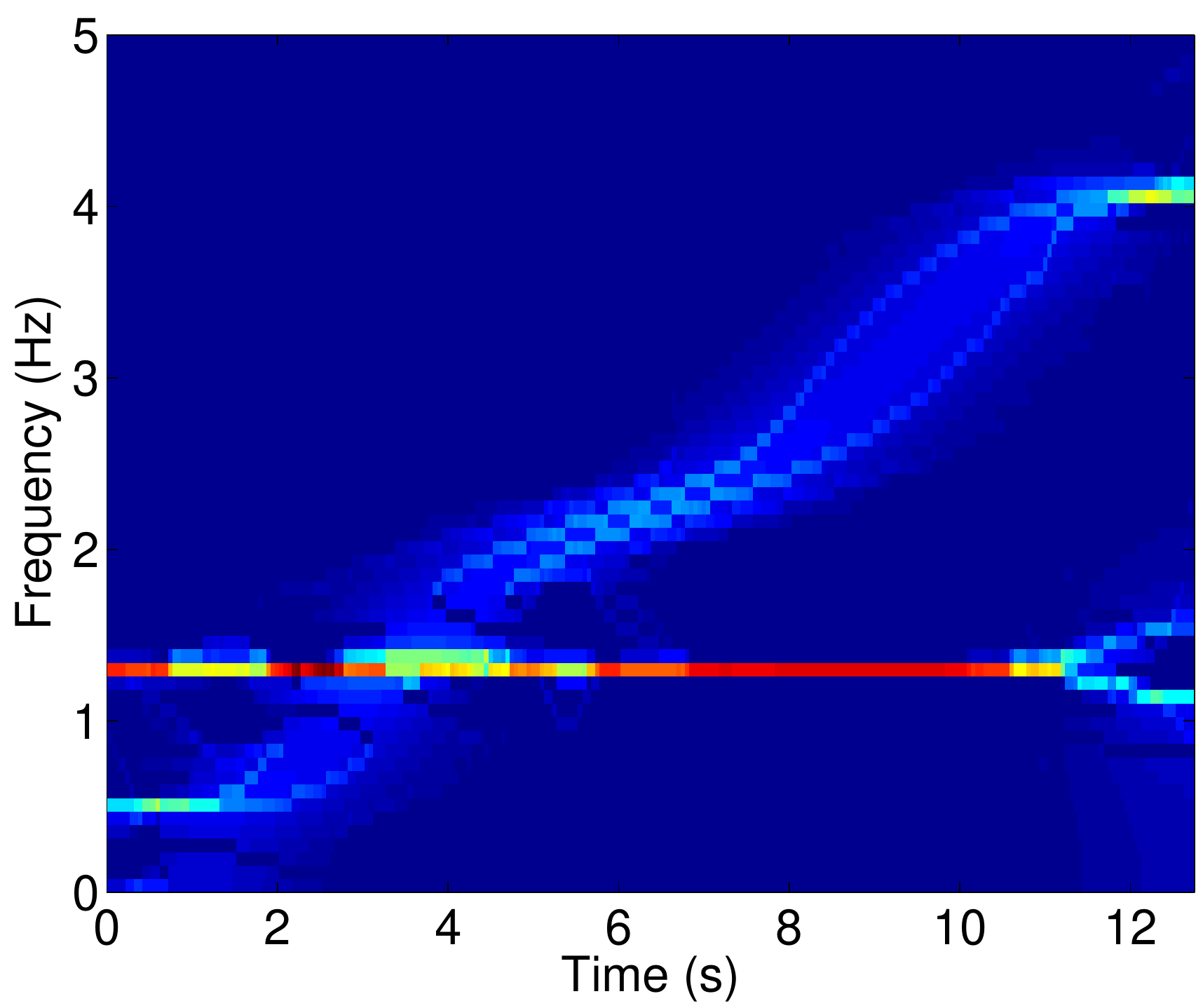}} \quad &
		\resizebox{1.6in}{1.2in}{\includegraphics{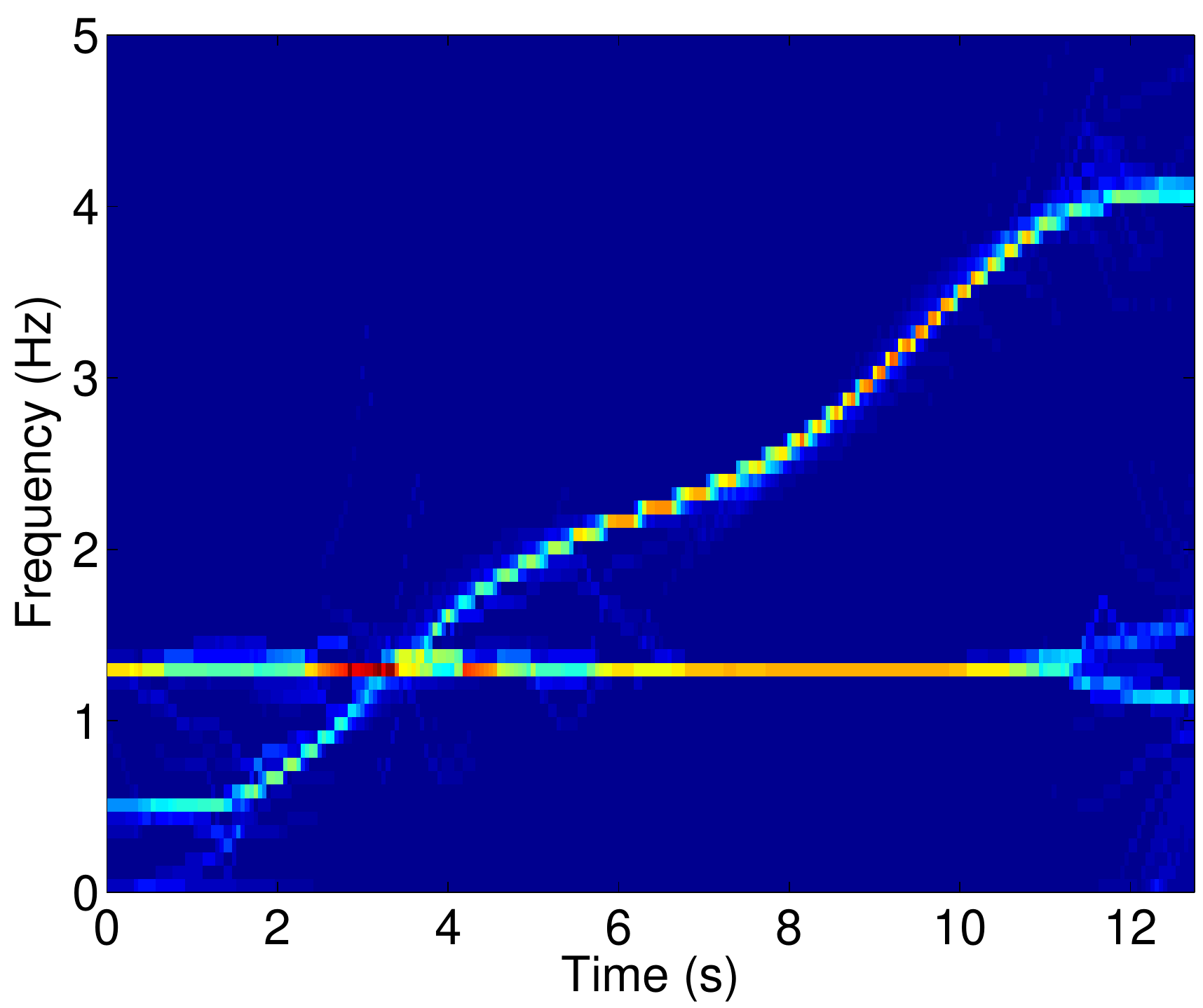}} \quad &
		\resizebox{1.6in}{1.2in}{\includegraphics{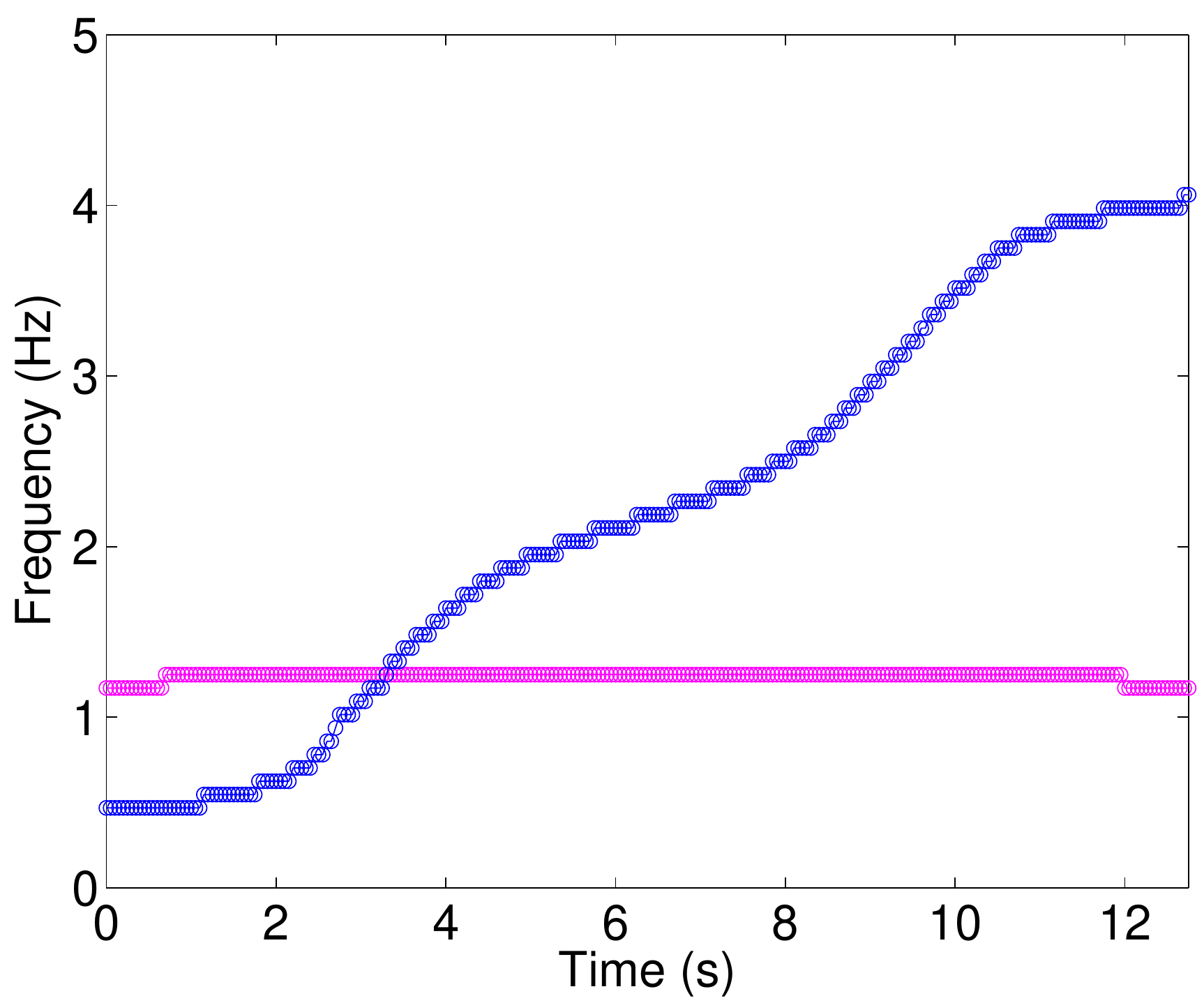}}\\
		\resizebox{1.6in}{1.2in}{\includegraphics{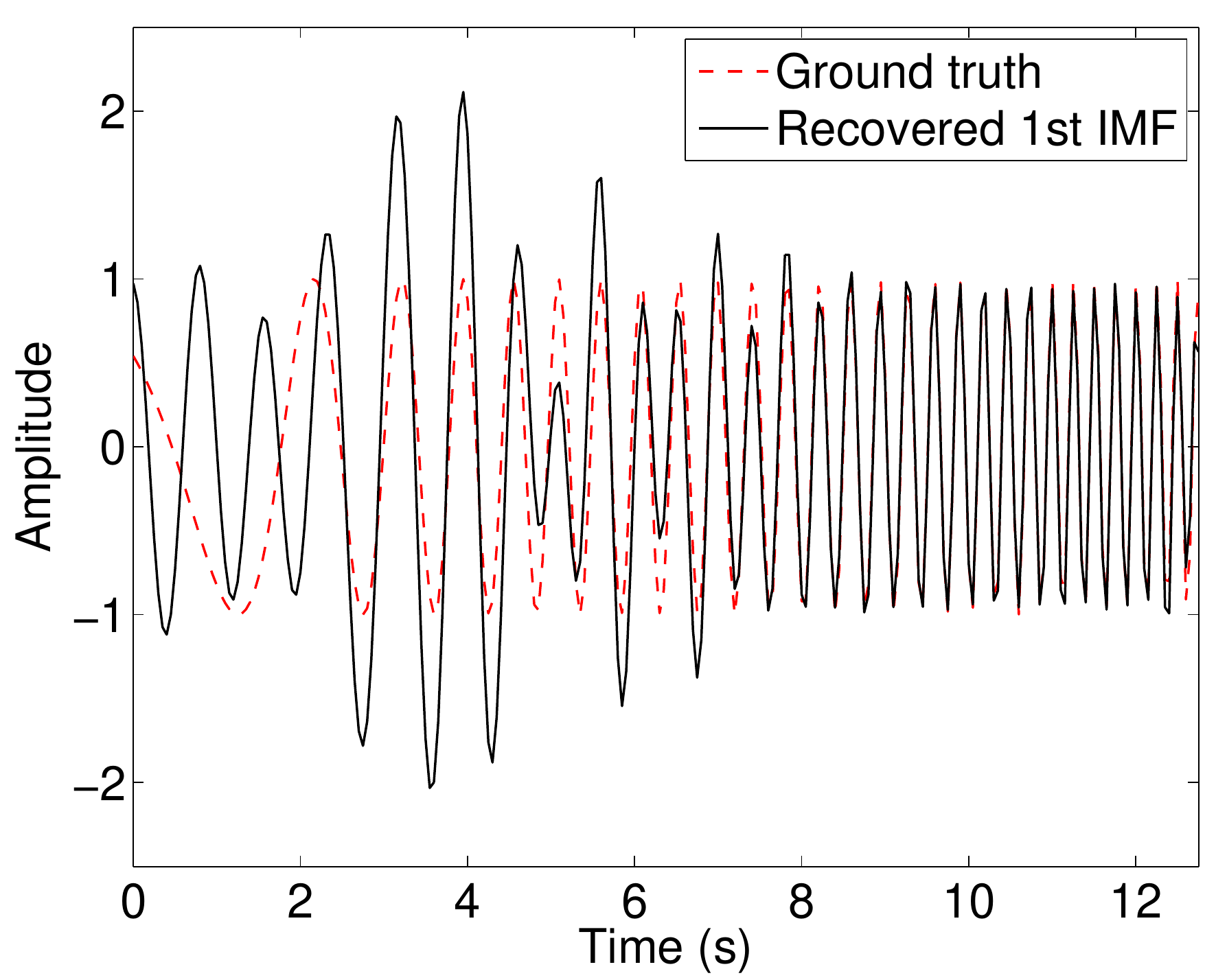}} \quad &
		\resizebox{1.6in}{1.2in}{\includegraphics{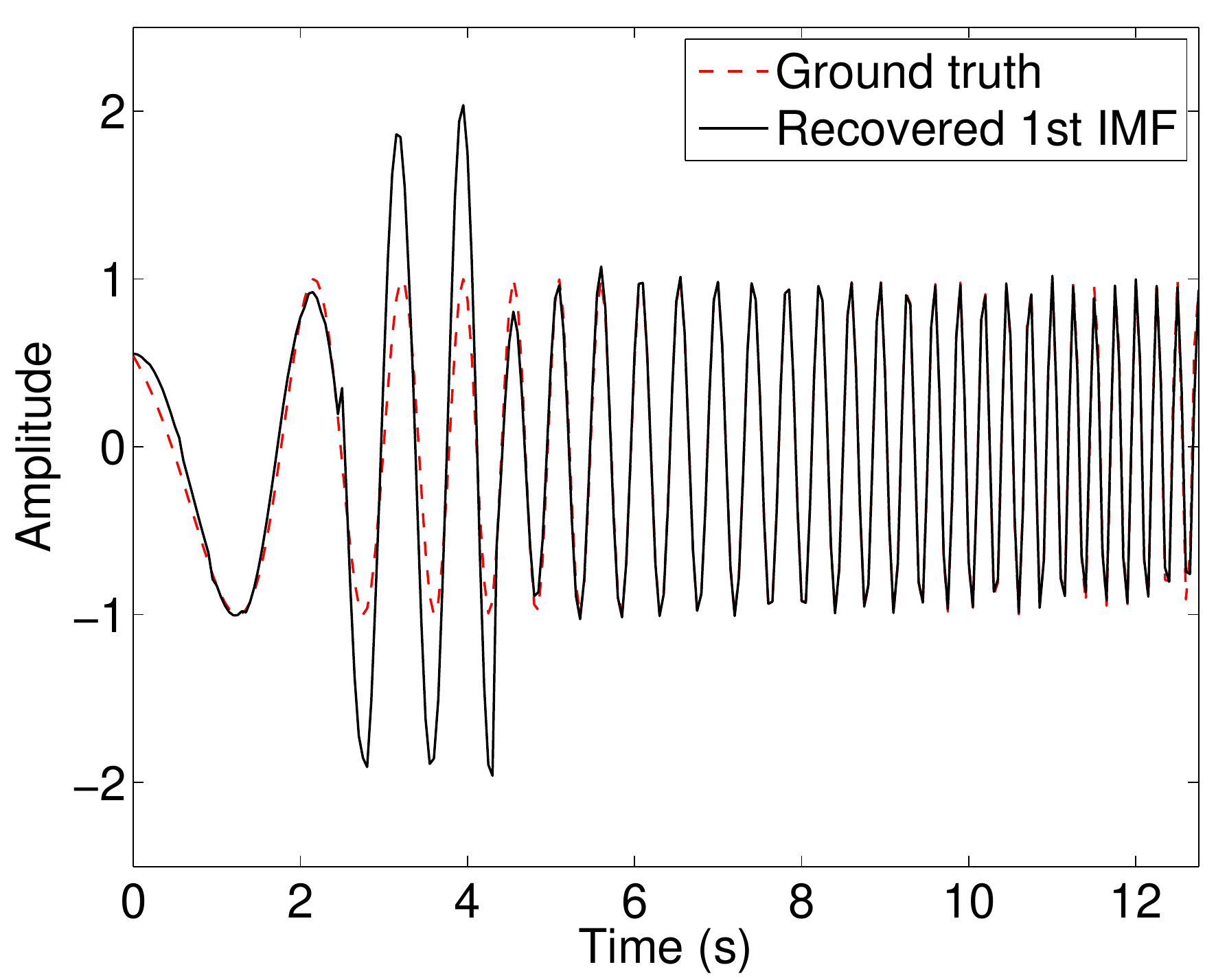}} \quad &
		\resizebox{1.6in}{1.2in}{\includegraphics{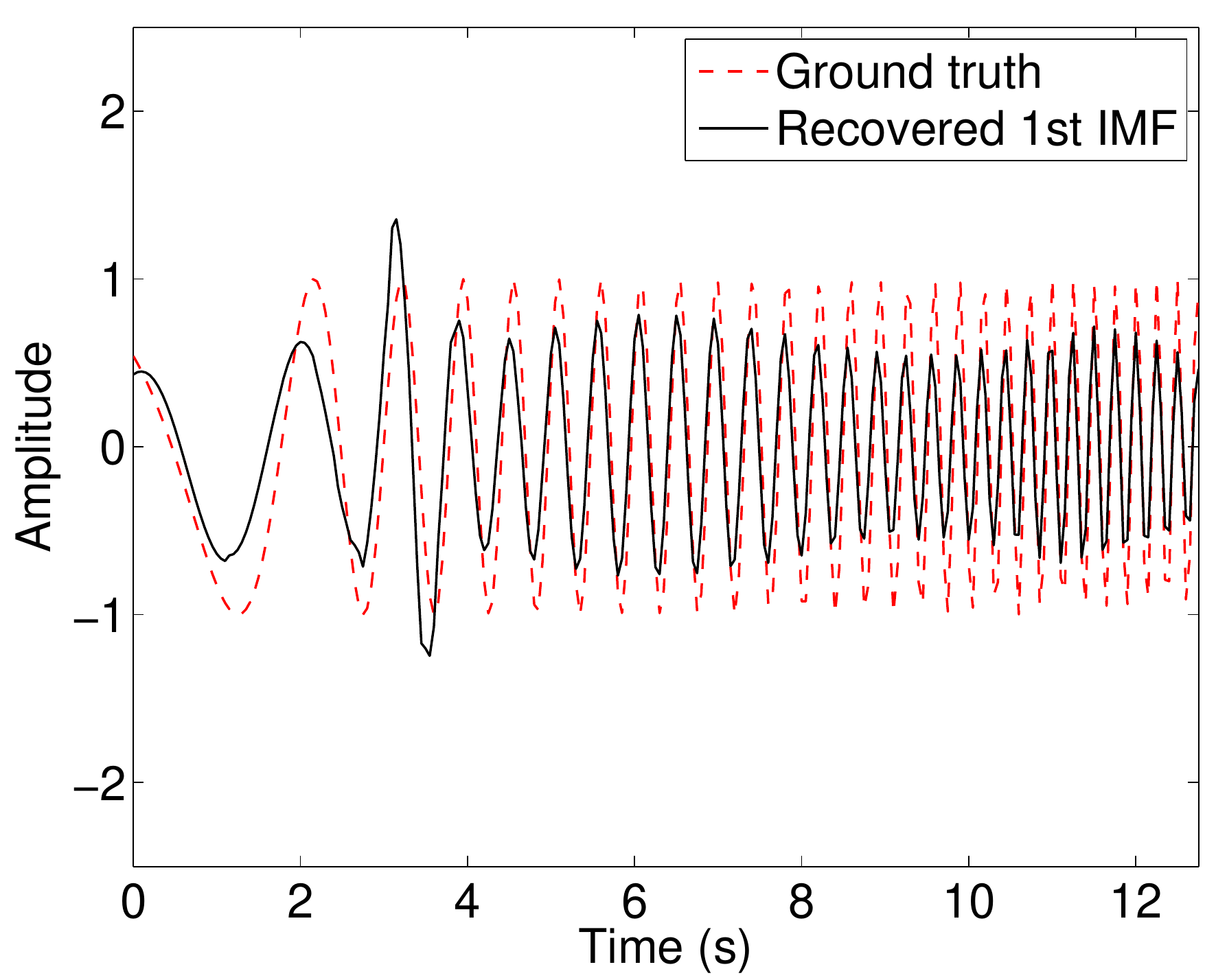}} \quad &
		\resizebox{1.6in}{1.2in}{\includegraphics{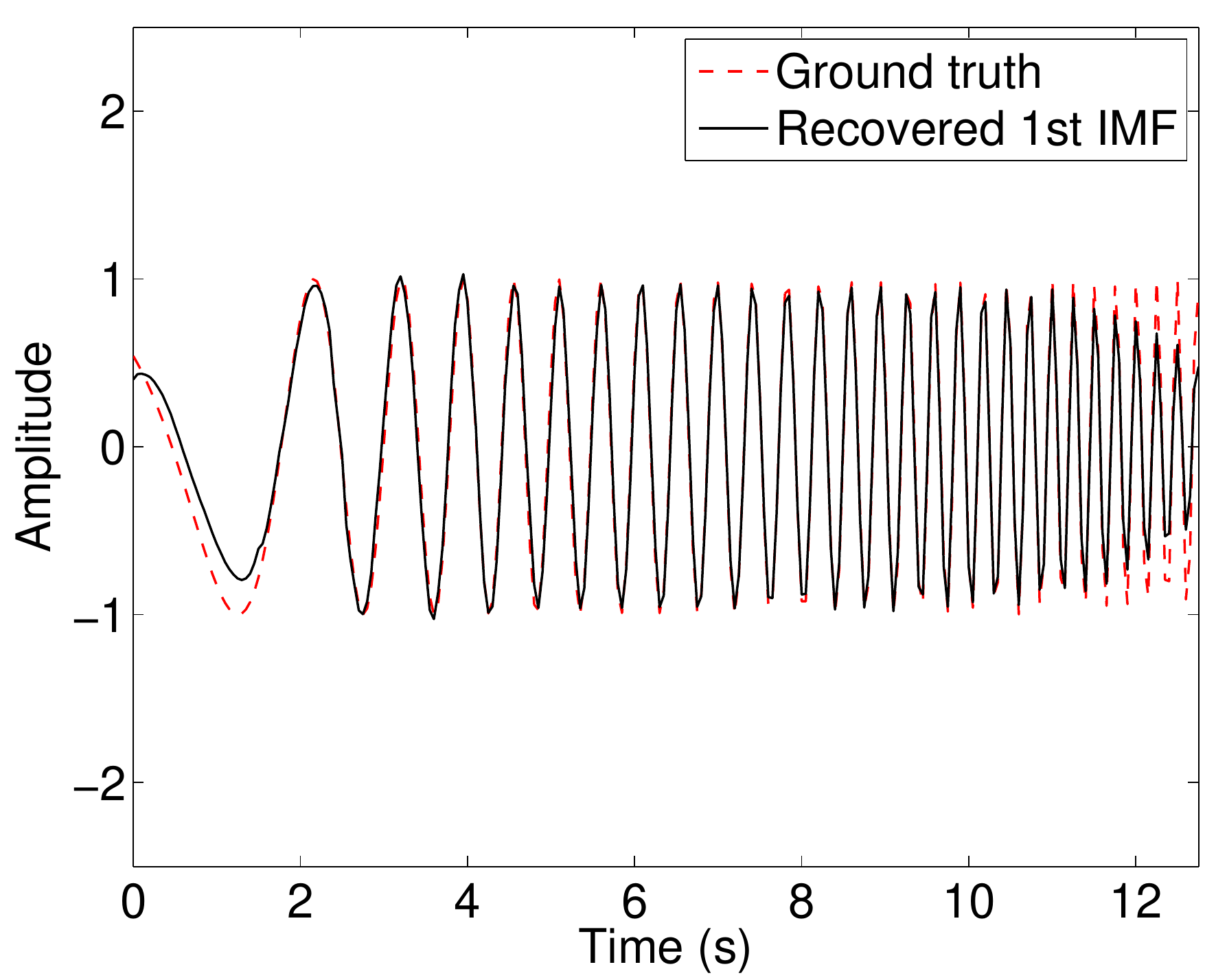}}\\
		\resizebox{1.6in}{1.2in}{\includegraphics{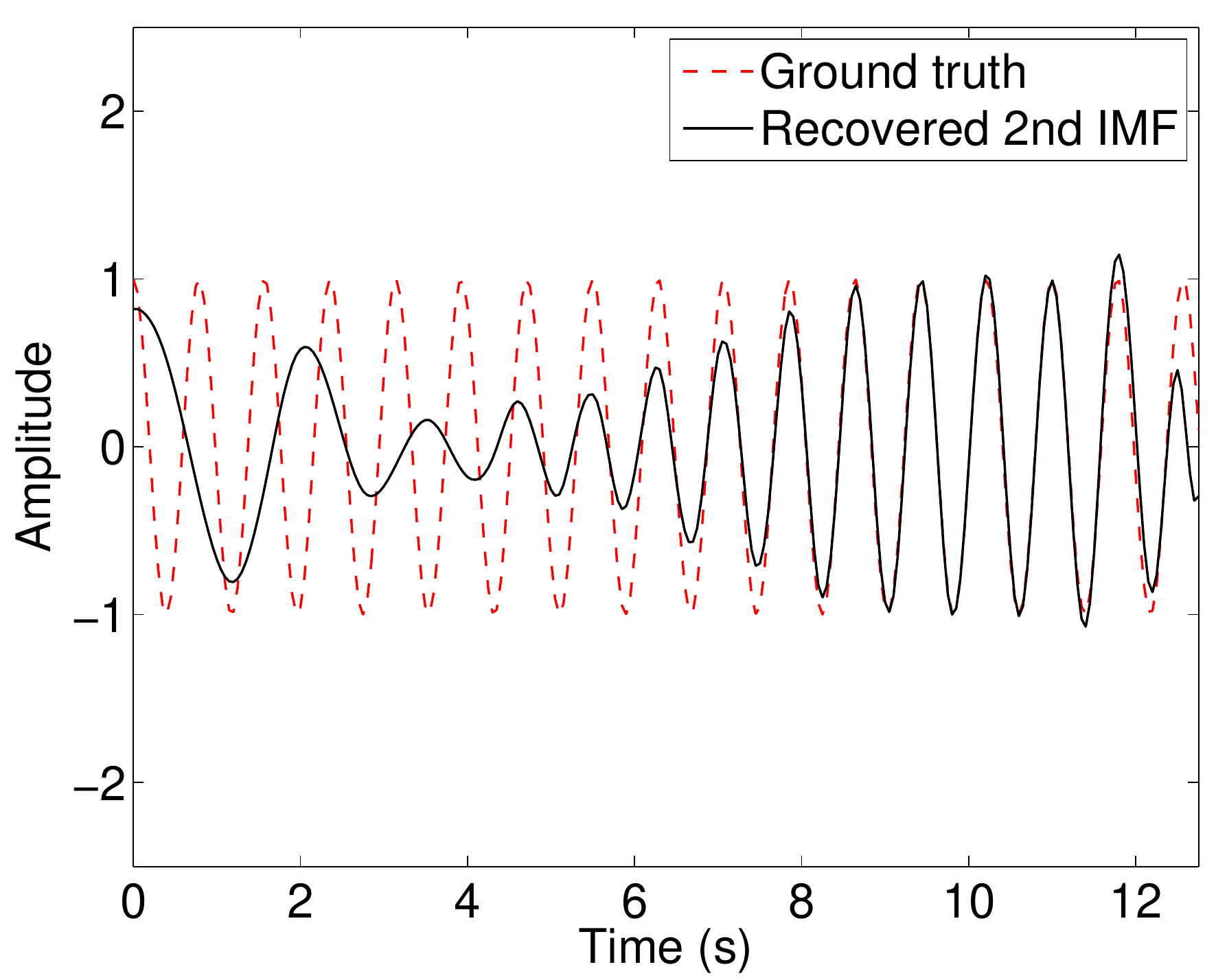}} \quad &
		\resizebox{1.6in}{1.2in}{\includegraphics{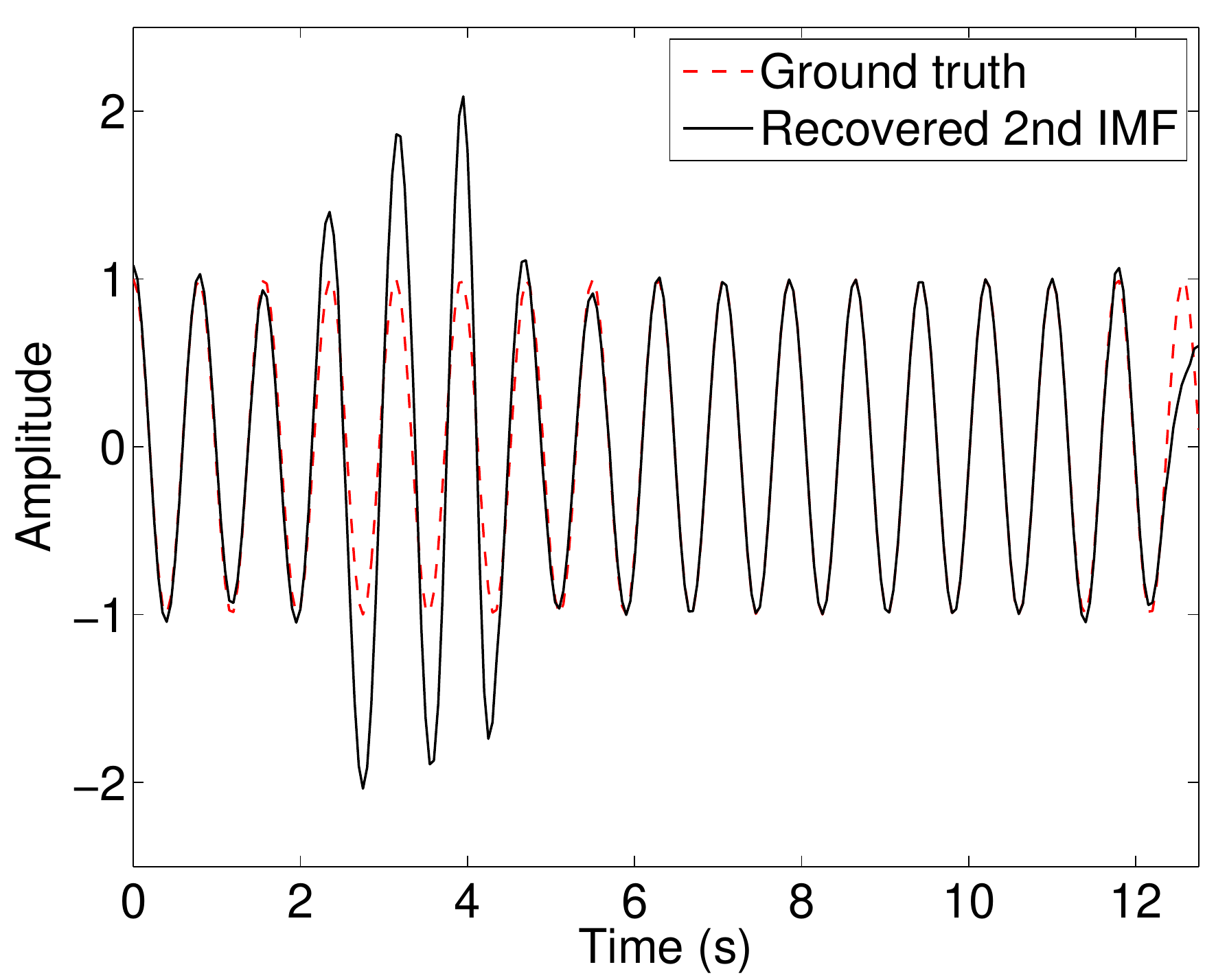}} \quad &
		\resizebox{1.6in}{1.2in}{\includegraphics{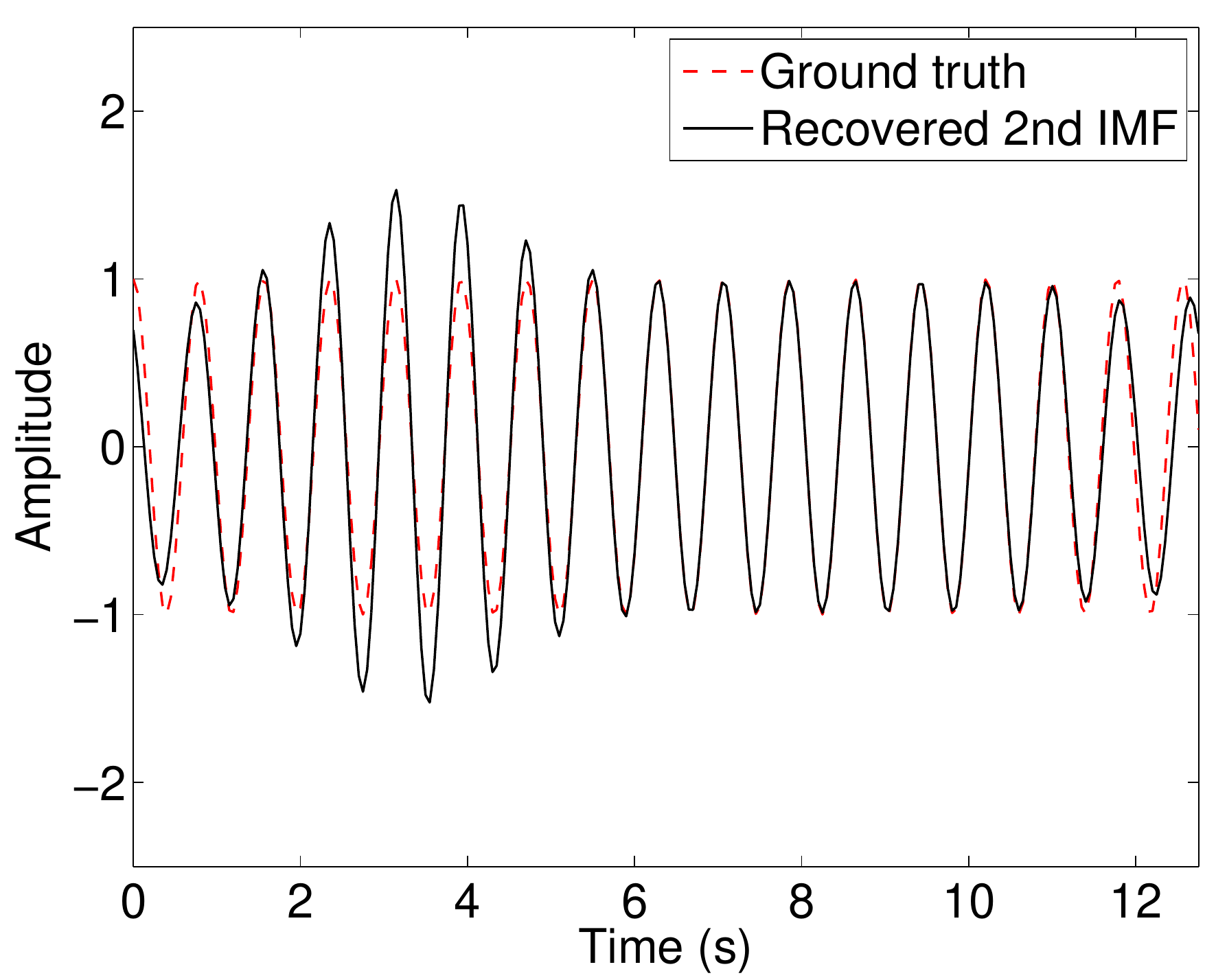}} \quad &
		\resizebox{1.6in}{1.2in}{\includegraphics{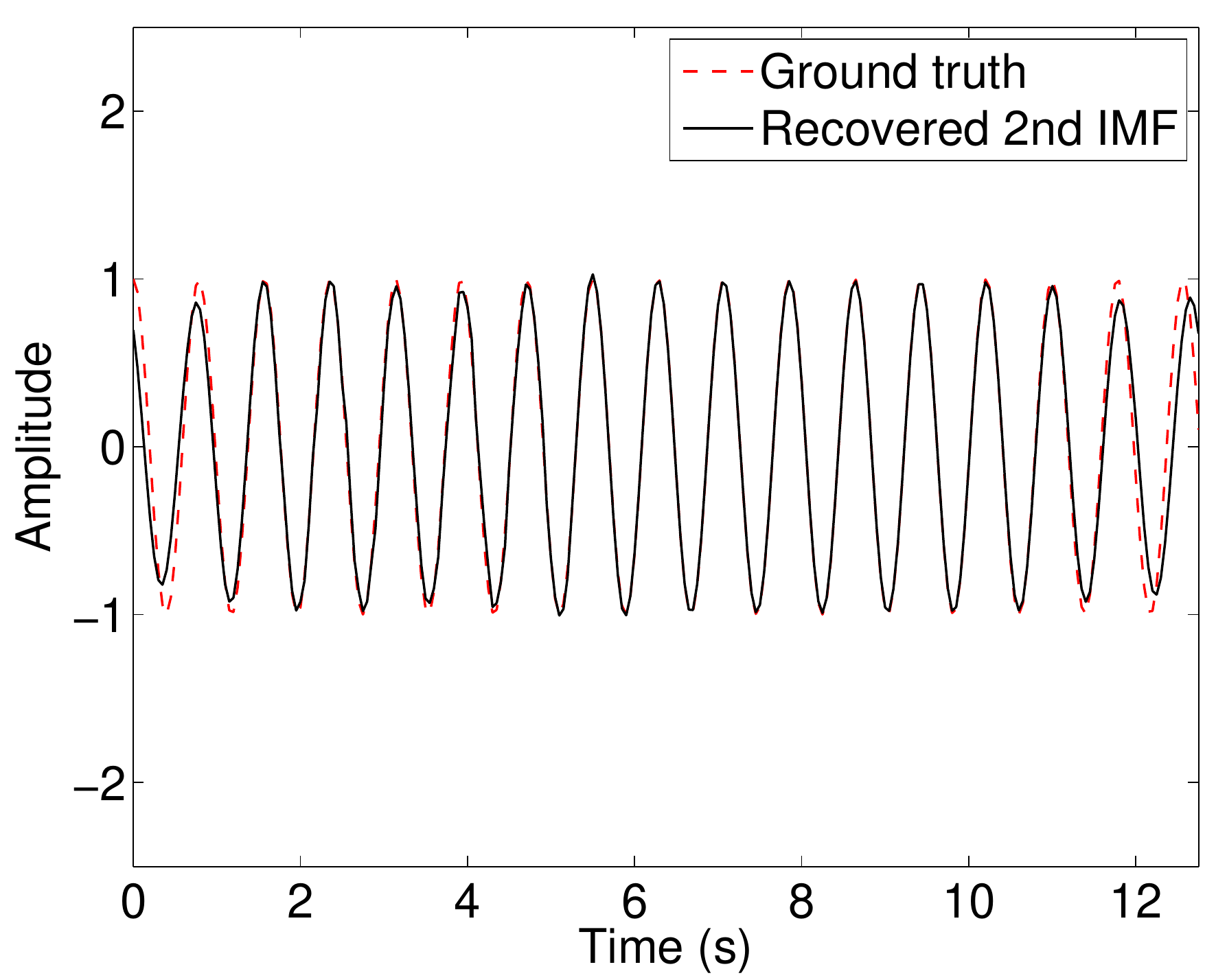}}
	\end{tabular}
	\caption{\small Results of component recovery  of two-component signal $f(t)$ in \eqref{two_component}. Source signal $f(t)$: Waveform and IFs (Top row, from left to right);
		Time-frequency diagrams: STFT, SST, 2nd-order SST and estimated IFs by our proposed filter-matched CT  (Second row, from left to right);
		Recovery results of $f_1(t)$ by different methods: EMD, 2nd-order SST, SSO and proposed method GFCT3S (Third row, from left to right);
		Recovery results of $f_2(t)$ by different methods: EMD, 2nd-order SST, SSO and proposed method GFCT3S (Bottom row, from left to right).}
	\label{Fig:two-component}
\end{figure}
\noindent One also has a reconstruction formula of $x_k(t)$ similar to \eqref{rec_real_x_def} with $V_x(t, \xi)$ replaced by $R_x(t, \xi)$. 
%Moreover, the 2nd-order \cite{MOM15} and high-order \cite{Pham17} SSTs were introduced based on higher order phase transformations.

Observe that signal reconstructions with the STFT and SST depend on the IF estimation of $\phi_k'(t)$ and a given threshold $\Gamma_1$, hence it is indirect.
In contrast, signal separation operator (SSO) \cite{Chui_Mhaskar15} extracts signal components via local frequencies directly.
The SSO  $T_{a,\delta}$, which is applied to signals $x$ in \eqref{AHM}, is defined by
	\begin{equation}
	\label{def_SSO}
	\left(T_{a,\delta}x\right)(t,\eta) = \frac 1 {\hbar_a} \sum \limits_{n \in \ZZ} x(t-n\delta) h \left(\frac n a\right) e^{i2\pi n\eta},
	\end{equation}
where $h$ is a compactly supported window function, $\eta \in [0,1]$ and $\delta,a>0$ are parameters, with $a$ so chosen that
	\begin{equation}
	\label{h_area}
	\hbar_a = \sum \limits_{n \in \ZZ} h \left(\frac n a\right) > 0.
	\end{equation}	
For a multicomponent signal $x(t)$ defined in \eqref{AHM}, satisfying certain conditions including
the well-separation condition \eqref{def_sep_cond}, the set
 \begin{equation*}
\left\{\eta\in[0,1]: |\left(T_{a,\delta}x\right)(t,\eta)|>\Gamma_2 \right\}
 \end{equation*}
can be expressed as a disjoint union of exactly $K$ non-empty sets $\Theta_\ell, \ell=1,2,...,K$, corresponding to the $K$ components of $x(t)$.	
The sub-signal $x_k(t)$ can be reconstructed by,
\begin{equation}
\label{SSO_recon}
\wh x_k(t) = 2 \Re e \left \{ \left(T_{a,\delta}x\right)(t,\wh \eta_k) \right\},
\end{equation}
where
\begin{equation}
\label{find_ridge}
\wh \eta_k(t) = \arg \mathop{\max}\limits_{\eta \in \Theta_k} |\left(T_{a,\delta}x\right)(t,\eta)|.
\end{equation}
As mentioned in Section 1, to recover components $x_k(t)$ with
SST or SSO,  it is required all IFs of different components be separated from each other, namely they be far away from each other and non-crossing as shown in \eqref{def_sep_cond}. In particular, there is no mathematical theorem %theoretical analysis 
to guarantee the recovery of the waveforms of components when their IFs are crossover with only one observation $x(t)$ available.
This paper is to provide a method to retrieve modes of such nonstationary multicomponent signals  and establish 
retrieve error bounds. Next let us consider an example to show the performances of EMD, SST, SSO and our proposed method GFCT3S in retrieving the modes of a signal with crossover IFs.

Let $f(t)$ be the two-component signal introduced in \cite{Daub_Lu_Wu11}:
\begin{equation}
\label{two_component}
f(t) =f_1(t) +f_2(t) = \cos\left(t^2+t+\cos(t)\right) + \cos(8t).
\end{equation}
Here we let the sampling rate $F_s = 20$Hz and we only analyze the truncation signal on $t \in [0,256/F_s]$, with $256$ discrete sampling points. The instantaneous frequencies of  $f_1(t)$ and $f_2(t)$ are $\phi_1'(t) = (2t+1-\sin(t))/(2\pi)$ and $\phi_2'(t) = 4/\pi$, respectively.

Figure \ref{Fig:two-component} shows the recovery results of $f_1(t)$ and $f_2(t)$. Observe that compared to the STFT and the STFT-based SST, the 2nd-order SST of $f(t)$ represents this two-component signal with crossing IF curves much sharper and clearer. However, the existing methods including EMD \cite{Huang98}, SST \cite{Thakur_Wu11}, the 2nd-order SST  \cite{MOM15} and SSO  \cite{Chui_Mhaskar15} are unable to recover the waveforms $f_1(t)$ and $f_2(t)$ accurately, see the recovered $f_1$ and $f_2$ by these methods in Figure \ref{Fig:two-component}. Our proposed GFCT3S in Algorithm 1 provided in Section 3.2 can recover the two components accurately as shown in the 4th panels (from the left) in Row 3 and Row 4 respectively in Figure \ref{Fig:two-component}.
Note that EMD, SST, 2nd-order SST and SSO all result in big recovery errors for either $f_1$ or $f_2$ around $t_0=3.38$ where IFs crossing occurs, while our method produces very small errors near $t_0$.
 The boundary effect is unavoidable for all methods since we only use the truncation signal for $ 0\le t \le 12.8$ and the boundary extension has not be considered in this example. In addition, in this example we simply use the same Gaussian window with constant variance $\sigma = 1.6$ for STFT, SST, the 2nd-order SST, SSO and GFCT3S .
%\clearpage

%%%%%%%%%%%%%%%%%%%the beginning of Figure 2 %%%%%%%%%%%%%%%

%%%%%%%%%%%the end of Figure 2 %%%%%%%%%%%%%%%%%%%%%

\section{Chirplet transform-based signal separation scheme}
% {Mode retrieval with localized polynomial Fourier transform}
In this section we propose a chirplet transform-based signal separation scheme (CT3S) to retrieve modes. We provide the main theorem on component recovery analysis.
In addition, we introduce filtered chirplet transform (CT) to make IFs crossover components further separated and more concentrated in the three-dimensional space of CT. Finally, 
to improve the performance of  CT3S in mode retrieval, we present an algorithm, called group filter-matched CT3S or GFTC3S for short,  based on filtered CT and the approximation of source signals with linear chirps at any local time.

\subsection{Main results}
%In this subsection, first we introduce LPFT$_{2}$. After that we define the class of IF crossover multicomponent signals which can be separated by LPFT$_{2}$ and an admissible window function used in for signal separation. The main result, Theorem \ref{theo:LQFT}, will be followed then.

The chirplet transform (CT) applied to a signal $x(t)$ is defined by  (see \cite{Mann95}-\cite{Czarnecki18}) 
\begin{eqnarray}
 \label{def_MSSO}
% \begin{array}{ll}
 \fs_x(t,\eta,\lambda)\hskip -0.6cm && = \int_{\RR} x(\tau) \frac 1 {\gs} g\big(\frac {\tau-t} {\gs}\big) e^{-i2\pi\eta(\tau-t) -i\pi \lambda (\tau-t)^2}
 d\tau\\
\nonumber && = \int_{\RR} x(t+\tau) 
\frac 1 {\gs} g\big(\frac \tau {\gs}\big) e^{-i2\pi\eta\tau -i\pi \lambda \tau^2}
d\tau,
% \end{array}
 \end{eqnarray}
%	where	
%	\begin{equation}
%	\label{new_kernel}
%	\mathcal{K}_\gs(\tau,\eta,\lambda)= \frac 1 {\gs} g\big(\frac \tau {\gs}\big) e^{-i2\pi\eta\tau -i\pi \lambda \tau^2},
%	\end{equation}
where $g(t)$ is a window function and $\gs>0$. $\fs_x$ is also called the localized polynomial Fourier transform (of order 2),  see \cite{Katkovnik95, Bi_Stankovic11}.  
Observe that when $\lambda=0$,  $\fs_x(t,\eta,\lambda)$ is the STFT.  
Thus we can also regard $\fs_x(t,\eta,\lambda)$ as the STFT after the quadratic term $ e^{ -i\pi \lambda \tau^2}$ is added to match the local change of a non-stationary signal. Hence, we also call  it the localized quadratic-phase Fourier transform in an early version of this paper. 

The CT represents a multicomponent signal in a three-dimension space of time, frequency and chirp rate. Note that when the IF curves of two components $x_{k-1}(t)$ and $x_k(t)$  are crossing, they may be well-separated in the three-dimensional space of the CT if 
$\phi_{k-1}^{\gp\gp}(t)\not = \phi_k^{\gp\gp}(t)$ for $t$ near the crossover time $t_0$. Thus a multicomponent signal with IFs crossover components could be well-separated %and concentrated 
in the three-dimensional space of the CT, and hence, it is feasible to propose  to reconstruct its components based on the CT.

In the following definition, we give some conditions on $A_k(t), \phi_k(t)$ under which the CT-based approach can estimate IFs and retrieve modes with certain error bounds.
%a set of specifications that allow us 
%to separate the blind source, to compute the instantaneous frequencies, and to find both the instantaneous amplitudes and signal components 
\begin{figure}[H]
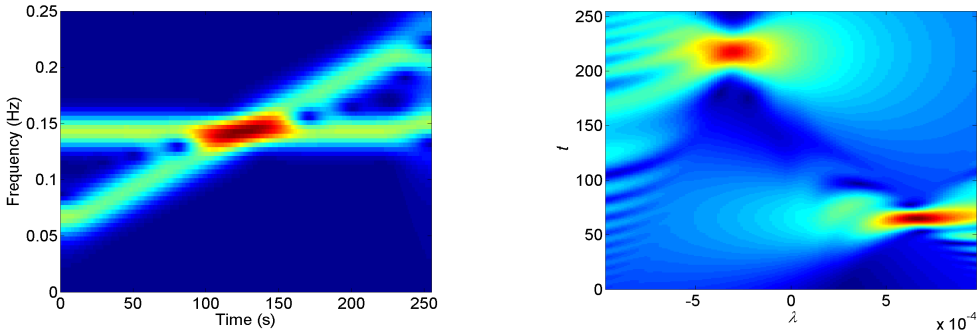

	\centering
	\begin{tabular} {cc}  %{c@{\hskip -0.9cm}c}
		\resizebox{2.5in}{1.8in}{\includegraphics{STFT_harmonic_chirp}}
		\quad &\quad \resizebox{2.5in}{1.8in}{\includegraphics{SSO_frequency_26}}
	\end{tabular}
	\caption{\small STFT, and a fixed slice of the CT of s(t).
		Left: STFT $|V_x(t, \eta)|$; Right: $|\fs_s(t,26/N,\lambda)|$.}
	\label{Figadd1}
\end{figure}	
\begin{mdef}
	\label{def:function_class}
	For an $\ga>0$, let $\mathcal {A}_\ga$ denote the set consisting of (complex) adaptive harmonic models (AHMs) of the form 
\begin{equation}
\label{MAHM}
x(t)=A_{0}(t)+\sum_{k=1}^K x_k(t)=\sum_{k=0}^K A_k(t) e^{i2\pi \phi_k(t)},  % + j \varphi_k \right),
\end{equation}
 with $A_k(t) \in L_\infty(\RR), \; A_k(t)>0, \phi_k(t)\in C^3(\RR), \inf\limits_{t\in \RR} \phi_k'(t)>0, \sup\limits_{t\in \RR} \phi_k'(t)<\infty$,   and
 $A_k(t), \phi_k(t)$ satisfying %\eqref{def_sep_cond_cros} for some $\rho\ge 0$ and $\gt>0$,
\begin{eqnarray}
\label{cond_A}
&&|A_k(t+\tau)-A_k(t)|\leq \alpha^{3}B_1 |\tau|A_{k}(t), ~~k=0,\cdots,K, \\
\label{cond_phi}
&&	\sup_{t\in \RR}|\phi'''_{k}(t)| \leq \alpha^{7} B_2,  ~~k=0,\cdots,K,
	\end{eqnarray}
where $B_1, B_2$ are some positive constants independent of $\ga$, 
and \eqref{def_sep_cond_cros} holds for some $\rho\ge 0, \gt> 0$. 
%\begin{equation}\label{def_sep_cond_cros}
%|\phi'_{k}(t)-\phi'_{\ell}(t)|+\rho |\phi'_{k}(t)-\phi'_{\ell}(t)| \ge 2 \gt,
%\end{equation}
%where $\rho\ge 0$ is a number decided by the user and $\gt>0$ is called the separation resolution. 
\end{mdef}

In this paper, we assume that the separation condition satisfies the condition \eqref{def_sep_cond_cros}.
When $\rho=0$, \eqref{def_sep_cond_cros} is reduced to the well-separated condition required for SST and SSO. The proposed condition allows the IFs of some components $x_k(t)$ to cross as long as their chirp rates
$\phi_k''(t)$ are different near the time $t_0$ where IFs crossing occurs. We consider a two-component signal given by 
\begin{equation*}
	\label{two_component_lfm}
\begin{aligned}
	s(t) &=s_1(t) +s_2(t)= \cos\left( 2 \pi c_1 t + \pi r_1 t^2\right) + \cos \left( 2 \pi c_2 t + \pi r_2 t^2\right).\\
\end{aligned}
	\end{equation*}
Here we let sampling rate $F_s = 1$Hz and just deal with the truncation signal on $t \in [0,255]$, with $N=256$ discrete sampling points. Especially, we consider the case $c_1 = 15/N$, $c_2 = 43/N$, $r_1 = 43/N^2$ and $r_2 = -20/N^2$.
Figure \ref{Figadd1} shows the STFT and a fixed slice of the CT of $s(t)$ with the Gaussian window with $\gs=0.1N$. As shown in this figure, frequencies of sub-signals are crossover at $t_{0}=114$ and however the nonstationary signals are well separated in the three-dimensional space of the CT.

In Definition \ref{def:function_class} and in the rest of this paper,  we also write the trend $A_0(t)$ as $x_0(t)=A_0(t) e^{i2\pi \phi_0(t)}$ with $\phi_0(t)=0$. In the real world, most signals are real-valued. Here for simplicity of presentation of the main result, Theorem 1,  and its proof, we consider complex AHM given in the form of \eqref{MAHM}. The statement of Theorem 1 and its proof still hold for (real-valued) AHM given in \eqref{AHM} by extra arguments.

Denote 
\begin{equation}\label{def_M_u}
\mu=\mu(t)= \min_{0\leq k\leq K} |A_{k}(t)|, ~~ M=M(t)=\sum_{k=0}^K |A_{k}(t)|.
\end{equation}
The main result of this paper to be stated in Theorem \ref{theo:LQFT} below applies to any value $t\in\mathbb{R}$. Since $t$ will be fixed throughout the proof of that theorem, for simplicity we also use $\mu$ and $M$ to denote $\mu(t)$ and $M(t)$ respectively, as shown in \eqref{def_M_u}. 

For a window function $g\in L_1(\RR)$, denote
\begin{equation}
\label{def_PFT}
\wb g (\eta, \gl)=\int_{\RR} g(\tau) e^{-i2\pi\eta\tau-i\pi \lambda \tau^2}d\tau.
\end{equation}
$\wb g(\eta, \gl)$ is called a polynomial Fourier transform (of order 2) of $g$ \cite{Bi_Stankovic11,Stankovic13}.

\begin{mdef}\label{definition2} {\rm ({\bf Admissible window function})} \;
A function  $g(t)\geq 0$
is called an admissible window function if $\int_\RR g(t)dt=1$,  \text{supp}$(g)\subseteq[-N, N]$ for some $N>0$, and satisfies the following conditions.
\begin{itemize}
\item[{\rm (a)}] There exists a constant $C$ such that
\begin{equation}\label{inequality_g1}
 %|\frac{1}{a}\int_{-\infty}^{\infty}h(\frac{\tau}{a}){\rm e}^{i2\pi\eta \tau}{\rm e}^{-i\pi\lambda \tau^{2}}d\tau|\leq\frac{L}{a\sqrt{\eta^{2}+\epsilon|\lambda|}}.
|\wb g (\eta, \gl)|\le \frac  {C}{\sqrt{ |\eta|+ |\lambda|}}, \; \forall \eta, \gl\in \RR.
\end{equation}
\item[{\rm (b)}]
If there exists a constant $D$ such that %for sufficiently small $\varepsilon>0$,
\begin{equation}
\label{inequality_g2}
1- |\wb g (\eta, \gl)|\leq D \varepsilon,
\end{equation}
holds for sufficiently small $\varepsilon>0$ and $(\eta, \gl)$ in the neighborhood of $(0, 0)$, then $\eta$ and $\gl$ must satisfy
\begin{equation}\label{inequality_g3}
|\eta |=o(1) ~~{\rm and} ~~ |\lambda|=o(1) ~~{\rm as}~~ \varepsilon \to 0.
\end{equation}
\end{itemize}
\end{mdef}

When $g$ is the Gaussian function given by
 \begin{equation}
\label{def_g}
g(t)=\frac 1{\sqrt {2\pi}} \; e^{-\frac {t^2}2},
\end{equation}
then (refer to \cite{Leon_Cohen,LCHJJ18, LCJ18})
\begin{equation}
\label{g_PFT}
\wb g(\eta, \gl)=\frac 1{\sqrt{1+i2\pi\gl}} e^{-\frac{2\pi^2 \eta ^2}{1+i2\pi \gl}}.
\end{equation}
One can verify that $|\wb g(\eta, \gl)|=\frac 1{(1+4\pi^2\gl^2)^{1/4}} e^{-\frac{2\pi^2 \eta ^2}{1+(2\pi \gl)^2}}$ satisfies conditions (a)(b) in Definition \ref{definition2}.

\bigskip

Observe that for an admissible window function $g$, we have
$$
|\wb g(\eta, \gl)|\le \wb g(0, 0)=1.
$$
In addition, from \eqref{inequality_g1}, we have
\begin{equation}\label{inequality_g}
 %|\frac{1}{a}\int_{-\infty}^{\infty}h(\frac{\tau}{a}){\rm e}^{i2\pi\eta \tau}{\rm e}^{-i\pi\lambda \tau^{2}}d\tau|\leq\frac{L}{a\sqrt{\eta^{2}+\epsilon|\lambda|}}.
|\wb g (\eta, \gl)|\le \frac  {L}{\sqrt{ |\eta|+ \rho |\lambda|}}, \; \forall \eta, \gl\in \RR,
\end{equation}
where $L=\max(1, \sqrt \rho) C$, and $\rho\ge 0$ is  the number in \eqref{def_sep_cond_cros}.

\begin{theo}\label{theo:LQFT}
Let $x(t)\in {\mathcal A}_\ga$ for some $\ga>0$, and $\fs_x(t,\eta,\lambda)$ be the CT of $x(t)$ with an admissible window function $g$. Let $\gs=\frac {c_0}{\ga^2}$ for some $c_0>0$.
%If $\phi^\gp_k(t), 0\le k\le K$ satisfy \eqref{def_sep_cond_cros} for some $\rho\ge 0$, $\gt>0$  and
%sufficiently small
If 
\begin{equation}
\label{cond_alpha}
\alpha\le \min\Big\{\frac\mu{4 M c_0 N(B_1+ \frac \pi 3 B_2 c_0^2 N^2)}, \frac{\mu\sqrt{c_0 \gt}}{4 M L}\Big\},
\end{equation}
then the following statements hold.
\begin{enumerate}
\item[{\rm (a)}] The set $\mathcal{G}(t)=\{(\eta,\lambda)  : |\fs_x(t, \eta,\lambda)| \geq \mu/2\}$ can be expressed as a disjoint union of exactly $K+1$ non-empty sets
\begin{equation}
 \label{def_Gell}
 \begin{array}{l}
 \mathcal {G}_{\ell}(t)=\big\{(\eta,\lambda) \in\mathcal{G}(t):\gs |\eta-\phi^{'}_{\ell}(t)|+\rho \gs^2 |\lambda-\phi^{''}_{\ell}(t)| \leq \big(\frac {4LM}{\mu}\big)^2
\big\},\ell=0,\cdots, K.
 \end{array}
 \end{equation}

%\begin{equation}
%\label{def_Gell}
%\begin{aligned}
%&\mathcal {G}_{\ell}(t)=\big\{(\eta,\lambda) \in\mathcal {G}(t):\\
% &\gs |\eta-\phi^{'}_{\ell}(t)|+\rho \gs^2 |\lambda-\phi^{''}_{\ell}(t)| \leq \big(\frac {4LM}{\mu}\big)^2
%\big\}, \\
%&\ell=0,\cdots, K.\\
%\end{aligned}
%\end{equation}
\item[{\rm (b)}] For each $t$, let
\begin{equation}
\label{def_hateta}
%\begin{aligned}&
(\wh{\eta}_\ell(t), \wh{\lambda}_\ell(t))={\rm argmax}_{(\eta,\lambda)\in\mathcal{G}_{\ell}(t)  }|\fs_x(t,\eta,\lambda)|, \; \ell=0,\cdots, K.\\
%\end{aligned}
\end{equation}
Then
\begin{eqnarray}
 &&\label{abs_IA_est}
 \big| |\fs_x(t,\widehat{\eta}_{\ell}(t), \widehat{\lambda}_{\ell}(t))|-A_{\ell}(t) \big|\le \ga M \Big( \frac {L}{\sqrt{c_0 \gt}}+c_0 N B_1+ \frac \pi 3 B_2 c_0^3 N^3\Big),\\
%\end{equation}\begin{equation}
 &&\label{phi_est}
|\wh{\eta}_{\ell}(t)-\phi_{\ell}^{'}(t)|=\frac{1}{\gs} o(1)=\ga^2 o(1), \\
%% \\|\wh{\lambda}_{\ell}-\phi_{\ell}^{''}(t)|=\frac 1{\gs^2} o(1)=\ga^4 o(1) ~~ \hbox{as $\ga\to 0^+$},
% \end{equation}\begin{equation}
 && \label{lambda_est}
|\wh{\lambda}_{\ell}(t)-\phi_{\ell}^{''}(t)|=\frac 1{\gs^2} o(1)=\ga^4 o(1),\\
&& \label{comp_est}
\big| \fs_x(t,\widehat{\eta}_\ell, \widehat{\lambda}_\ell)-x_\ell (t)\big |\leq o(1) +\ga M \Big( \frac {L}{\sqrt{c_0 \gt}}+c_0 N B_1+ \frac \pi 3 B_2 c_0^3 N^3\Big),
 \end{eqnarray}
%\begin{equation}
%\label{comp_est}
%\begin{aligned}
%&\big| \fs_x(t,\widehat{\eta}_\ell, \widehat{\lambda}_\ell)-x_\ell (t)\big |\\
%&\leq o(1) +\ga M \Big( \frac {L}{\sqrt{c_0 \gt}}+c_0 N B_1+ \frac \pi 3 B_2 c_0^3 N^3\Big),\\
%\end{aligned}
%\end{equation}
as $\ga\to 0^+$.
\end{enumerate}
\end{theo}

The proof of Theorem \ref{theo:LQFT} is provided in Appendix. Theorem \ref{theo:LQFT} shows that we can use the CT to separate a multicomponent signal, even when the IF curves of different components are crossover.
More precisely, for a multicomponent signal $x\in \cal A_\ga$,
its sub-signal $x_\ell(t)$ can be reconstructed by
		\begin{equation}
		\label{MSSO_recon_comp}
		\wh x_\ell(t) = \fs_x(t,\wh \eta_\ell, \wh \lambda_\ell), \; \ell=1, 2, \cdots, K.
		\end{equation}
Especially for real-valued signals, we have
		\begin{equation}
		\label{MSSO_recon_real}
		\wh x_\ell(t) =2 \Re e \left \{\fs_x(t,\wh \eta_\ell, \wh \lambda_\ell) \right\}, \; \ell=1, 2, \cdots, K.
		\end{equation}
		The trend in \eqref{MAHM} can also be recovered by
		\begin{equation}
		\label{MSSO_recon_trend}
		\wh x_0(t) = \fs_x(t,0,0)
		\end{equation}
	    for complex signals or
	    \begin{equation}
	    \label{MSSO_recon_trend_real}
	    \wh x_0(t) = 2\Re e \left \{ \fs_x(t,0,0) \right \}
	    \end{equation}
	    for real-valued signals.

It is important to point out that although the parameters $\alpha =\alpha(t)$, $\sigma=\sigma(t)$ are time-dependent, our experiments show that the dependence is not quite sensitive for reasonably well-behaved signals in $A(t)$. Therefore, we use the constant parameters $\alpha$ in Section $4$. However, in fact one can select the time-varying parameter so that the corresponding adaptive chirplet transform of the components of a multicomponent signal have sharp representations and are well-separated.

 We call the method to obtain IFs $\phi_k^\gp(t)$, chirp rates $\phi_k^{\gp\gp}(t)$ by \eqref{def_hateta}, and 
retrieve modes $x_\ell(t)$ and trend $A_0(t)$ by \eqref{MSSO_recon_comp}-\eqref{MSSO_recon_trend_real}
the chirplet transform-based signal separation scheme (or CT3S for short). 

Next we would like to explain the implementation of the proposed method. The first step is to apply  the CT  in (\ref{def_MSSO}) to (uniform or nonuniform) samples of $x(t)$. Then there are many narrow bands (clusters) corresponding to blind sources in time-frequency-chirprate spectrogram $|\fs_x(t, \eta,\lambda)|$. By applying an appropriate smoothing curve fitting scheme to obtain each local maximum curve. The second step is to find extreme curves in the narrow bands. Since the number of IF curves is unknown, curve fitting can be carried out by estimating one IF curve at a time, till there appears to be no curve left behind. After obtaining  local maximum values $\widehat{\eta}_{\ell}$ and $\widehat{\lambda}_{\ell}$, the third step is to compute instantaneous amplitudes, frequencies, chirp rates, and signal components by (\ref{abs_IA_est}), (\ref{phi_est}), (\ref{lambda_est}), and (\ref{comp_est}), respectively.
	
The dominant computational complexity of the proposed algorithm is to calculate the CT of a given signal. Assume that the number of samples of a given signal is $m$, the number of grids of the chirp rate is $n$ and the length of a window is $L$, then the computational complexity of the CT is $\mathcal {O}(nmL\log_{2}L)$. The computational cost of clustering is $\mathcal {O}(nmL)$. Then the set $\mathcal{G}(t)=\{(\eta,\lambda)  : |{\mathfrak{S}}_x(t, \eta,\lambda)| \geq \mu/2\}$ can be expressed as a disjoint union of exactly $K+1$ non-empty sets. For each set, the computational cost of recovering magnitudes, frequencies, waveforms and chirp rates is  $\mathcal {O}(m)$. Therefore, the total computational complexity is $\mathcal {O}(nmL\log_{2}L)$.

\begin{rem}  \;  The error bounds in \eqref{abs_IA_est} and \eqref{comp_est} are related to
\begin{equation*}
 M \big(\frac {L}{\sqrt{c_0 \gt}}+c_0 N B_1+ \frac \pi 3  B_2 c_0^3 N^3).
\end{equation*}
We shall choose $c_0$ such that the above quantity is sufficient small. Considering that the third power $c_0^3$ of $c_0$ appears in this quantity, we shall choose $c_0\le 1$. With $c_0\le 1$, one may choose $c_0$ such that $\frac {L}{\sqrt{c_0 \gt}}=c_0 N B_1$ with which
$\frac {L}{\sqrt{c_0 \gt}}+c_0 N B_1$ gains its minimum. Thus we may let
\begin{equation*}
c_0=\min\Big(1, \big(\frac {L^2}{B_1^2 N^2 \gt}\big)^{\frac 13}\Big).
\end{equation*}
\end{rem}

Note that $\gs=\frac {c_0}{\ga^2}$. To improve recovery
performance, we need to increase the window length $\gs$. However, for an arbitrary nonstationary signal,
the recover error cannot be reduced just by setting a sufficiently large number $\gs$.
We propose a modified algorithm in the next subsection to reduce the recovery errors.

\subsection{A improved scheme derived from linear chirp local approximation}

From the uncertainty principle \cite{Leon_Cohen}, one can get the minimum time-bandwidth product with a Gaussian window, which means the optimal two-dimensional resolution of time and frequency is attained when a Gaussian window is used. So next we consider the CT with the Gaussian window function given by \eqref{def_g}.

 When $x(t)\in \mathcal A_\ga$, then conditions \eqref{cond_A} and \eqref{cond_phi} imply that
 each component $x_k(t)$ is well-approximated locally by linear chirps (also called linear frequency modulation signals) at any time $t$  in the sense of (see Appendix)
 \begin{equation}
 \label{LFM_approx}
 x_k(t+\tau)\approx x_k(t)e^{i2\pi (\phi^\gp_k(t)\tau+\frac 12\phi^{\gp \gp}_k(t)\tau^2) }, \; t\in \RR,
 \end{equation}
 for $\tau \approx 0$. For a fixed $t$, the right side of \eqref{LFM_approx}, as a function of $\tau$, is a linear chirp (linear frequency modulation signal).
 Thus from \eqref{LFM_approx}, we have 
 \begin{eqnarray*}
\nonumber &&\fs_{x_k}(t,\eta,\lambda)\approx\int_{\RR}  x_k(t)e^{i2\pi (\phi^\gp_k(t)\tau+\frac 12\phi^{\gp \gp}_k(t)\tau^2) }  \frac 1 {\gs} g\big(\frac \tau {\gs}\big) e^{-i2\pi\eta\tau-i\pi \lambda\tau^2}d\tau
\\
\nonumber &&~~~~~~~~~~~~~~=x_k(t)\int_{\RR}  \frac 1 {\gs} g\big(\frac \tau {\gs}\big) e^{-i2\pi(\eta-\phi^\gp_k(t))\tau-i\pi (\lambda-\phi^{\gp\gp}_k(t)) \tau^2}d\tau
\\
\nonumber&&~~~~~~~~~~~~~~=x_k(t) \wb g\big(\gs(\eta-\phi^\gp_k(t)), \gs^2(\gl-\phi^{\gp\gp}_k(t))\big)
\\
\nonumber &&~~~~~~~~~~~~~~= \frac 1 {\sqrt{1+i2 \pi \sigma^2(\lambda-\phi^{\gp\gp}_k(t))}} x_k(t) e^{-\frac{2\pi^2 \sigma^2}{1+i2\pi \sigma^2(\lambda-\phi^{\gp\gp}_k(t))}(\eta-\phi^\gp_k(t))^2} \;\hbox{(by \eqref{g_PFT})}. 
\end{eqnarray*}
Denote 
$$
\mathbb{A}(\gl)=\frac 1 {\sqrt{1+i2 \pi \sigma^2\gl}}, \; \Omega(a,b)=e^{-\frac{2\pi^2 \sigma^2}{1+i2\pi \sigma^2 a}b^2}. 
$$
Then the above equation leads to 
\begin{equation}
\label{MSSO_LFM}
\fs_{x_k}(t,\eta,\lambda)\approx x_k(t) \mathbb{A}(\lambda-\phi^{\gp\gp}_k(t))
		\Omega(\lambda -\phi^{\gp\gp}_k(t),\eta-\phi^\gp_k(t)). 
\end{equation}
%{\scriptsize{
%		\begin{eqnarray}		
%		\fs_{x_k}(t,\eta,\lambda)
%		\hskip -0.6cm &&\approx \int_{\RR}  x_k(t)e^{i2\pi (\phi^\gp_k(t)\tau+\frac 12\phi^{\gp \gp}_k(t)\tau^2) }\mathcal{K}_\sigma(\tau,\theta,\lambda)d\tau \nonumber \\
%		\nonumber&&=x_k(t)\int_{\RR}  \frac 1 {\gs} g\big(\frac \tau {\gs}\big) e^{-i2\pi(\eta-\phi^\gp_k(t))\tau-i\pi (\lambda-\phi^{\gp\gp}_k(t)) \tau^2}d\tau\\
%	\nonumber	&&=x_k(t) \wb g\big(\gs(\eta-\phi^\gp_k(t)), \gs^2(\gl-\phi^{\gp\gp}_k(t))\big)\\ 	
%	\nonumber	&&= \frac 1 {\sqrt{1+i2 \pi \sigma^2(\lambda-\phi^{\gp\gp}_k(t))}} x_k(t) e^{-\frac{2\pi^2 \sigma^2}{1+i2\pi \sigma^2(\lambda-\phi^{\gp\gp}_k(t))}(\eta-\phi^\gp_k(t))^2} \; \hbox{(by \eqref{g_PFT})} \\
%	\label{MSSO_LFM}	&& =: x_k(t) \mathbb{A}(\lambda-\phi^{\gp\gp}_k(t))
%		\Omega(\lambda -\phi^{\gp\gp}_k(t),\eta-\phi^\gp_k(t)),
%		\end{eqnarray}}}
%Hence by Theorem \ref{theo:LQFT} and \eqref{MSSO_recon_comp}, 
Observe that $\mathbb{A}(0)=1, \Omega(\cdot, 0)=1$, 
therefore, if $\big(\wh \eta_k, \wh \lambda_k\big) =  (\phi^\gp_k(t), \phi^{\gp\gp}_k(t))$, then 
$\fs_{x_k}(t,\wh \eta_k, \wh \lambda_k) =  x_k(t)$. From the expression of $\mathbb{A}$, we see $|\fs_{x_k}(t,\eta,\lambda)|$ is not exponentially decaying with respect to variable $\gl$ as  $\gl \to \infty$.
For a multicomponent signal $x(t)$, we need to find the extreme points $(\wh \eta_\ell, \wh \lambda_\ell), \ell=1,2,\cdots, K$ of $|\fs_x(t, \eta, \gl)|$ with $(\eta, \gl) \in \cG_\ell$.
Therefore, we first need to find these $K$ non-empty sets $\mathcal{G}_\ell, \ell=1,2, \cdots, K$ as defined in Theorem \ref{theo:LQFT} on the  three-dimensional space of the CT $\fs_x(t,\eta,\lambda)$.
% is very important. Actually,
However when the IF curves of two components $x_{\ell-1}(t)$ and $x_\ell(t)$  are crossing,
$\mathcal{G}_{\ell-1}$ and $\cG_\ell$ are not sufficiently separated in the three-dimensional space of the CT due to the fact that
	%the threshold $\mu$ in Theorem \ref{theo:LQFT} is not easy to be selected
$|\fs_x(t,\eta,\lambda)|$ is not a function with a fast rate of decay with respect to the chirp rate $\lambda$. Next we also use an example to explain this point.

Let $s(t)$ be a two-component signal %consisting of two LFM modes
given by
	\begin{equation}
	\label{two_component_lfm}
\begin{aligned}
	s(t) &=s_1(t) +s_2(t)= \cos\left( 2 \pi c_1 t + \pi r_1 t^2\right) + \cos \left( 2 \pi c_2 t + \pi r_2 t^2\right).\\
\end{aligned}
	\end{equation}
\begin{figure}[H]
		\centering
		%\hspace{-0.5cm}
		%\begin{tabular}{cccc}
		\begin{tabular}{c@{\hskip -0.06cm}c@{\hskip -0.06cm}c@{\hskip -0.06cm}c}
			\resizebox{1.6in}{1.2in}{\includegraphics{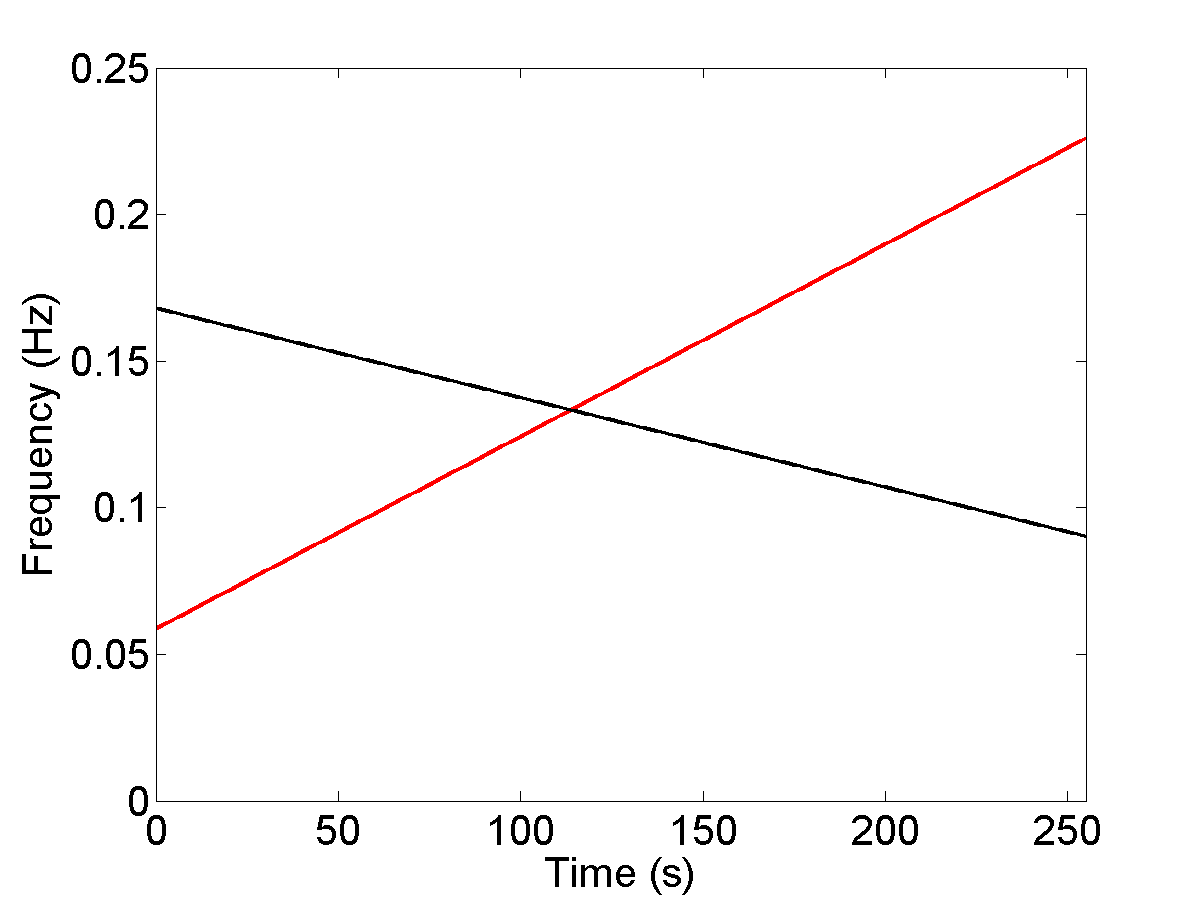}} \quad &
			\resizebox{1.6in}{1.2in}{\includegraphics{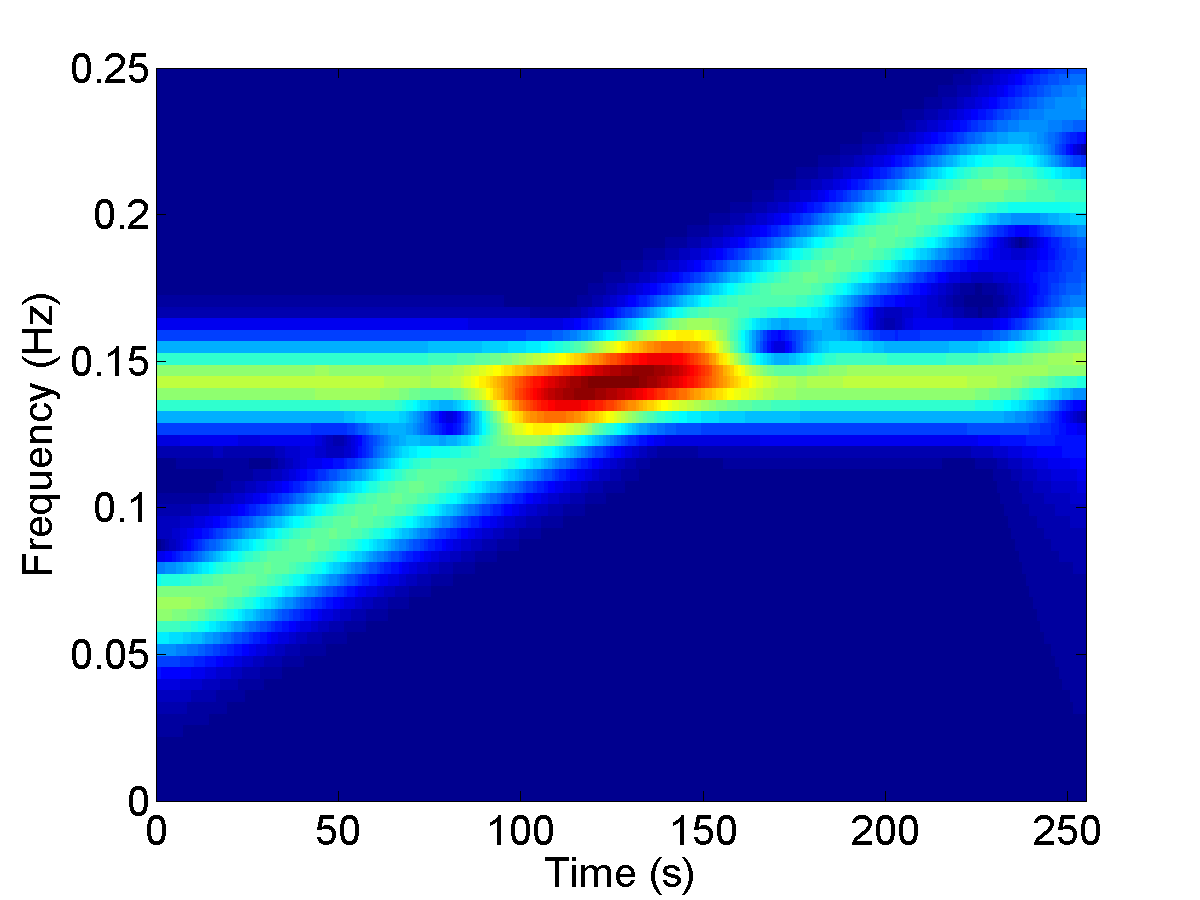}} \quad &
			\resizebox{1.6in}{1.2in}{\includegraphics{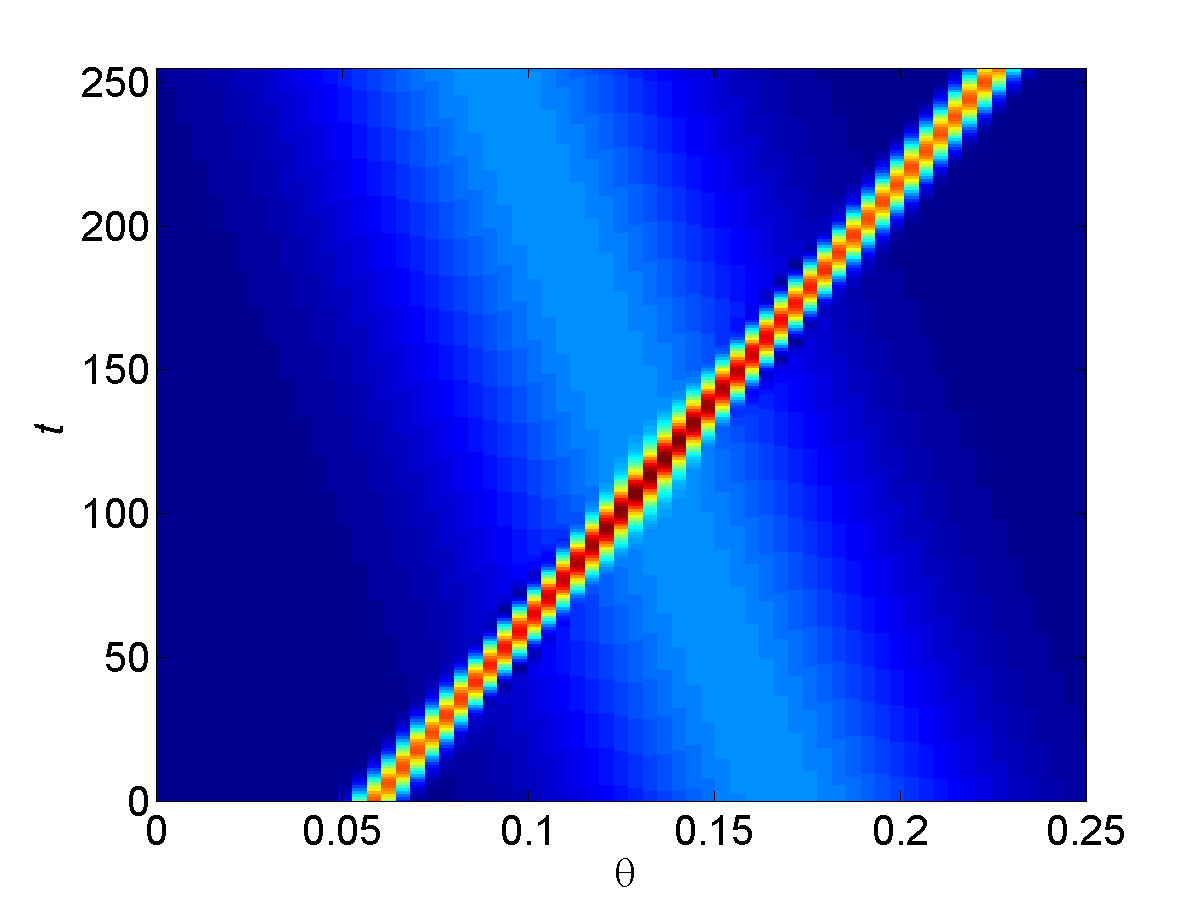}} \quad &
			\resizebox{1.6in}{1.2in}{\includegraphics{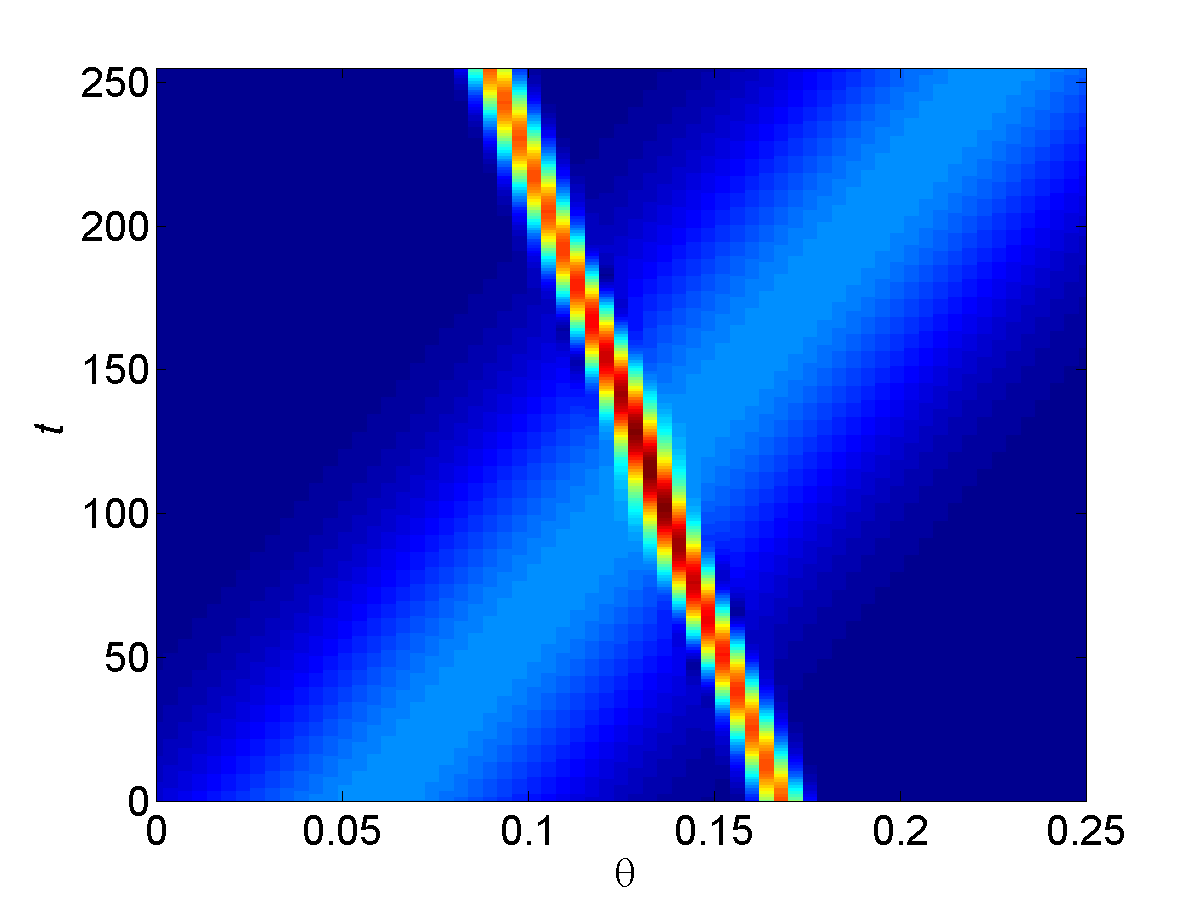}}\\
			\resizebox{1.6in}{1.2in}{\includegraphics{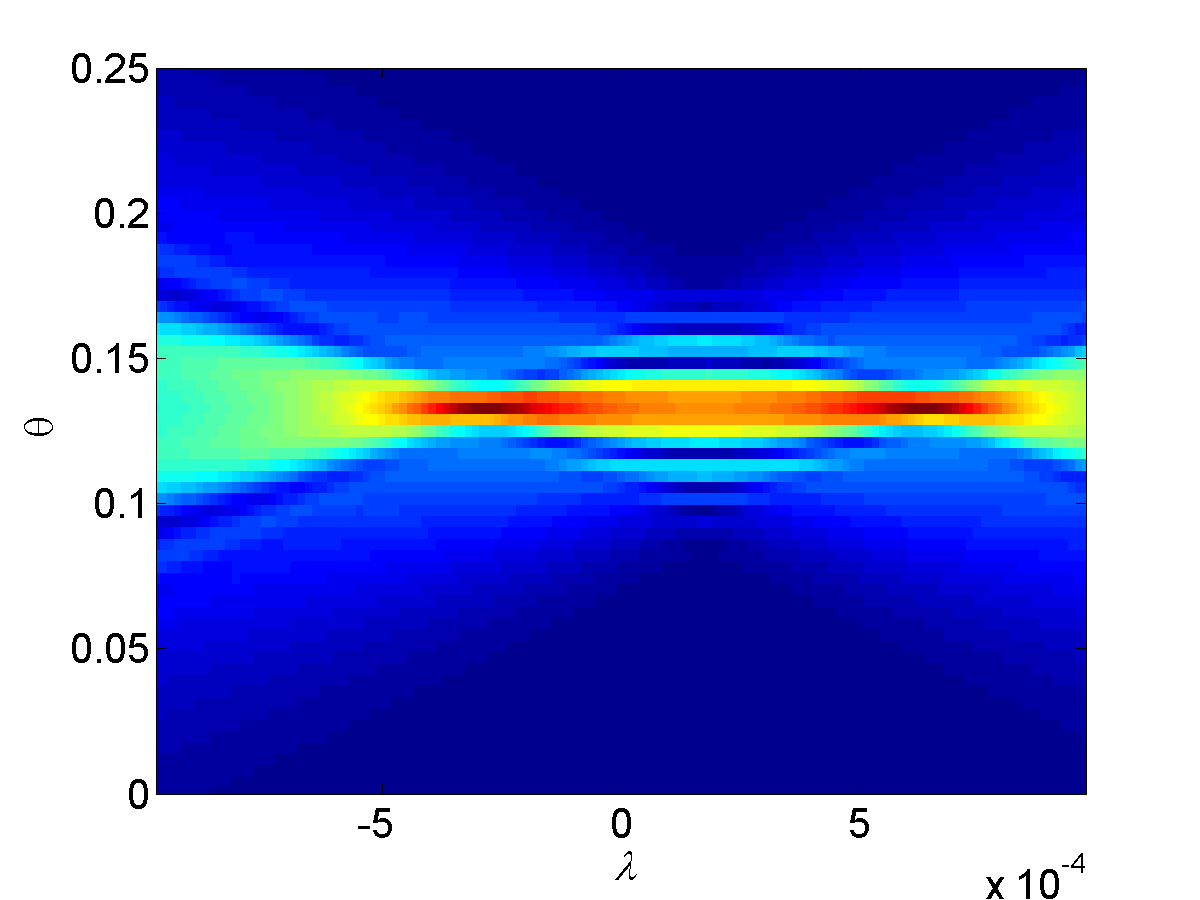}} \quad &
			\resizebox{1.6in}{1.2in}{\includegraphics{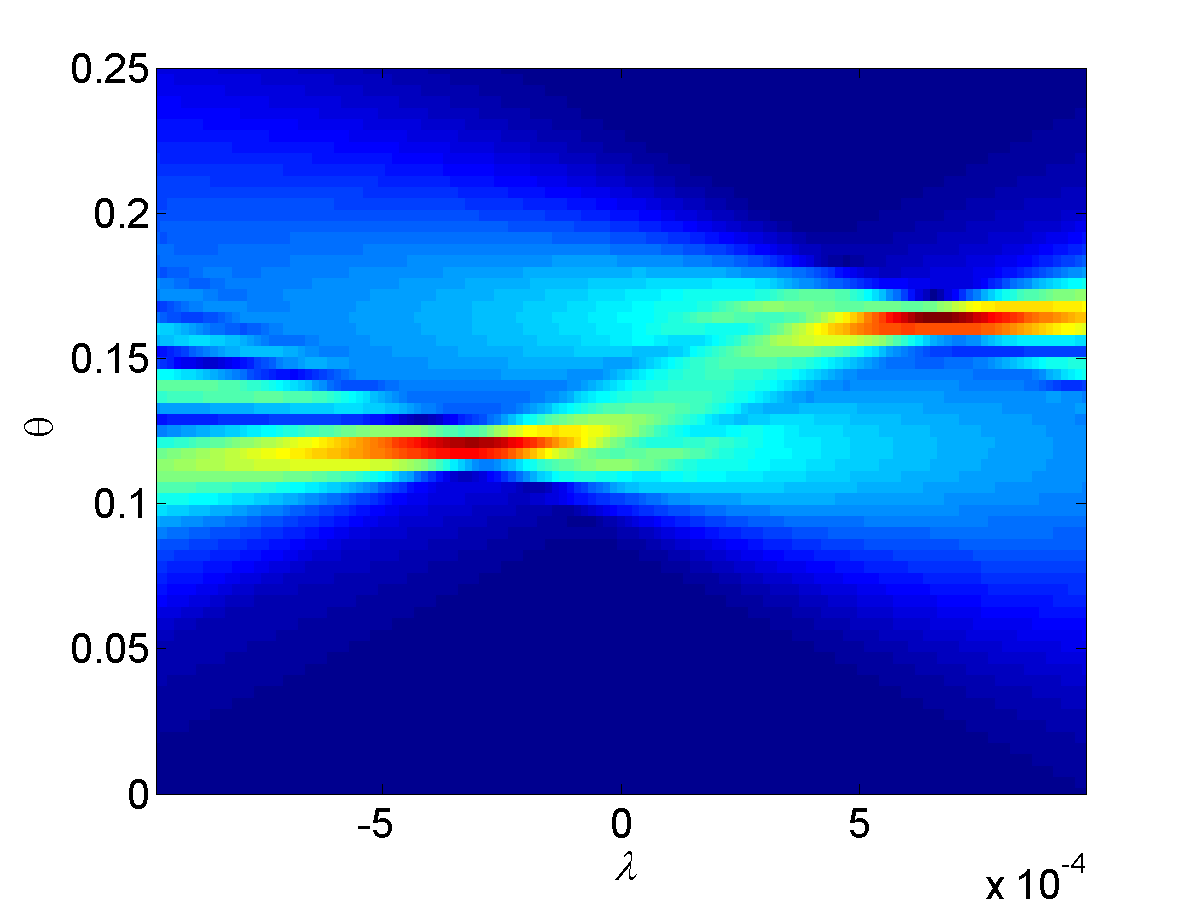}} \quad &
			\resizebox{1.6in}{1.2in}{\includegraphics{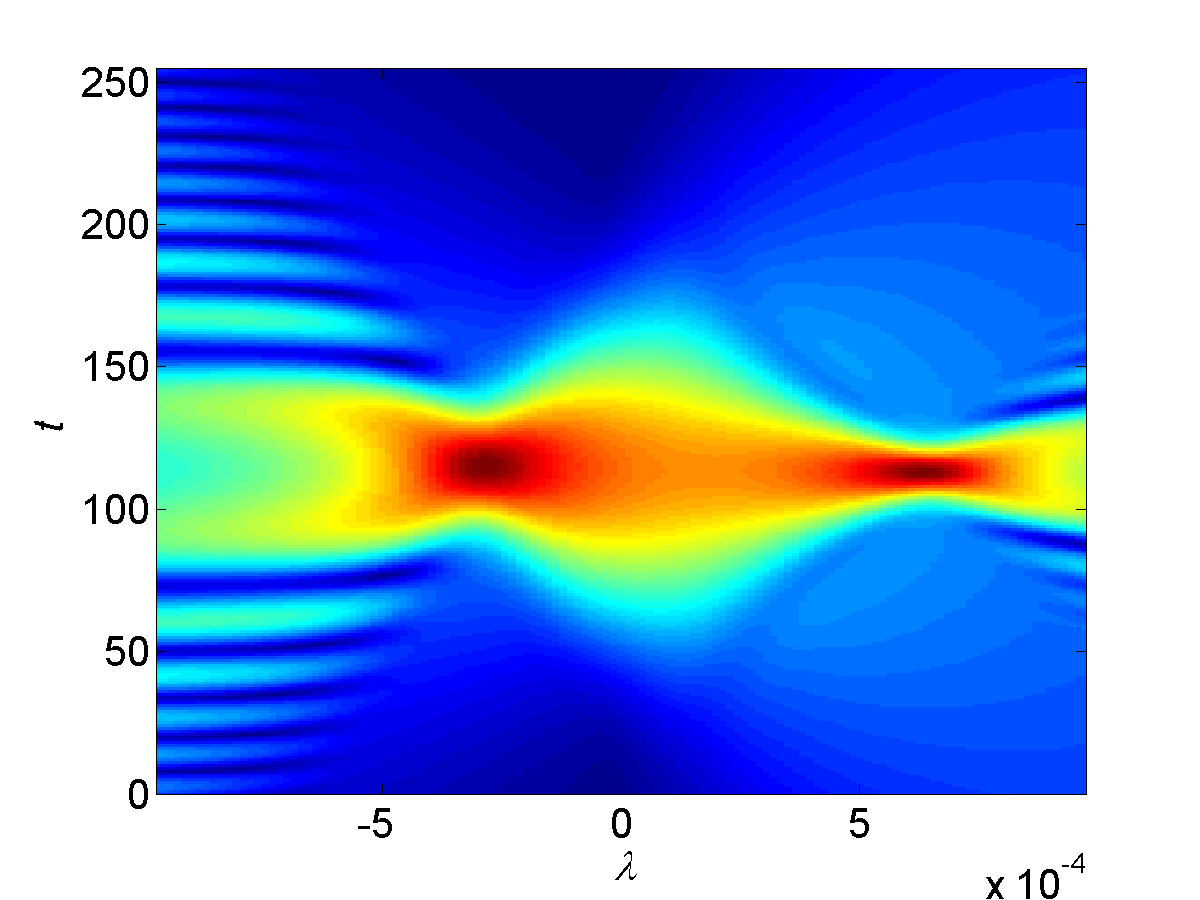}} \quad &
			\resizebox{1.6in}{1.2in}{\includegraphics{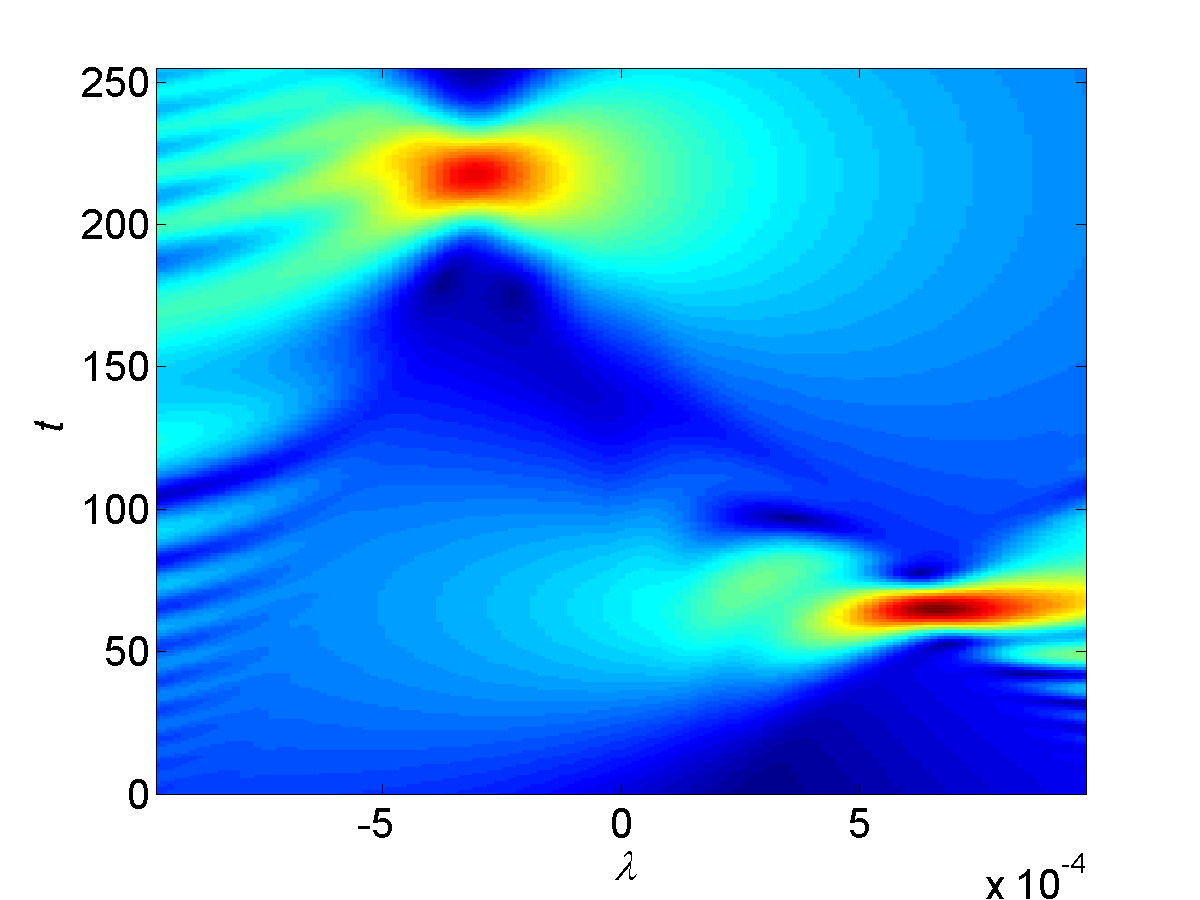}}
		\end{tabular}
	\caption{\small IFs, STFT, and some slices of the CT of the two-component signal $s(t)$ in \eqref{two_component_lfm}.
			Top row (from left to right): IFs, STFT $|V_x(t, \eta)|$, $|\fs_s(t,\eta,r_1)|$ and $|\fs_s(t,\eta,r_2)|$;
			Bottom row  (from left to right): $|\fs_s(114,\eta,\lambda)|$, $|\fs_s(160,\eta,\lambda)|$,  $|\fs_s(t,34/N,\lambda)|$ and $|\fs_s(t,26/N,\lambda)|$.}
		\label{Fig:two_component_SSOmatrix}
	\end{figure} 
\noindent Here we let sampling rate $F_s = 1$Hz and just deal with the truncation signal on $t \in [0,255]$, with $N=256$ discrete sampling points. Especially, we consider the case $c_1 = 15/N$, $c_2 = 43/N$, $r_1 = 43/N^2$ and $r_2 = -20/N^2$.
Figure \ref{Fig:two_component_SSOmatrix} shows the IFs and the STFT of $s(t)$, and some special slices of the CT of the two-component signal $s(t)$ with the Gaussian window with $\gs=0.1N$. Observe that when $\lambda = r_1$ or $\lambda = r_2$, $s_1(t)$ and $s_2(t)$ are well represented on the two time-frequency planes $|\fs_s(t,\eta, r_1)|$ and $|\fs_s(t,\eta, r_2)|$, respectively.  
	
Now we focus on the crossing point of the two IFs of $s_1(t)$ and $s_2(t)$, which is located at around $(t_0, \eta_0)$ with $t_0=114$ and $\eta_0= 34/N$. From the 1st and 3rd panels in the bottom row of Figure \ref{Fig:two_component_SSOmatrix},  namely the crossover point, the two components appear to two separated peaks on the chirp rate-frequency and chirp rate-time planes,  respectively, but not very clearly and sharply. Taking $t=t_0$ and $\eta=\eta_0$ in \eqref{MSSO_LFM}, we have
\begin{eqnarray*}
%\begin{array}{l}
&&\mathcal{L} (\lambda)=|\fs_s(114,\eta_0,\lambda)| \\
&&~~~~~~= \bigg| \frac 1 {2\sqrt{1+i2\pi \sigma^2(\lambda-r_1)}}\tilde s_1(t) +
 \frac 1 {2\sqrt{1+i2\pi \sigma^2 (\lambda-r_2)}}\tilde s_2(t) \bigg|\\
 &&~~~~~~\approx \frac 1 {2\sqrt [4] {1+4 \pi^2 \sigma^4(\lambda-r_1)^2}} + \frac 1 {2\sqrt [4] {1+4 \pi^2 \sigma^4 (\lambda-r_2)^2}},
%\end{array}
\end{eqnarray*}	
where $\tilde s_k(t) = e ^{i2\pi c_k t +i\pi r_k t^2 }$ denotes the analytic signal of $s_k(t)$, $k=1,2$. 	
Note that the two parts in $\mathcal{L} (\lambda)$ above are corresponding to $s_1(t)$ and $s_2(t)$, respectively, which are centered at $\lambda=r_1=43/N^2$ and $\lambda=r_2=-20/N^2$.
Then let $t=114$ and $\lambda=r_2$, we obtain
	$$ \mathcal{T} (\eta)=|\fs_s(114,\eta, r_2)| \approx \frac{1}{2} e^{-2 \pi^2 \gs^2 (\eta-\eta_0)^2 }. $$
	Compared to $\mathcal{T} (\eta)$, $\mathcal{L} (\lambda)$ is a slowly attenuated function from the two extrema located at $\lambda=r_1$ and $\lambda=r_2$.
This explains why there are two components which not well separated in either the 1st or the 3rd panel in the bottom row of Figure \ref{Fig:two_component_SSOmatrix}, while there is only one component in the 4th panel of the top row.
		%%%%%%%%%%%%%%%%%%%the beginning of Figure 3 %%%%%%%%%%%%%%%

	%%%%%%%%%%%%%%%%the end of Figure 3 %%%%%%%%%%%%%%%%%%%%%
	 	
	To solve the above problem, we first consider the popular integral transform, namely Radon transform \cite{Beylkin_87}. Since $\fs_s(t,\eta,\lambda)$ is a three dimensional function, we may use the 3D Radon transform as that in \cite{Averbuch03} to detect the signal components, which is similar to the case of two dimensional Radon transform in \cite{Chui_98}.

For each time $t$, with a pair of angles $(\varphi_1,\varphi_2)$,  the 3D Radon transform is an  integral transform along the lines in the 3D space,
\begin{equation*}
\begin{aligned}
%	\label{define_RT}
	R&(t',\varphi_1,\varphi_2)
= \int_{-\infty}^{\infty}\int_{-\infty}^{\infty}\int_{-\infty}^{\infty}
	\fs_s(t,\eta,\lambda)\delta (t'-\eta \sin \varphi_1 \cos \varphi_2 + \eta \sin \varphi_1 \sin \varphi_2+ t \cos \varphi_1) dt d\eta d \lambda,\\
\end{aligned}
	\end{equation*}
where $\delta$ is a  Dirac's delta function.
	
	If all the components of a multicomponent signal are linear chirps, then they can be well represented in the 3D space of $R(t',\varphi_1,\varphi_2) $. However, the computational burden is heavy for this 3D Radon transform. Moreover, a nonstationary signal is approximated by linear chirps just for local time around $t$.
	
Considering the relation among $t$, $\eta$ and $\lambda$ in $\fs_s(t,\eta,\lambda)$,   	
	namely $\phi'(t+u) = \eta +\lambda u, u\in [-b, b]$, where $b>0$ is small enough,
we introduce a time-frequency
filter operator $\mathcal{F}$ on $\fs_x$ to make different sub-signals more distinguishable in the 3D space of time, frequency and chirp rate.

For a multicomponent signal $x(t)$, the time-frequency filter-matched CT is defined by
			\begin{equation}
			\label{def_FLQFT}
			\mathcal{F}^h\left(\fs_x\right)(t,\eta,\lambda) = \frac 1 b \int_{\RR}  %{\hbar_b}
			h\big(\frac ub\big) \left|\fs_x(t+u,\eta +\lambda u,\lambda)\right | du,
			\end{equation}
where $h(t)$ is a fast decay window function with $h(t)\ge 0$ and $\int_\RR h(t) dt=1$.
In the following experiments, in order to simplify calculations, we will use a rectangular window with width of $2b$.
	
%%%%%%%%%%%%%%%%%%the beginning of Figure 4 %%%%%%%%%%%%%%%
\begin{figure}[H]
		\centering
		\begin{tabular}{ccc}
		% \begin{tabular}{c@{\hskip -0cm}c@{\hskip -0cm}c@{\hskip -0cm}c}
			\resizebox{1.6in}{1.2in}{\includegraphics{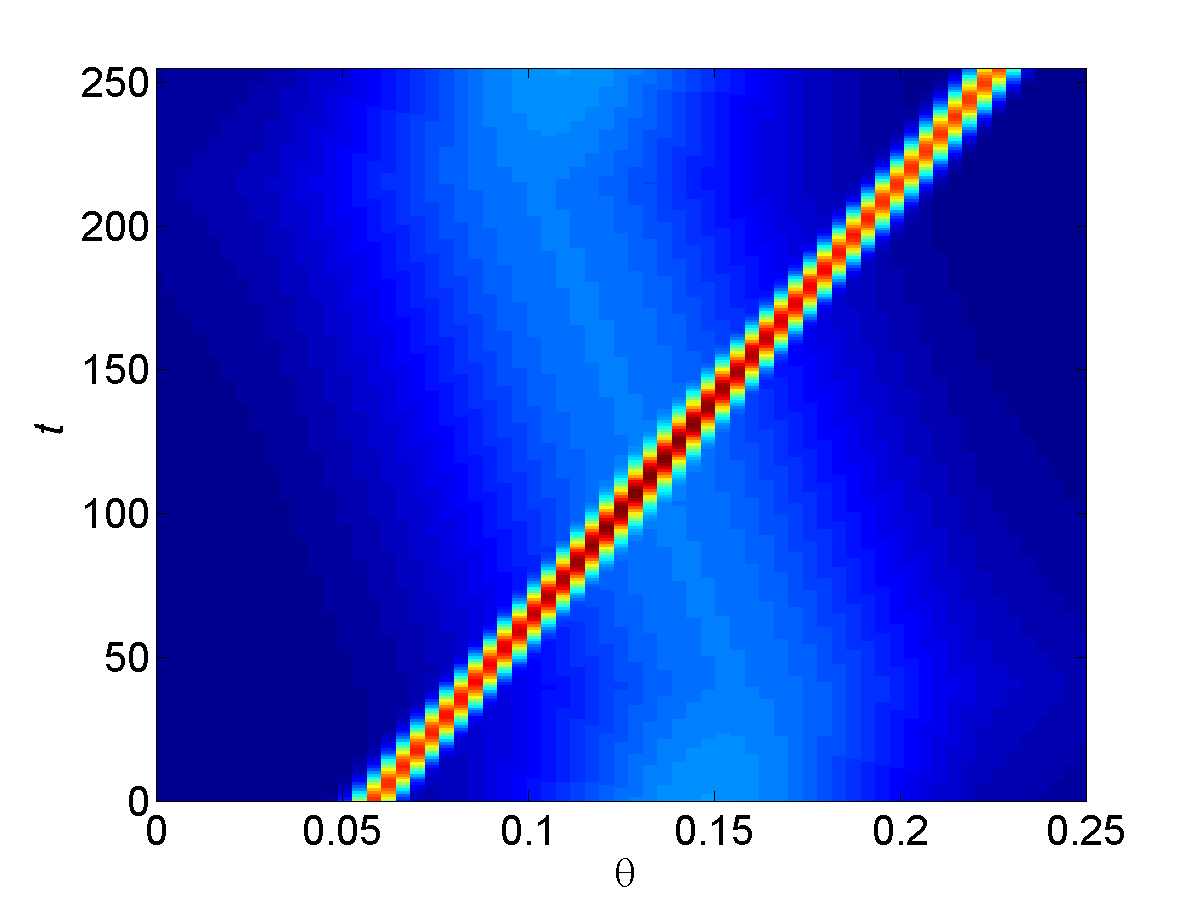}} \quad &
			\resizebox{1.6in}{1.2in}{\includegraphics{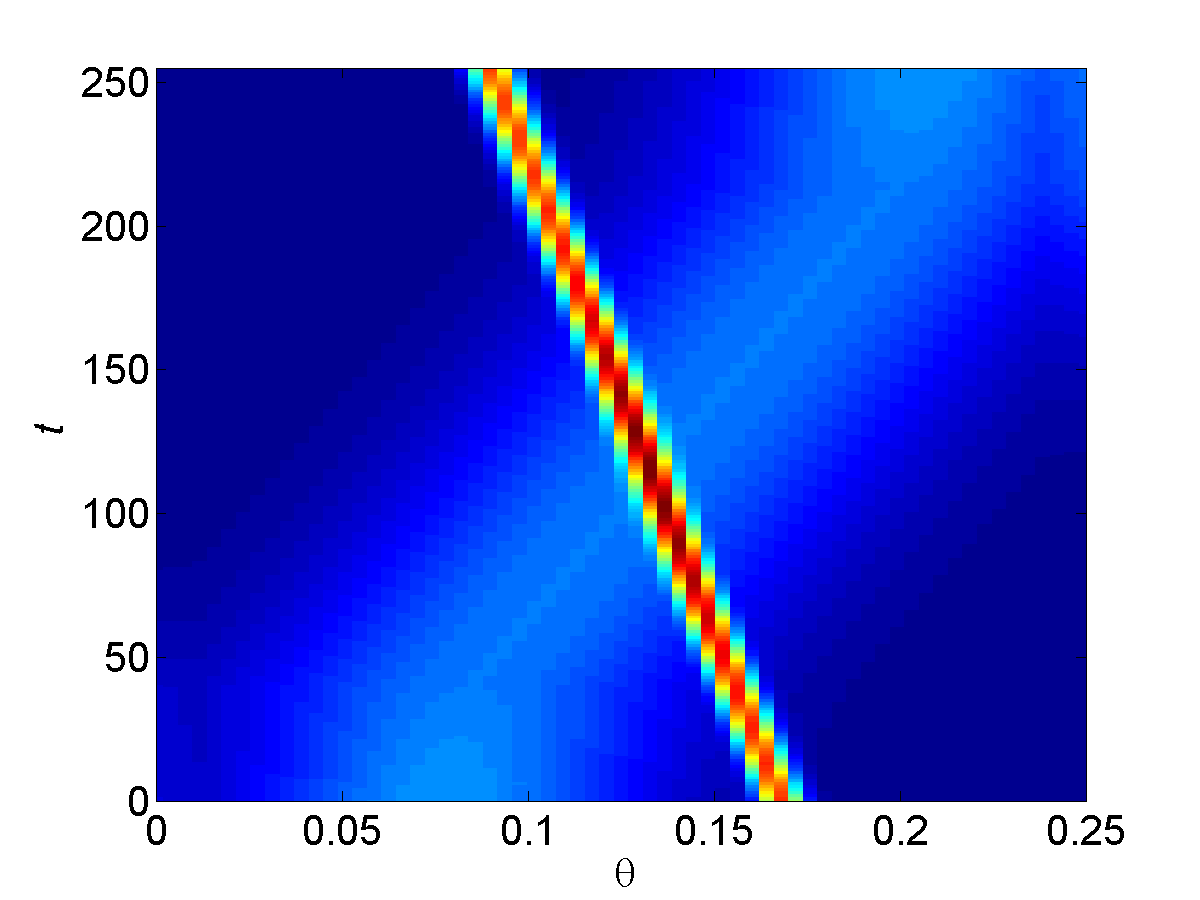}} \quad &
			\resizebox{1.6in}{1.2in}{\includegraphics{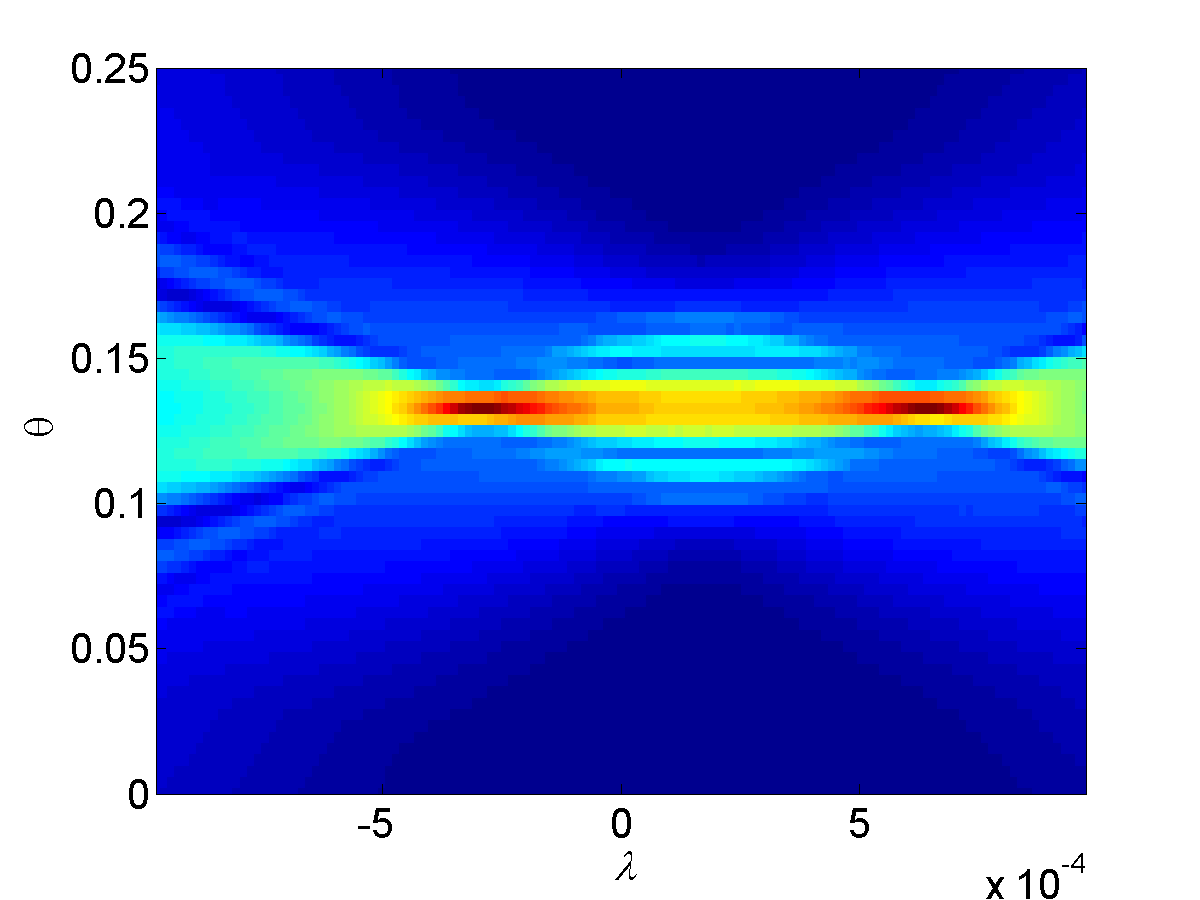}} \\
			\resizebox{1.6in}{1.2in}{\includegraphics{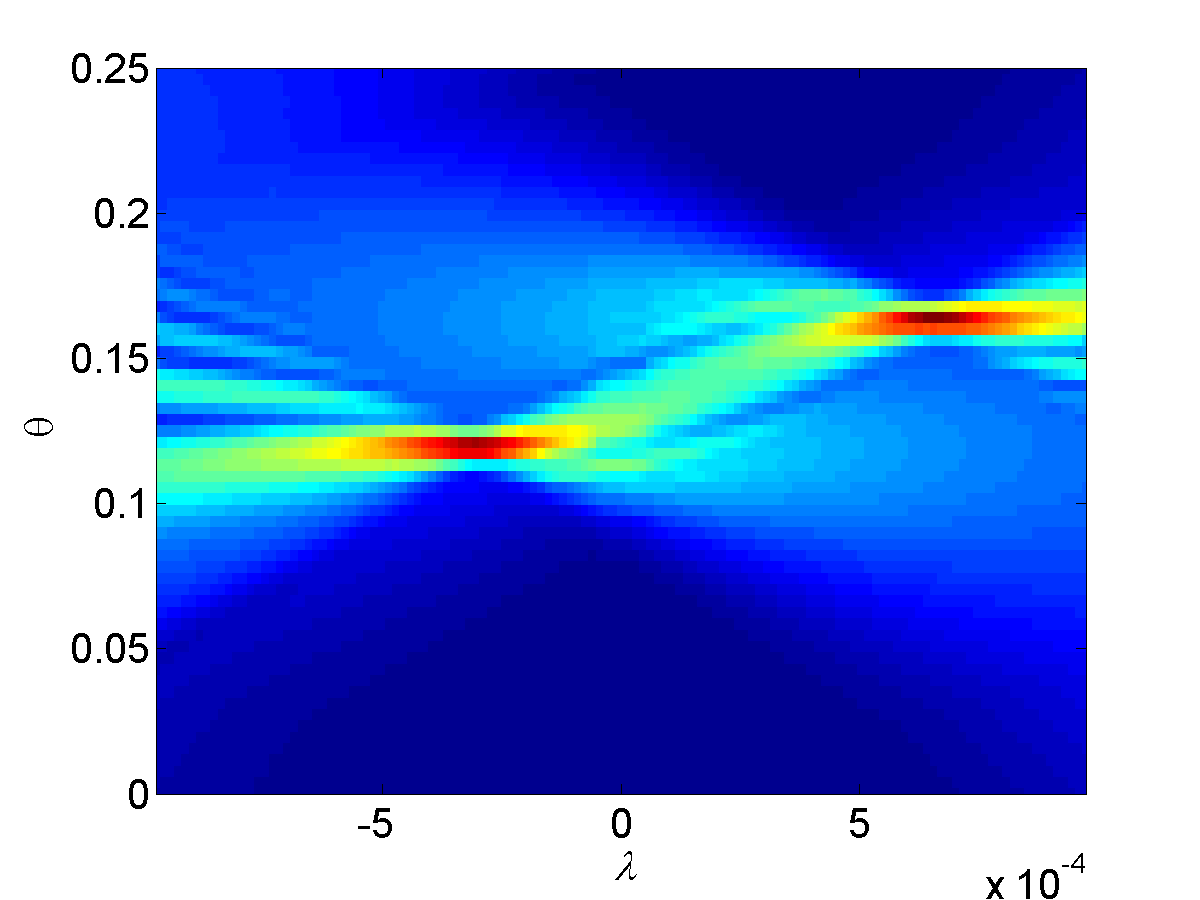}} \quad &
			\resizebox{1.6in}{1.2in}{\includegraphics{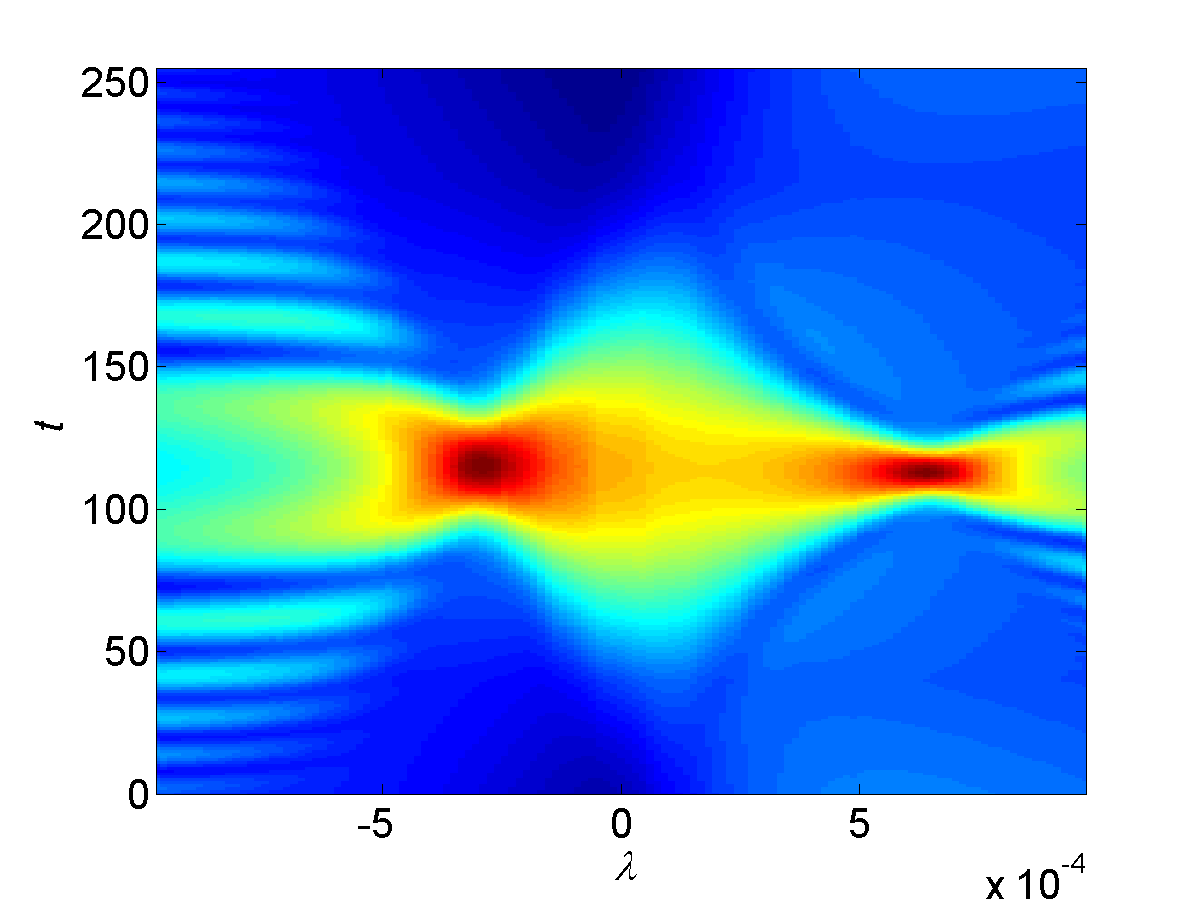}} \quad &
			\resizebox{1.6in}{1.2in}{\includegraphics{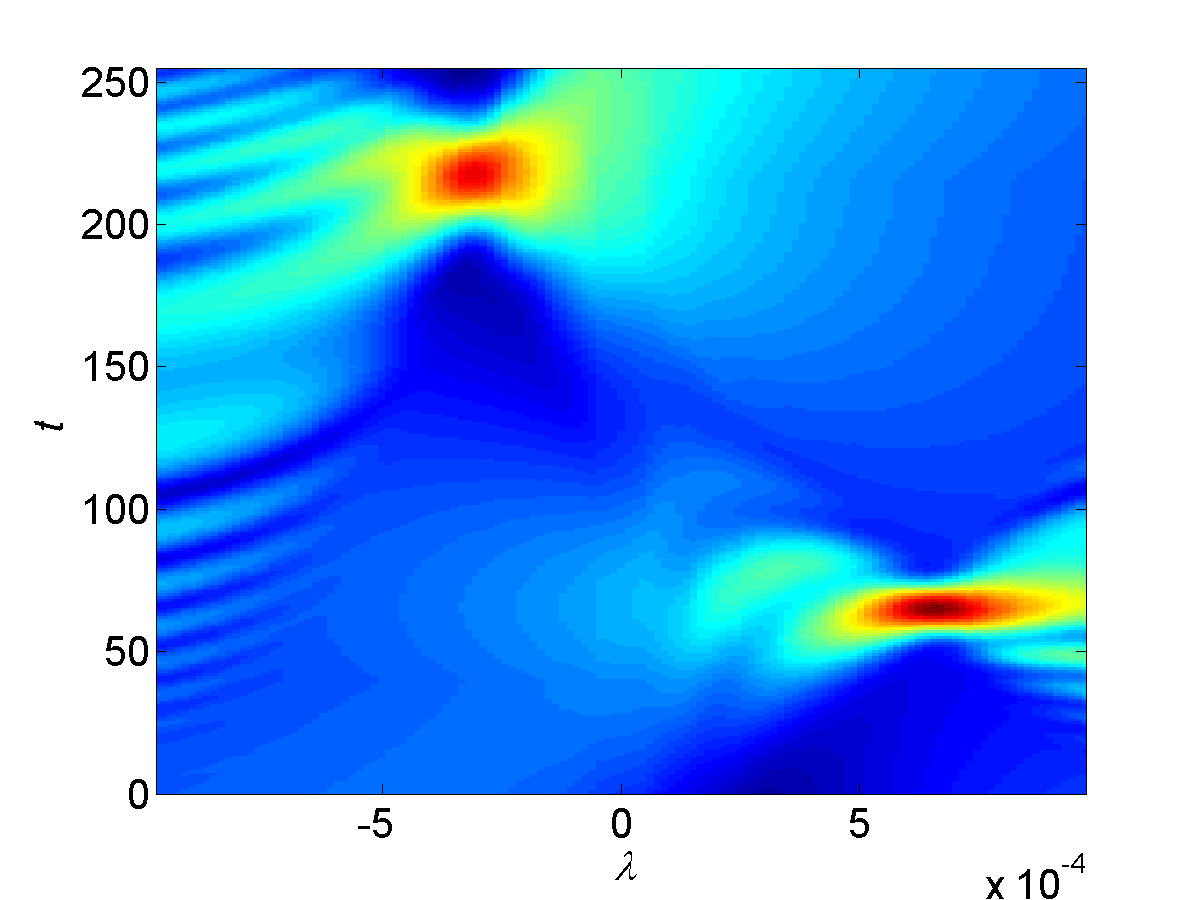}}
		\end{tabular}
\caption{\small Slices of filter-matched CT defined by \eqref{def_FLQFT} corresponding to the slices in Figure \ref{Fig:two_component_SSOmatrix}.
			Top row (from left to right): $|\mathcal F^h(\fs_s)(t,\eta,r_1)|$, $|\mathcal F^h(\fs_s)(t,\eta,r_2)|$ and $|\mathcal F^h(\fs_s)(114,\eta,\lambda)|$;
			Bottom row  (from left to right):  $|\mathcal F^h(\fs_s)(160,\eta,\lambda)|$,  $|\mathcal F^h(\fs_s)(t,34/N,\lambda)|$ and $|\fs_s(t,26/N,\lambda)|$.}
		\label{Fig:two_LFM_FSSO}
	\end{figure}	
	%%%%%%%%%%%%%%%%the end of Figure 4 %%%%%%%%%%%%%%%%%%%%%		

Consider again the %two-component LFM 
signal in \eqref{two_component_lfm}, Figure \ref{Fig:two_LFM_FSSO} shows the slices of the 
filter-matched CT defined by \eqref{def_FLQFT} with Gaussian window and $b=20$ (discrete points, unitless). Observe that by comparing the corresponding pictures in Figure \ref{Fig:two_component_SSOmatrix} and Figure \ref{Fig:two_LFM_FSSO}, the filter-matched CT indeed improves the separability of the two components when their IFs are crossover.

For $x\in \mathcal A_\ga$, we will use the filter-matched CT  to estimate $\phi^\gp_\ell(t), \phi^{\gp\gp}_\ell(t)$ by
\begin{equation}
\label{IF_estimate_filter}
\phi^\gp_\ell(t)\approx \wh \eta_\ell^h(t), \phi^{\gp\gp}_\ell(t)\approx \wh \gl_\ell^h(t)
\end{equation}
with  
\begin{equation}
\label{IF_estimate_FSSO}
(\wh \eta^h_\ell(t), \wh \lambda^h_\ell(t))= {\rm argmax}_{\eta,\lambda \in \cG_\ell(t)}\mathcal F^h(\fs_x )(t,\eta,\lambda),
\end{equation}
where $\cG_\ell(t)$ is defined in Theorem \ref{theo:LQFT}. We call $(t, \wh \eta^h_\ell(t), \wh \lambda^h_\ell(t))$ the ridge of the filter-matched CT.

With the resulting $\wh \eta_\ell^h(t), \wh \lambda_\ell^h(t)$, we may use the CT (namely \eqref{MSSO_recon_comp} or \eqref{MSSO_recon_real}) to recover sub-signal $x_\ell(t), 1\le \ell\le K$. Observe that the recovering formula \eqref{MSSO_recon_comp} or \eqref{MSSO_recon_real} recovers $x_\ell(t)$ one by one.
Next, by considering all $x_\ell(t), 0\le \ell \le K$ as a whole group,
we propose an innovative algorithm based on the filter-matched CT to recover sub-signals $x_\ell(t)$.

Recall that when $g$ is the Gaussian window given in \eqref{def_g}, then \eqref{MSSO_LFM} holds. Thus for $x\in \mathcal A_\ga$, we have
\begin{equation*}
\fs_{x}(t,\eta,\lambda)
\approx \sum_{\ell=0}^K  x_\ell(t) \mathbb{A}(\lambda-\phi^{\gp\gp}_\ell(t))
		\Omega(\lambda -\phi^{\gp\gp}_\ell(t),\eta-\phi^\gp_\ell(t)).
\end{equation*}
Hence if $\wh \eta^h_\ell(t), \wh \lambda^h_\ell(t), 1\le \ell\le K$ obtained from the filter-matched CT by \eqref{IF_estimate_FSSO} are good approximations to
 $\phi^\gp_\ell(t), \phi^{\gp\gp}_\ell(t)$, then, with $\wh \eta_0=\wh \gl_0=0$, we have
 $$
\fs_{x}(t,\eta,\lambda)
\approx \sum_{\ell=0}^K  x_\ell(t) \mathbb{A}(\lambda-\wh \gl_\ell^h(t))
		\Omega(\lambda -\wh \gl_\ell,\eta-\wh \eta_\ell^h(t)).
 $$
 In particular, with $\eta=\wh \eta_m^h(t), \gl=\wh \gl_m^h(t)$,
 we have
 \begin{equation}
 \label{group_est}
\fs_{x}(t,\wh \eta_m^h(t), \wh \lambda_m^h(t))
\approx \sum_{\ell=0}^K  a_{m, \ell} x_\ell(t),  \quad m=0, 1, \cdots, K,
 \end{equation}
 where
 \begin{equation}
 \label{a_lm}
 a_{m,\ell}=\mathbb{A}\big(\wh \gl_m^h(t)-\wh \gl_\ell^h(t)\big)\;
		\Omega\big(\wh \gl_m^h(t) -\wh \gl_\ell^h(t), \wh \eta_m^h(t)-\wh \eta_\ell^h(t)\big).
 \end{equation}
Denote 
\begin{equation}
\label{def_A}
\mathbf{A}_t %=[a_{m, \ell}]_{0\le m, \ell \le K}
=\begin{bmatrix}
					1 & a_{0,1} & \cdots & a_{0,K} \\
					a_{1,0} & 1 & \cdots & a_{1,K} \\
					\vdots  & \vdots  & \ddots & \vdots  \\
					a_{K,0} & a_{K,1} & \cdots & 1
					\end{bmatrix}, 
					\quad 
					\mathbf{X}_t =
					\begin{bmatrix}
					x_0(t)\\
					x_1(t)\\
					\vdots\\
					x_K(t)
					\end{bmatrix}, \quad 					
					\wh{\mathbf{X}}^h_t =
					%[\wh x_0(t),\wh x_1(t), \cdots,\wh x_K(t)]^T
					\begin{bmatrix}
					\fs_x\big(t,\wh \eta_0^h(t), \wh \lambda_0^h(t)\big)\\
					\fs_x\big(t,\wh \eta_1^h(t), \wh \lambda_1^h(t)\big)\\
					\vdots\\
					\fs_x\big(t,\wh \eta_K^h(t), \wh \lambda_K^h(t)\big)
					\end{bmatrix}. 
\end{equation}
 Then \eqref{group_est} can be written as 
 $$
 \wh{\mathbf{X}}^h_t \approx \mathbf{A}_t  \mathbf{X}_t.  
 $$
 Hence 
  \begin{equation}
  \label{group_relation}
   \mathbf{X}_t  \approx \mathbf{A}_t^{-1}\;  \wh{\mathbf{X}}^h_t. 
\end{equation}
Based \eqref{group_relation}, we propose the following algorithm, called the group filter-matched chirplet transform-based signal separation scheme (GFCT3S), to recover components $x_\ell(t)$ and trend $x_0(t)$.

 \begin{alg}
\label{alg:Recov_cross_SLQFT}
{\rm (Mode retrieval with GFCT3S)}. Let $x(t)$ be a multicomponent signal $x\in \mathcal A_\ga$ satisfying \eqref{def_sep_cond_cros} and $\fs_x (t,\eta,\lambda)$ be the CT of $x$ with the Gaussian window function.
\begin{itemize}
\setlength{\itemindent}{1.5em}
\item[{\rm Step 1.}] Set $\wh \eta_0^h=\wh \gl_0^h=0$. Calculate $\wh \eta_\ell^h, \wh \lambda_\ell^h, 1\le \ell\le K$ by \eqref{IF_estimate_FSSO} with the filter-matched CT. 
%where $\cG_\ell(t)$ is defined in Theorem \ref{theo:LQFT}.

\item[{\rm Step 2.}]
Calculate
			\begin{equation}
			\label{MSSO_recon_cross}
			\wt {\mathbf{X}}_t= \mathbf{A}^{-1}_t \; \wh{\mathbf{X}}^h_t, 
%		= \begin{bmatrix}					1 & a_{0,1} & \cdots & a_{0,K} \\
%					a_{1,0} & 1 & \cdots & a_{1,K} \\
%					\vdots  & \vdots  & \ddots & \vdots  \\
%					a_{K,0} & a_{K,1} & \cdots & 1
%					\end{bmatrix} ^{-1} \wh{\mathbf{X}}^h_t,
			\end{equation}
where $\mathbf{A}_t$ and $\wh{\mathbf{X}}^h_t$  are defined by  \eqref{def_A}. 
%$a_{m, \ell}$ are defined by \eqref{a_lm} and $\mathbf{A}=[a_{m, \ell}]_{0\le m, \ell \le K}$.

\item[{\rm Step 3.}] The components $\wt x_\ell(t) $ of $\wt{\mathbf{X}}_t =: [\wt x_0(t), \wt x_1(t), \cdots, \wt x_K(t)]^T$ are the recovered sub-signals $x_\ell(t)$ of $x(t)$.
		\end{itemize}
		
		\end{alg}

We first apply the CT to the blind source signal $x(t)$, with output $\fs_x(t,\eta,\lambda)$ for thresholding with the appropriate parameter $\mu$ that can be determined by the energy distribution of 
$|\fs_x(t,\eta,\lambda)|$. This allows us to separate the thresholded set $\{(\eta,\lambda)  : |\fs_x(t, \eta,\lambda)| \geq \mu/2\}$ into different clusters $\mathcal {G}_{\ell}(t)$ $(\ell=0,\cdots, K)$. Note that $\Omega(\gl,\eta)$ is a fast decay function with respect to $\eta$. 
%If the IF of sub-signal $x_m(t)$ is well-separated  from other signals' IF (that is \eqref{def_sep_cond} holds), then we have $a_{\ell,m}\approx 0$ for $\ell\ne m$. Furthermore, 
If all the sub-signals $x_k(t)$ are well-separated in the time-frequency plane, namely satisfy \eqref{def_sep_cond}, then $a_{\ell,m}\approx 0$ for $\ell\ne m$, and hence, $\mathbf{A}_t$ in \eqref{MSSO_recon_cross} 
is essentially the identity matrix.
So the reconstruction algorithm above is fit for both cases in \eqref{def_sep_cond} and \eqref{def_sep_cond_cros}.

Finally, we discuss the algorithm of real-time processing of an consecutive input signal. To reduce the computational cost, we need to predefine some variables. Let $T\in \ZZ^+$ denote the truncation length of the input signal, e.g. $T=128$ or $T=256$, then frequency $\eta$ is  discretized into $\eta = 0,1/T,...,(T/2-1)/T$ and $\lambda$ is  discretized into $\lambda = -(T/2-1)/T^2,-(T/2-2)/T,...,0,...(T/2-1)/T^2$. Define ${\bf S}^{t_m}_{T/2 \times (T-1)} = \fs_x (t_m,\cdot,\cdot) $, where $t_m$ is a fixed time. Suppose we will separate $x$ when $t = t_m$, then we need to calculate ${\bf S}^{t_m-L_{m}},\cdots,{\bf S}^{t_m+L_{m}}$ first, where $2L_{m}+1$ should equivalent to the window length $b$ in \eqref{def_FLQFT} and  $2L_{m}+1 \ll T$. Then for the next time instant $t = t_m+1$, we just need to calculate ${\bf S}^{t_m+1+L_{m}}$, and continue the procedure.

	\begin{figure}[H]
	\centering
	\begin{tabular}{cc}
		\resizebox{3.14in}{1.2in}{\includegraphics{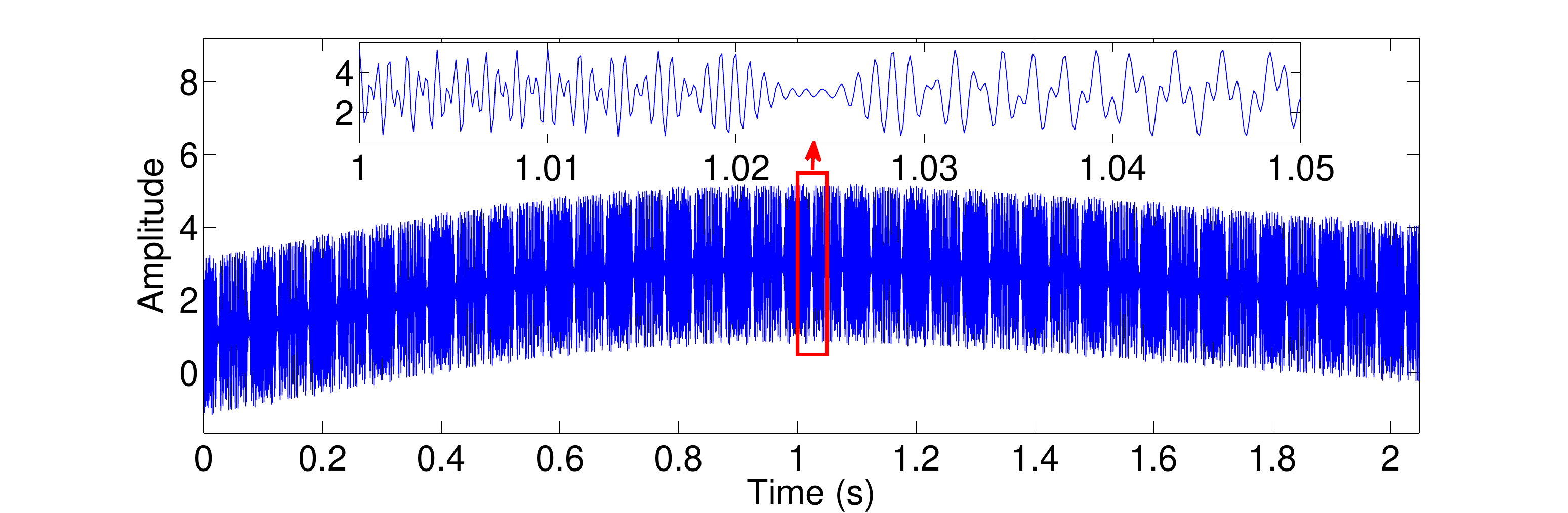}}
		\resizebox{3.14in}{1.2in}{\includegraphics{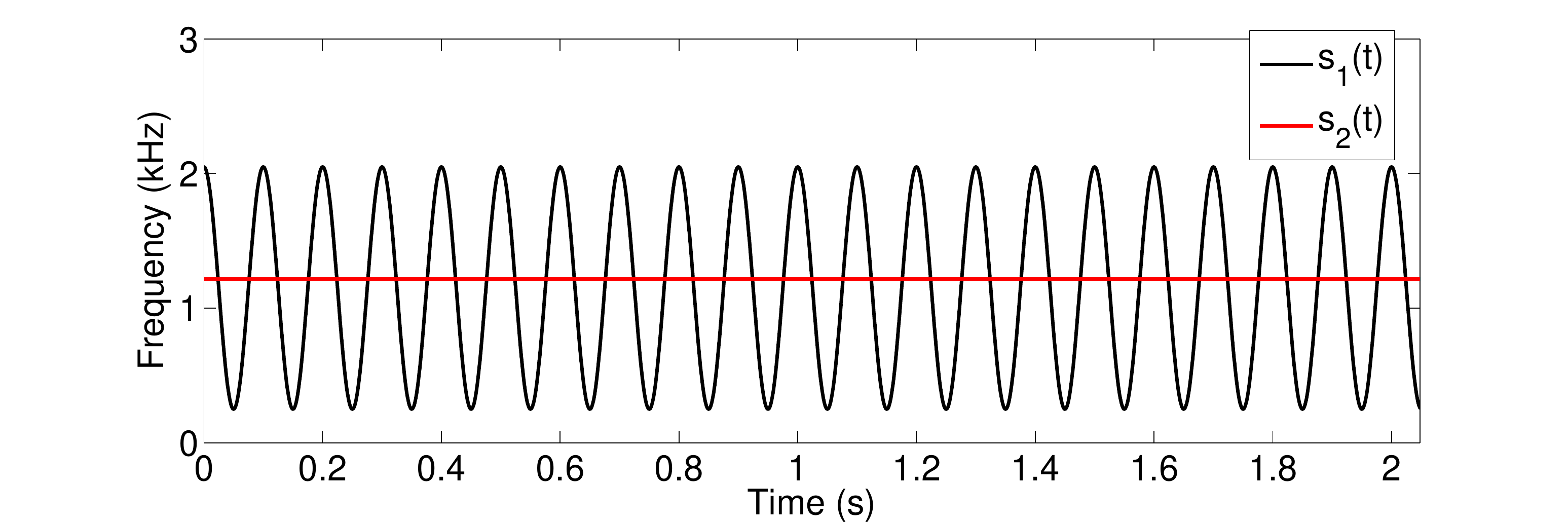}} 		
	\end{tabular}
	\caption{\small Waveform of $s(t)$ in \eqref{s_t} with enlarged picture around $t=1$ and IFs of two AM-FM components. }
	\label{Fig:exp1}
\end{figure}

\begin{figure}[H]
	\centering
	\begin{tabular} {cc}  %{c@{\hskip -0.9cm}c}
		\resizebox{3.14in}{1.2in}{\includegraphics{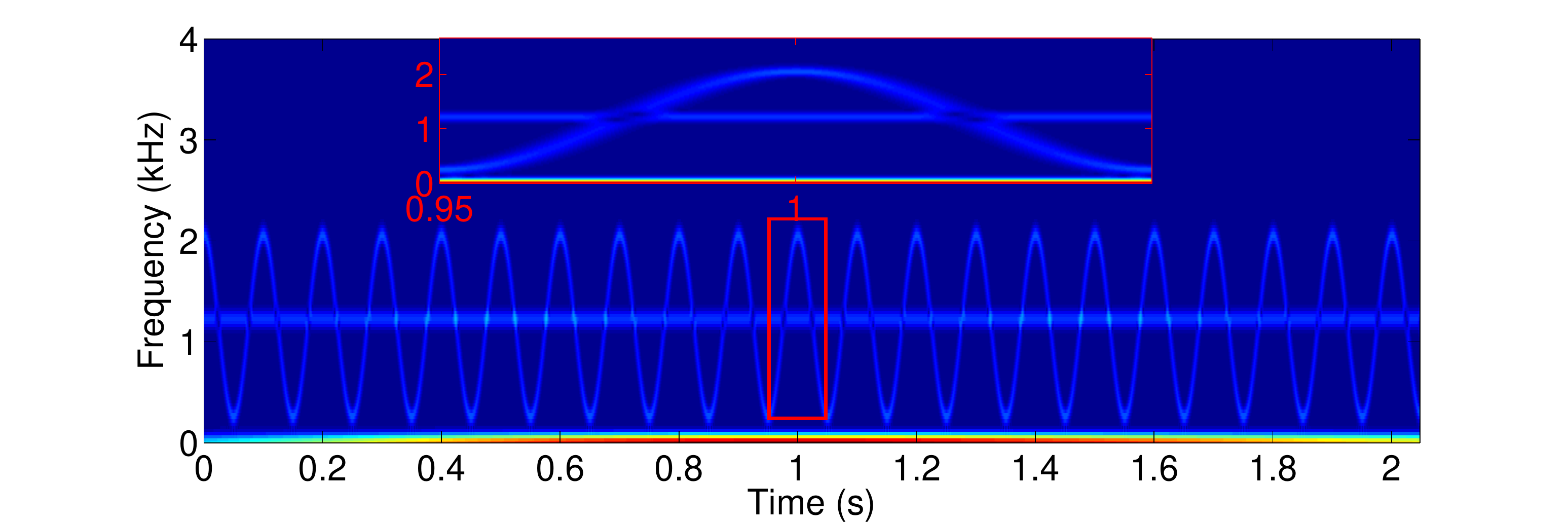}}
		\resizebox{3.14in}{1.2in}{\includegraphics{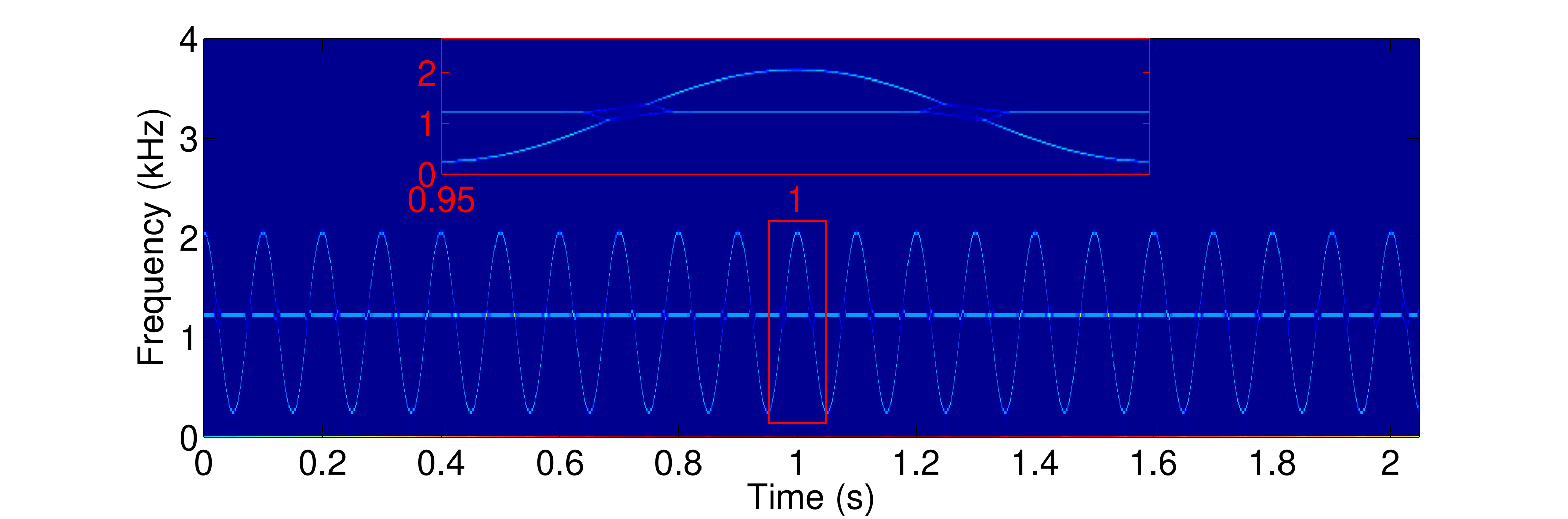}} \\
		\resizebox{3.14in}{1.2in}{\includegraphics{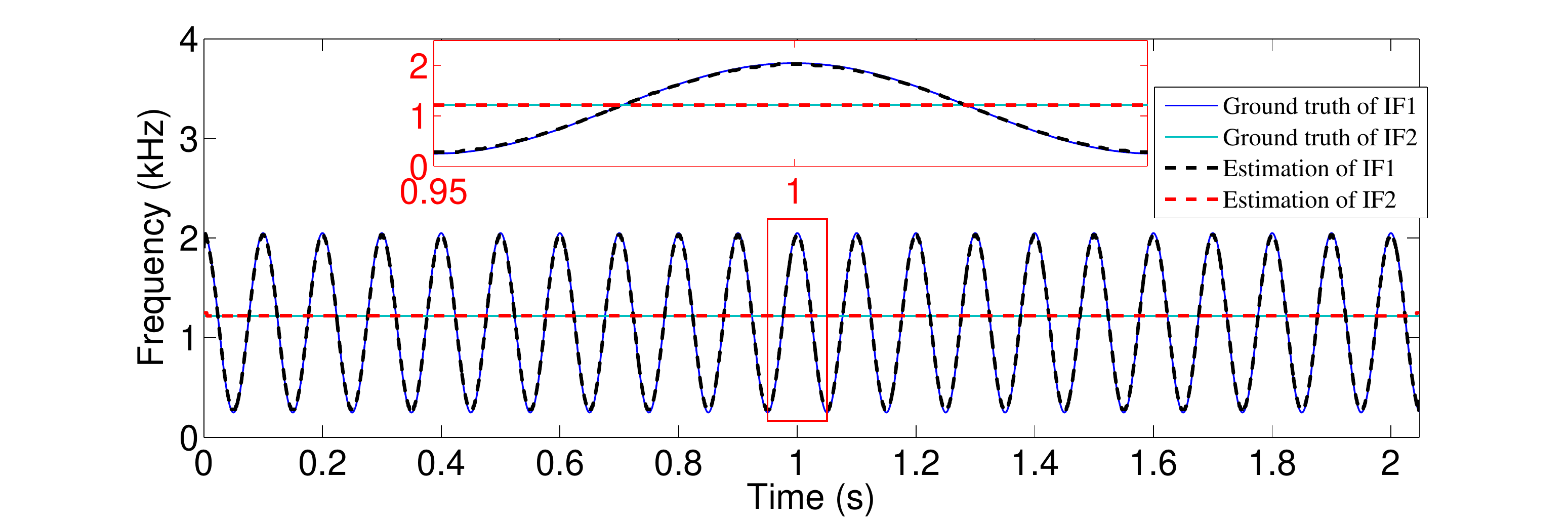}}		
	\end{tabular}
	\caption{\small Recovered IFs of $s(t)$ in \eqref{s_t} with SST and ridge of the filter-matched CT (with enlarged picture around $t=1$). Top row (from left to right): Time-frequency diagram of STFT, time-frequency diagram of the 2nd-order SST in \cite{MOM15}; Bottom row: Estimated IFs by ridges of the filter-matched CT given in \eqref{IF_estimate_FSSO}}.
	\label{Fig:compar_SST_MSSO}
\end{figure} 

\section {Numerical experiments}

Figure \ref{Fig:two-component} demonstrates our proposed method GFCT3S (Algorithm 1) 
is efficient for the two-component signal in
\eqref{two_component} with one cross point of the IFs.
In this section, we first consider another synthetic multicomponent signal $s(t)$, given as
\begin{equation}\label{s_t}
\begin{array}{l}
s(t) =s_1(t) +s_2(t) +A_0(t) = 1.2\cos(2300\pi t +90 \sin(20\pi t)) + \cos(2438\pi t)+(1+(t^2+t)e^{1-t^{1.5}}),
\end{array}
\end{equation} 
where $t>0$. Note that the IFs of $s_1(t)$ and $s_2(t)$ are $\phi'_1(t)=1150+900\cos(20\pi t)$ and $\phi'_2(t)=1219$, respectively, called IF1 and IF2 in Figure \ref{Fig:compar_SST_MSSO}.

In this experiment, we discretize $s(t)$ with a sampling rate $F_s = 8$kHz and data length $M= 2^{14}$, namely $t\in [0, 2.048]$. Figure \ref{Fig:exp1} shows the waveform of  $s(t)$ and the IFs of two AM-FM components, $s_1(t)$ and $s_2(t)$. As the expression in \eqref{s_t}, $s(t)$ consists of one trend and two oscillating AM-FM components. Meanwhile, the IFs of the two AM-FM components are crossover.
	
Observe that the signal $s(t)$ to be processed contains lots of samples, which should be analyzed or separated locally. Here we use a sliding truncated Gaussian window with length of $N=2^8$ points for methods of STFT, SST and GFCT3S. Hence the frequency bins is discretized as $\frac{F_s}{N}\{-N/2+1,-N/2+2,...,N/2-1,N/2  \}$ for complex signals, or just $\frac{F_s}{N}\{0,1,...,N/2-1,N/2  \}$ for real signals. Note that we set $N$ as a power of 2 to take full advantage of the fast Fourier transform.
Figure \ref{Fig:compar_SST_MSSO} shows the time-frequency diagrams of the STFT and the 2nd-order SST in \cite{MOM15} and the recovered IFs by the ridges of the filter-matched CT given in \eqref{IF_estimate_FSSO}.
%some of the experimental results of SST and GFCT3S. 
The enlarged pictures around $t=1$ are also attached with each sub-Figure. Obviously, when the IFs of $s_1(t)$ and $s_2(t)$ are crossover, either the STFT or the 2nd-order SST  hardly represents the sub-signals reliably. Thus we cannot use the time-frequency diagram of the STFT or the 2nd-order SST to extract the IFs of the sub-signals for the purpose of recovering their waveforms.
However, using the ridges $\wh \eta_\ell^h(t)$ of the filter-matched CT, we can extract the IFs of $s(t)$ accurately.

\begin{figure}[H]
	\centering
	\begin{tabular} {c}  %{c@{\hskip -0.9cm}c}
		\resizebox{3.3in}{1.5in}{\includegraphics{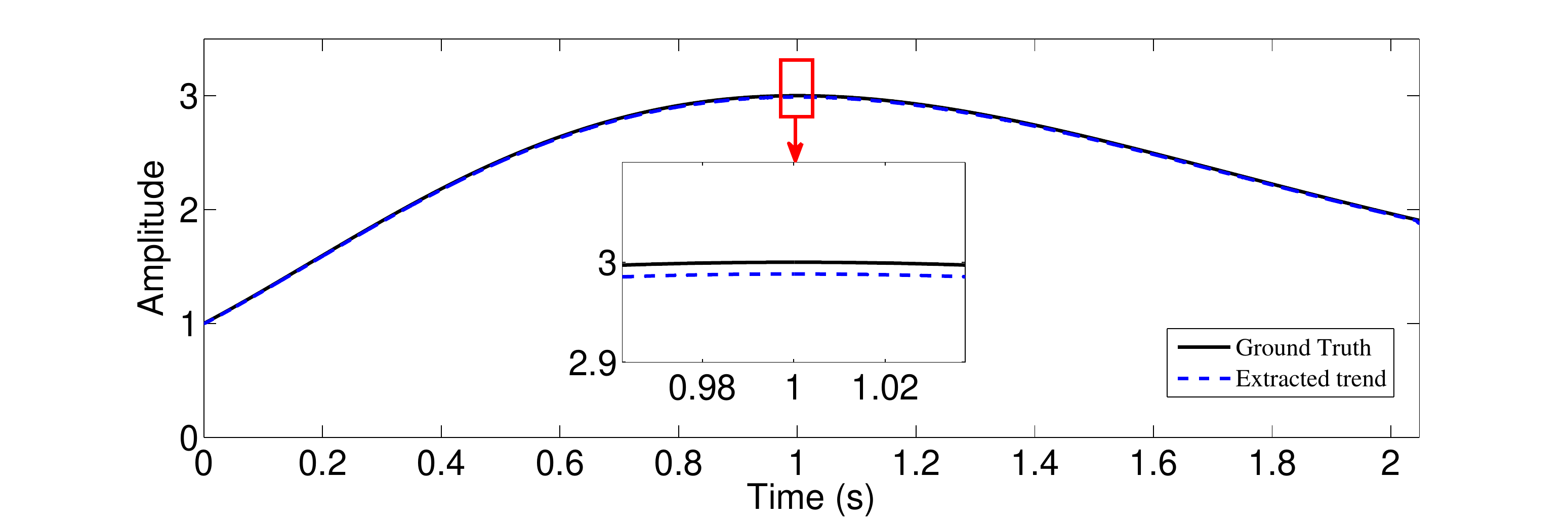}}\\
		\resizebox{3.3in}{1.5in}{\includegraphics{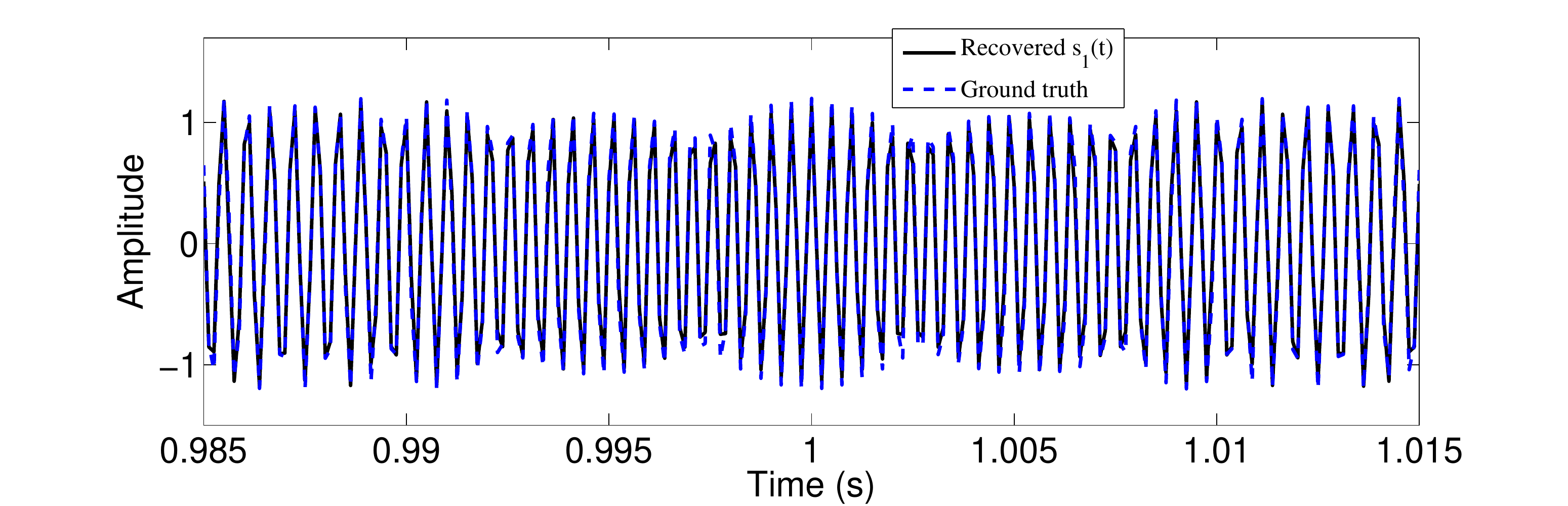}}
		\resizebox{3.3in}{1.5in}{\includegraphics{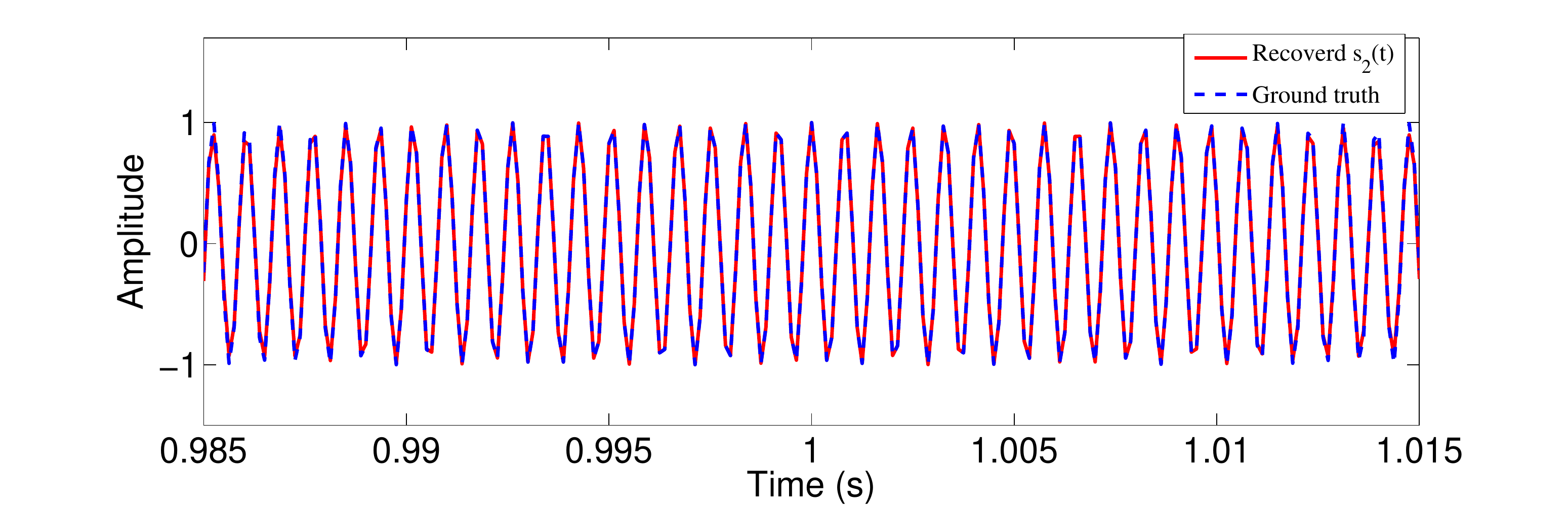}} 	
	\end{tabular}
	\caption{\small Mode recovery results of $s(t)$ in \eqref{s_t} with proposed GFTC3S (Algorithm \ref{alg:Recov_cross_SLQFT}).
	Top: Recovered trend with enlarged picture around $t=1$; Bottom: Recovered modes $s_1(t)$ (Left) and   $s_2(t)$ (Right) around $t=1$.}
	\label{Fig:recovery_s(t)}
\end{figure}
\begin{figure}[H]
	\centering
	\begin{tabular} {cc}  %{c@{\hskip -0.9cm}c}
		\resizebox{2.5in}{1.8in}{\includegraphics{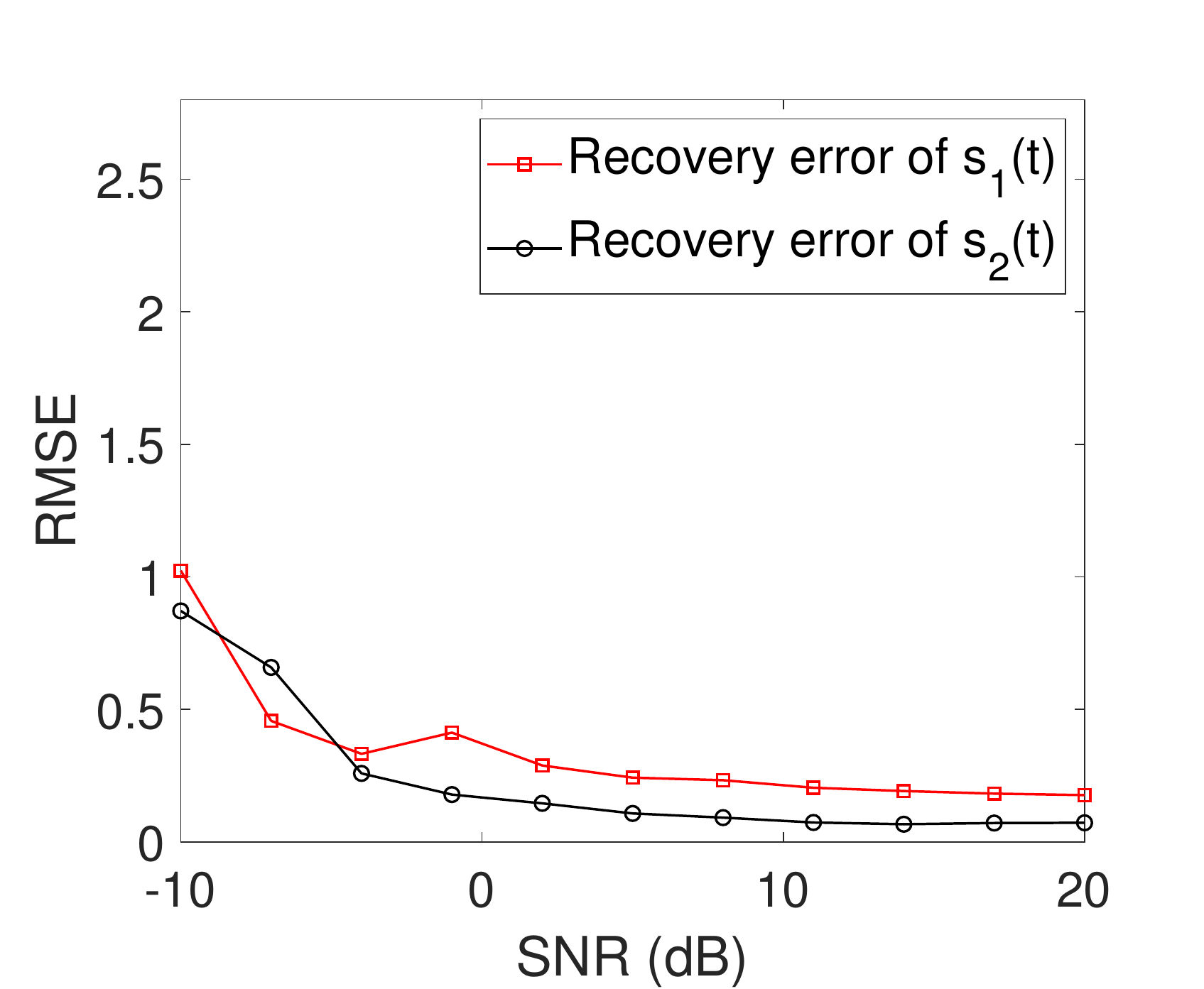}}
		\quad &\quad \resizebox{2.5in}{1.8in}{\includegraphics{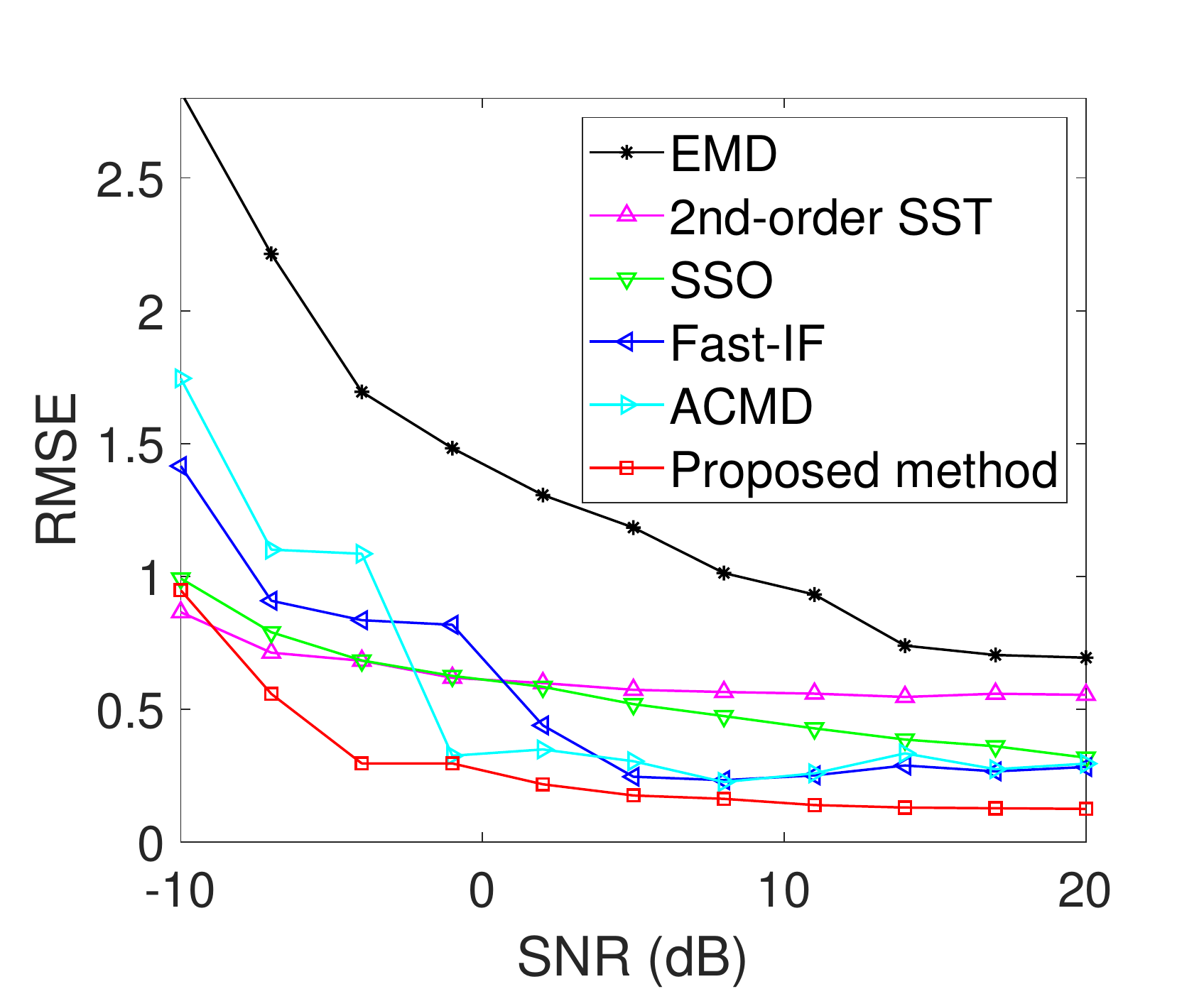}}
	\end{tabular}
	\caption{\small Recovery errors of $s(t)$ under different SNRs.
		Left:  Recovered errors of $s_1(t)$ and $s_2(t)$ by proposed GFTC3S; Right: Comparison with various methods.}
	\label{Fig:rmse_snr_s(t)}
\end{figure}

Figure \ref{Fig:recovery_s(t)} provides recovered modes $s_1(t)$ and $s_2(t)$ around $t=1$ with GFTC3S (Algorithm \ref{alg:Recov_cross_SLQFT}) proposed in Section 3.2. Since there are so many sample periods of the signal $s(t)$, we just show a small truncation around $t=1$. Observe the differences between the recovered waveforms and the truth ones are extremely small.
%Since the existing signal decomposition methods based on EMD or SST cannot solve the IF crossover problem,  we will not provide  the recovery results of those methods here.

To demonstrate the efficiency of our computational scheme for signal data with additive noise, we compared the proposed algorithm with the following methods: EMD \cite{Huang98}, 2nd-order SST \cite{MOM15}, SSO \cite{Chui_Mhaskar15}, Fast-IF \cite{add4}, and ACMD \cite{Over3}. Among
these algorithms, Fast-IF, and ACMD are the ones that take
crossover frequencies into account. To measure the errors, we use the root-mean-square error (RMSE) defined by
\begin{equation*}
E_f = \frac{||f-\wt f||_2}{||f||_2},
\end{equation*}
where $\wt f$ is the estimation of a function $f$. The left panel of Figure \ref{Fig:rmse_snr_s(t)} provides the RMSEs of $E_{s_1}$ and $E_{s_2}$ when signal-to-noise ratios (SNRs) vary from $-10$dB to $20$dB.  Note that the recovery performance tends to be stationary when SNR is larger than 0dB. From the right panel of Figure \ref{Fig:rmse_snr_s(t)}, where the average RMSE is defined by $E_{s} = \frac{E_{s_1}+E_{s_2}}{2} $,  we find that the existing methods EMD and 2nd-order SST and SSO are hardly to recover the sub-signals $s_1(t)$ and $s_2(t)$ with crossover IFs. 
%Fast-IF, ACMD and the proposed algorithm present same performance when 0$<$SNR$<$10, and however 
%Our method outperforms these two methods with low SNR conditions. 
Our method outperforms Fast-IF and ACMD. As shown in this figure, our algorithm posses the best stability with respect to robustness.

\begin{figure}[H]
	\centering
	\begin{tabular} {c} % {c@{\hskip -0.9cm}c}
		\resizebox{1.6in}{1.2in}{\includegraphics{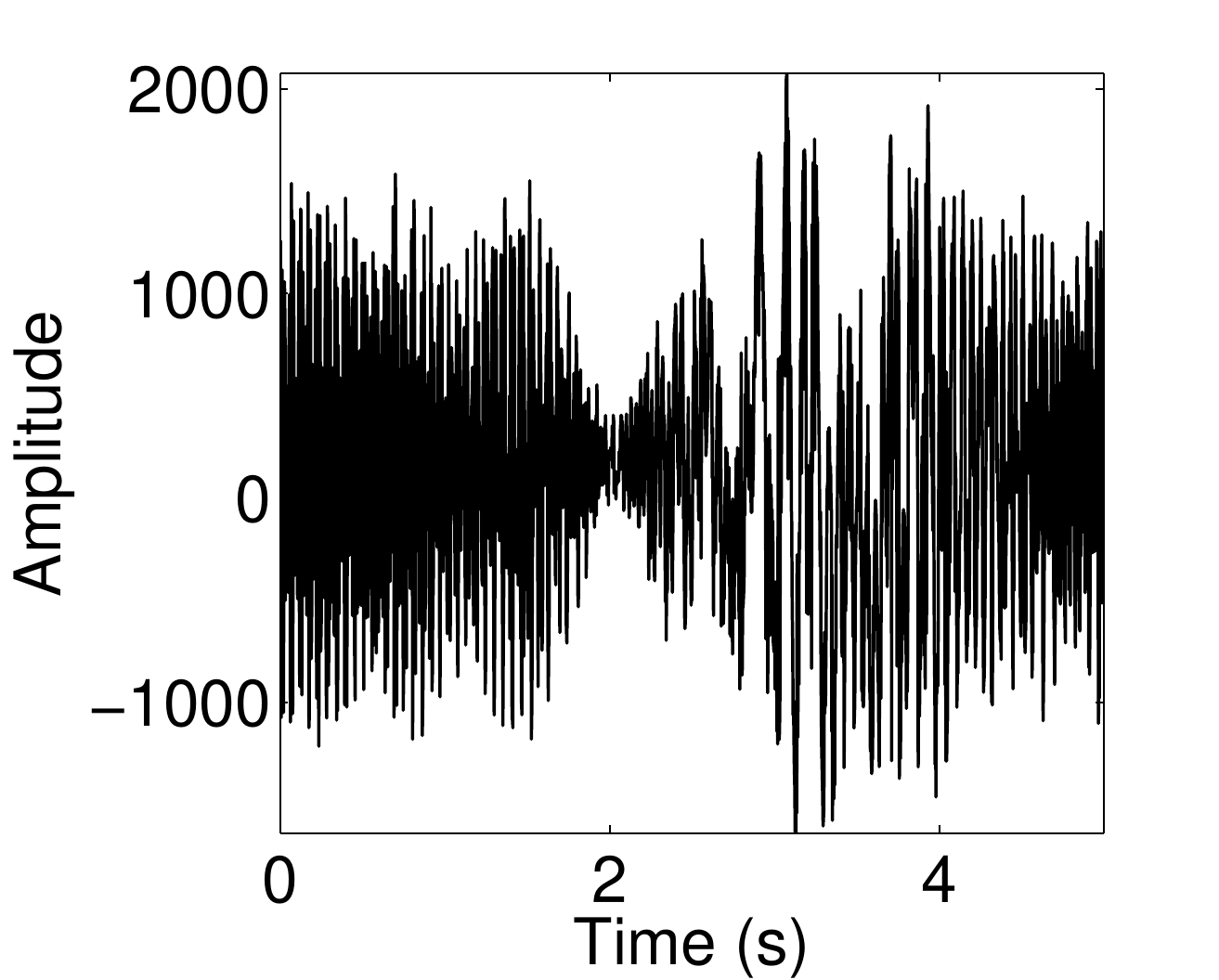}}
		\resizebox{1.6in}{1.2in}{\includegraphics{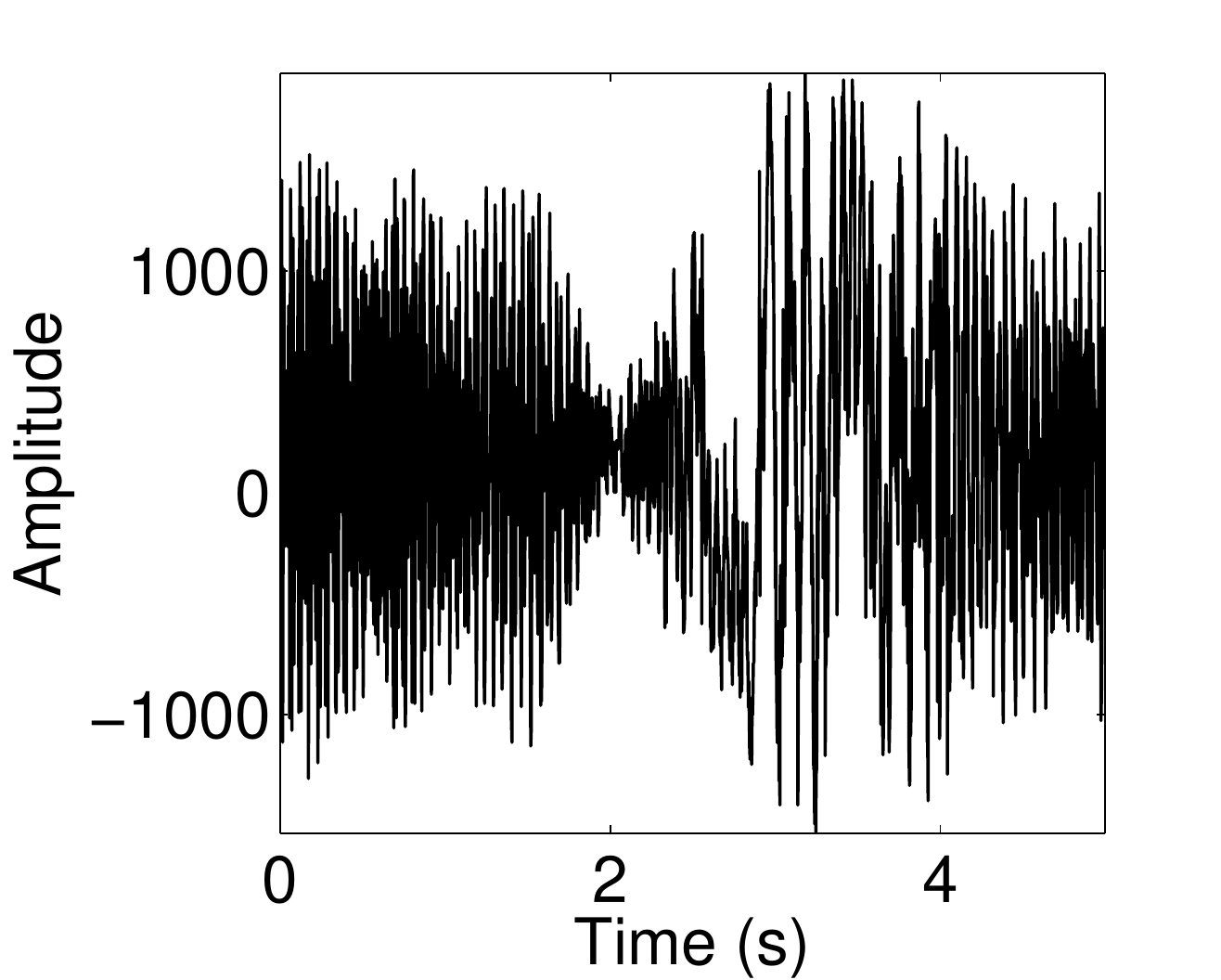}} \\
		\resizebox{3.0in}{1.2in}{\includegraphics{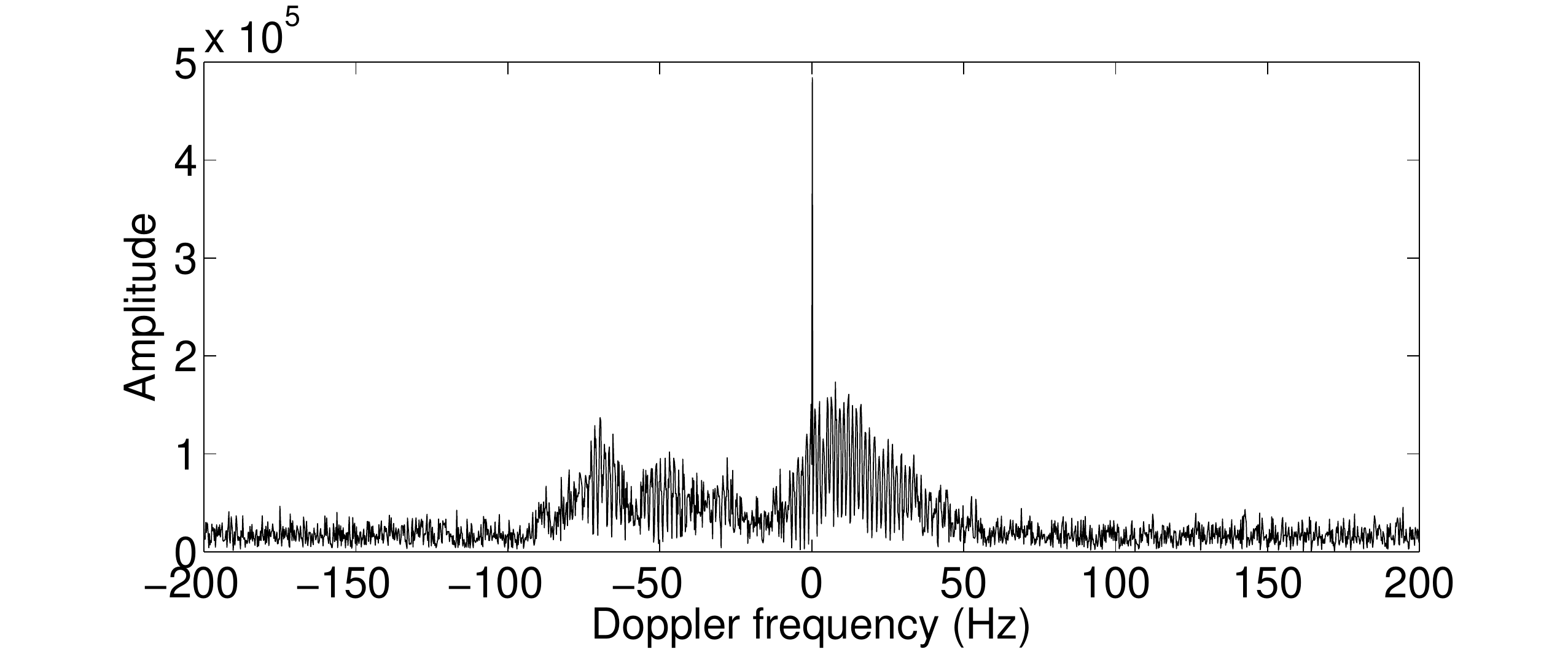}} 	\\
	\end{tabular}
	\caption{\small Waveform and spectrum of the radar echoes. Top-left: Real part of the waveform; Top-right: Imaginary part of the waveform; Bottom: Spectrum. }
	\label{Fig:radar_echoes}
\end{figure}
\begin{figure}[H]
	\centering
	%\begin{tabular}{ccc}
	\begin{tabular}{c@{\hskip -0.0cm}c@{\hskip -0.0cm}c}
		\resizebox{1.7in}{1.3in}{\includegraphics{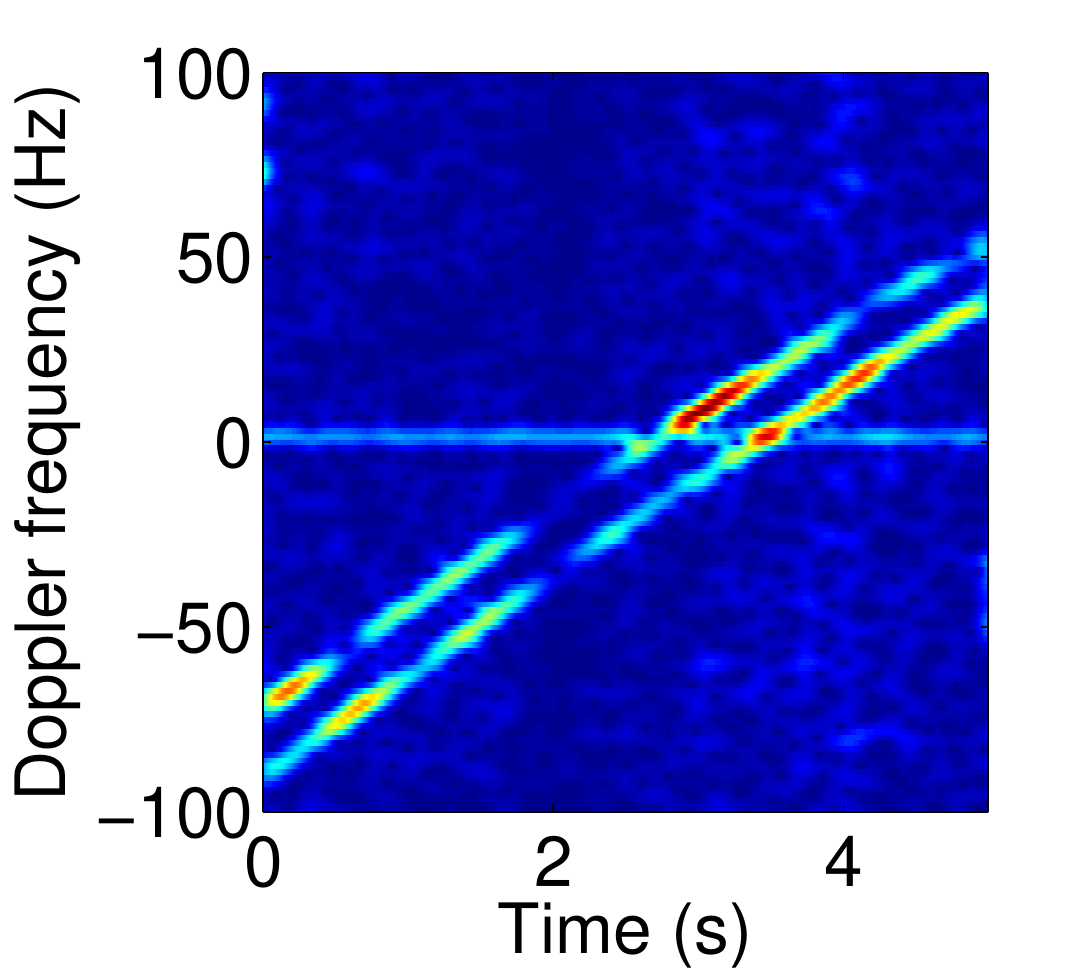}} \quad &
		\resizebox{1.7in}{1.3in}{\includegraphics{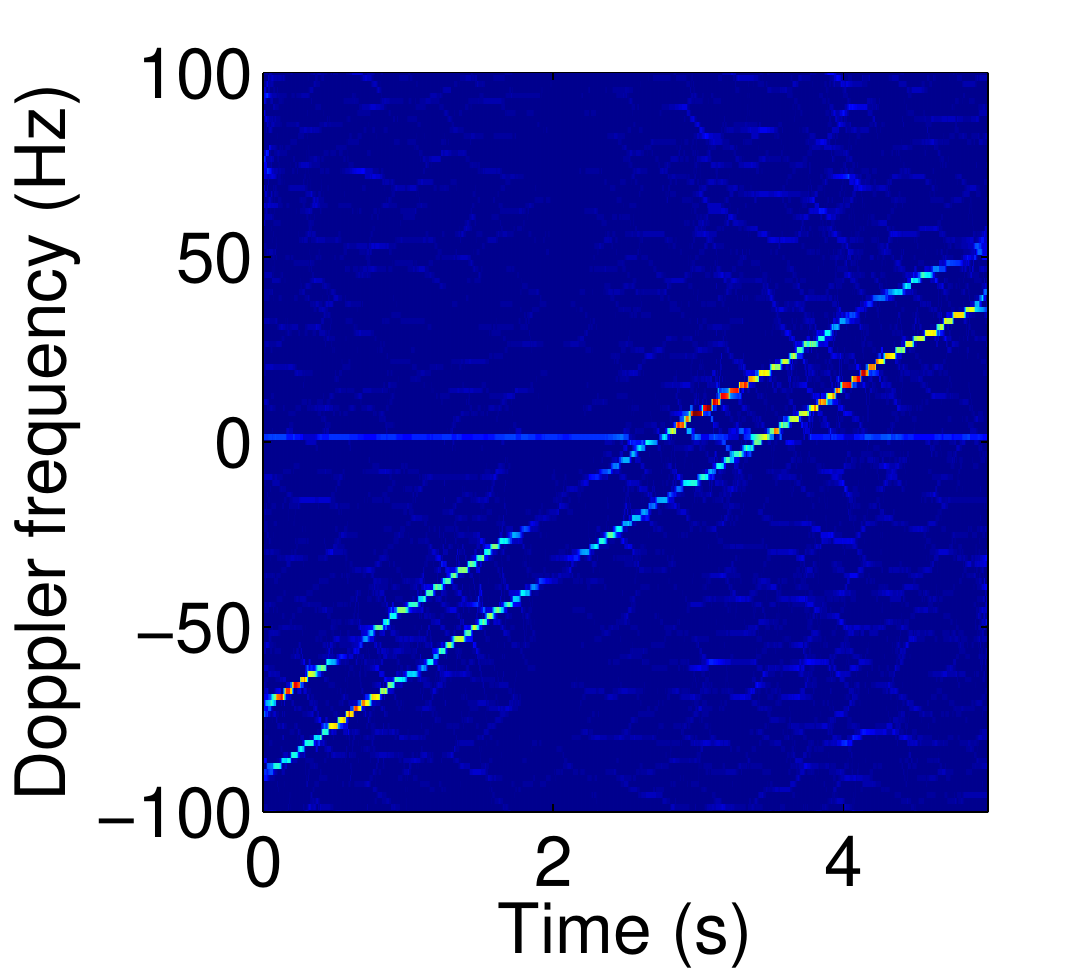}} \quad &
		\resizebox{1.7in}{1.3in}{\includegraphics{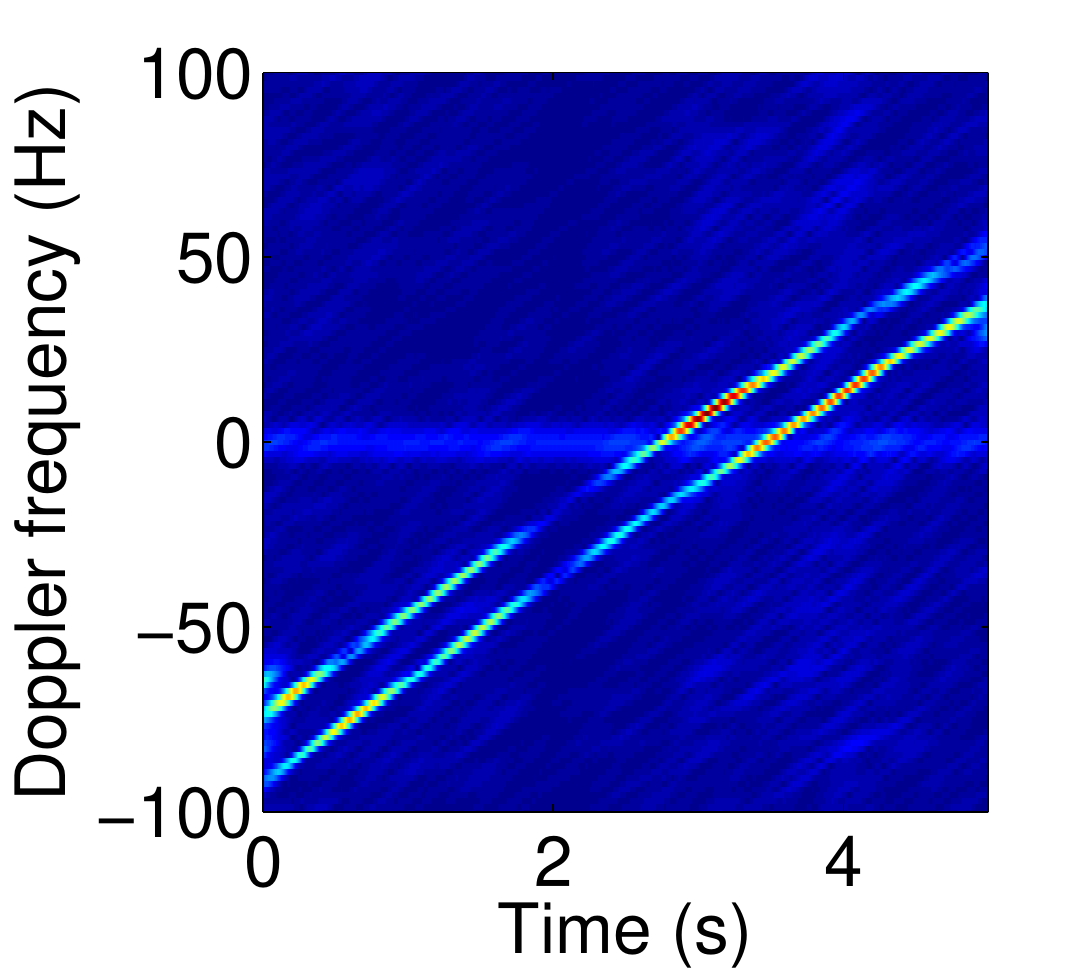}} \\
		\resizebox{1.7in}{1.3in}{\includegraphics{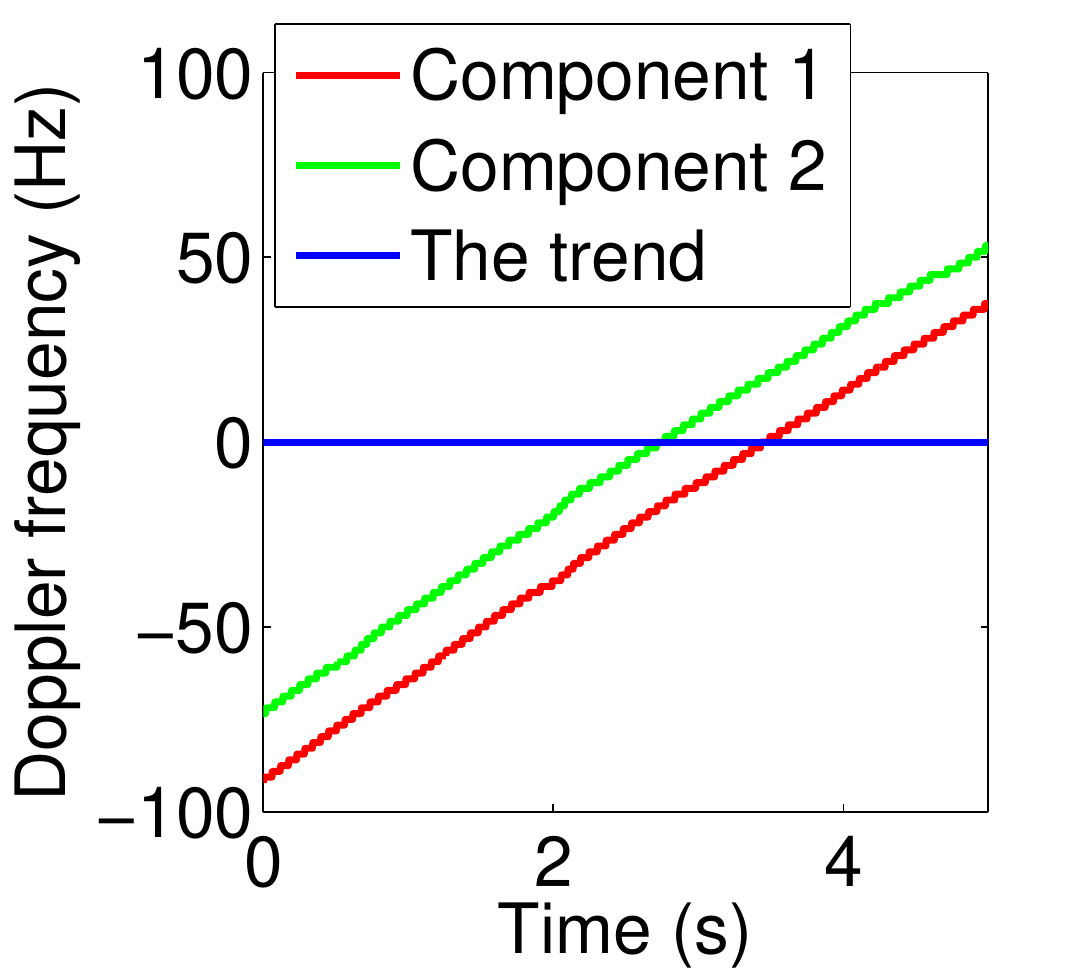}} \quad &
		\resizebox{1.7in}{1.3in}{\includegraphics{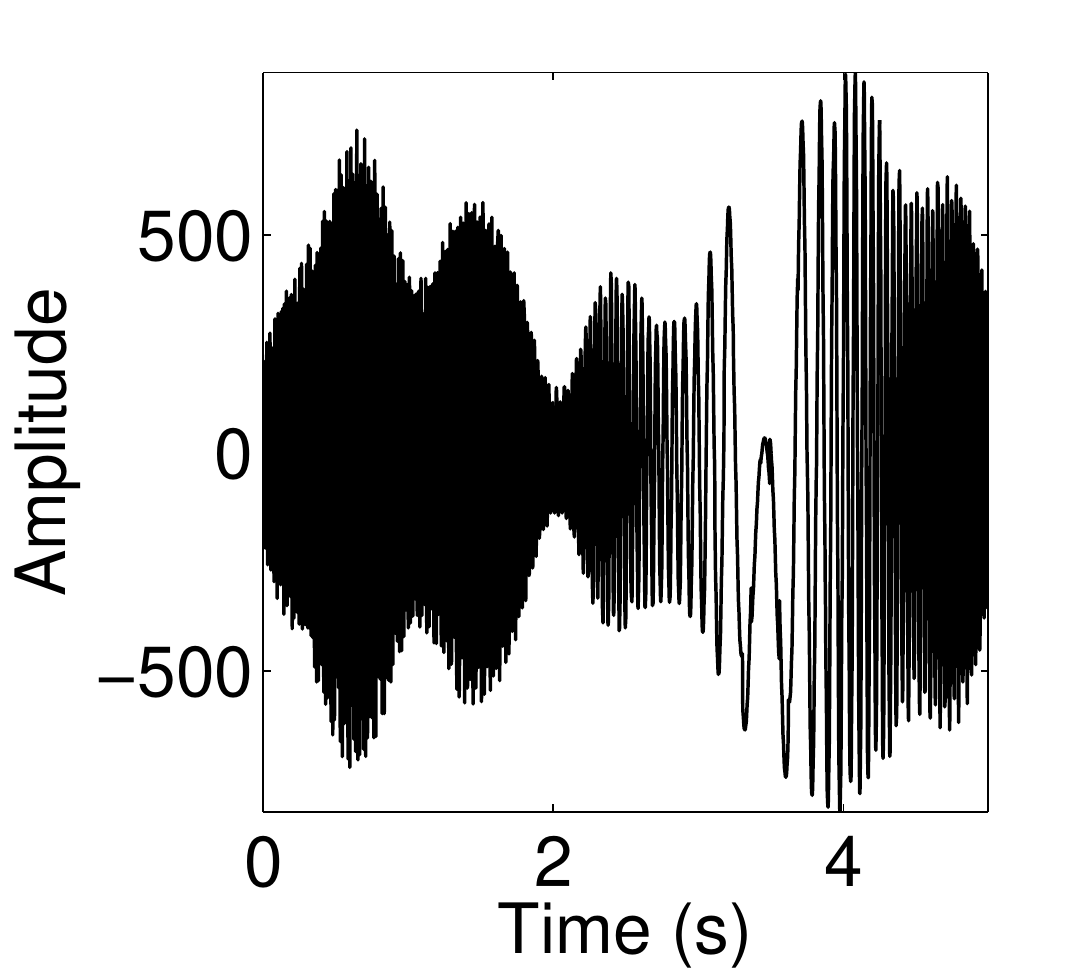}} \quad &
		\resizebox{1.7in}{1.3in}{\includegraphics{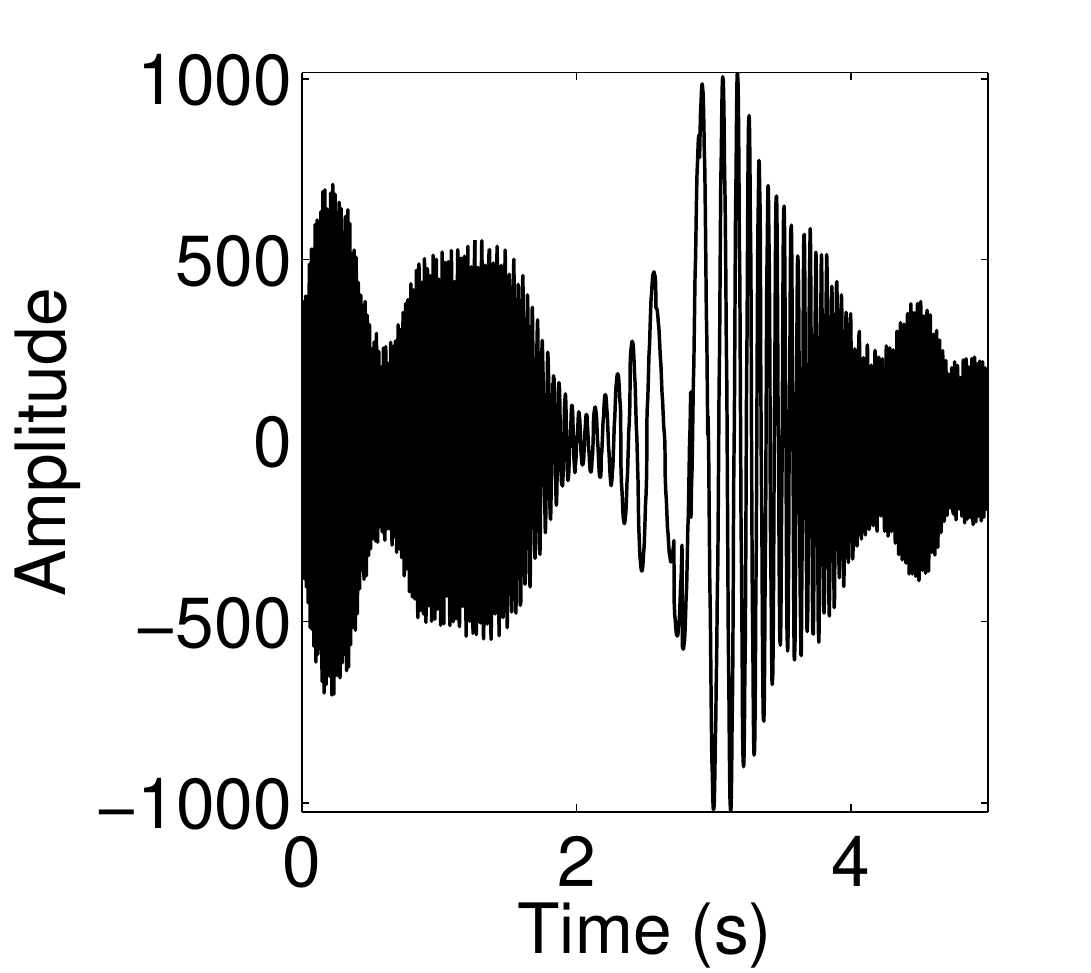}}
	\end{tabular}
	\caption{\small Results of the radar echoes.
		Top row (from left to right): STFT,  2nd-order SST, a specific slice of 3D filter-matched CT matrix;
		Bottom row  (from left to right):  Estimated IFs by ridges of  filter-matched CT given in \eqref{IF_estimate_FSSO}, recovered waveform of Component 1 (real part) and recovered waveform of Component 2 (real part) by GFCT3S. }
	\label{Fig:results_radar_echoes}
\end{figure} 

Finally, we consider a real signal, the radar echoes. Figure \ref{Fig:radar_echoes} shows the waveform and spectrum of the radar echoes. Note that the sampling rate here is equal to 400 Hz, which is just the pulse repetition frequency of the radar. The bottom panel of Figure \ref{Fig:radar_echoes}, namely the spectrum shows that this signal consists of several broadband components and a trend. 

Figure \ref{Fig:results_radar_echoes} shows some results of the radar echoes given in Figure \ref{Fig:radar_echoes}. From their STFTs (top-left panel in Figure \ref{Fig:results_radar_echoes}), there are two components (two radar targets) in the echoes. Meanwhile, the IF curves of these two targets are crossover with the trend component. Although the 2nd-order SST can squeeze the time-frequency plane of STFT, it is still affected by the trend component when they are crossover.
% {\bf  To Lin and Ninging, what\rq{}s purpose for this slice and what\rq{}s meaning of the next sentence of \lq\lq{}Observe that the trend component is weakened in this specific slice.\rq\rq{}?} 
The top-right panel of Figure \ref{Fig:results_radar_echoes} shows a specific slice of 3D filter-matched CT matrix, where the chirp rate is close to those of Component 1 and Component 2. 
Observe that in this specific slice, compared with STFT and 2nd-order SST, the two signal components are enhanced a lot. 
The left-bottom panel shows 
the estimated IFs by the ridges of filter-matched CT given in \eqref{IF_estimate_FSSO}. The recovered waveforms of Component 1 (real part) and Component 2 (real part) by GFCT3S
are provided in the middle and right panels in the bottom row of Figure \ref{Fig:results_radar_echoes}. 
%The results demonstrate the proposed method in this paper is efficient in separating multicomponent with crossover IF curves. {To Lin and Ninging, \bf \lq\lq{}The results\rq\rq{} means IF estimate or component recovery or both? From \lq\lq{}separating multicomponent with crossover IF curves\rq\rq{}, it means component recovery. Since we do not know the ground truth components, how can we tell from the middle and right panels in the bottom row of Fig. \ref{Fig:results_radar_echoes}, our method is efficient? } 
	
\section{Conclusions}
In this paper, we propose a method based the chirplet transform (CT) to retrieve modes of multicomponent signals with crossover instantaneous frequencies. We define the modified adaptive harmonic model and the conditions to represent crossover components separately by the proposed CT-based method.
The error bounds for instantaneous frequency estimation and sub-signal recovery
are provided. Based on the approximation of source signals with linear chirps at any local time, we propose an improved CT-based signal reconstruction algorithm with all signal components taken into account. The numerical experiments demonstrate the proposed method are more accurate and consistent in signal separation than EMD, SST, SSO and other approaches such as Fast-IF and ACMD. The proposed method has a great potential for a variety of engineering applications such as channel detection in communication, fault monitoring in mechanical systems, radar signal processing etc.

\section*{Appendix}

In this appendix we provide the  proof of Theorem \ref{theo:LQFT}.  Write
$$
x(t+\tau)=x_{\rm m}(t, \tau)+ x_{\rm r}(t, \tau),
$$
where
\begin{equation*}
 \begin{array}{l}
 x_{\rm m}(t, \tau)=\sum_{k=0}^K x_k(t)e^{i2\pi (\phi^\gp_k(t)\tau+\frac 12\phi^{\gp \gp}_k(t)\tau^2) },
 \\x_{\rm r}(t, \tau)=\sum_{k=0}^K\Big\{ (A_k(t+\tau)-A_k(t))e^{i2\pi \phi_k(t+\tau)}+x_k(t)e^{i2\pi (\phi^\gp_k(t)\tau+\frac 12\phi^{\gp \gp}_k(t)\tau^2)}\\
\big(e^{i2\pi (\phi_k(t+\tau)-\phi_k(t)-\phi_k^\gp(t) \tau- \frac 12\phi^{\gp \gp}_k(t)\tau^2)}-1\big)\Big\}.
 \end{array}
 \end{equation*}
%\begin{equation*}
%\begin{aligned}
%&x_{\rm m}(t, \tau)=\sum_{k=0}^K x_k(t)e^{i2\pi (\phi^\gp_k(t)\tau+\frac 12\phi^{\gp \gp}_k(t)\tau^2) }, \\
%&x_{\rm r}(t, \tau)=\sum_{k=0}^K\Big\{ (A_k(t+\tau)-A_k(t))e^{i2\pi \phi_k(t+\tau)}\\
%&+x_k(t)e^{i2\pi (\phi^\gp_k(t)\tau+\frac 12\phi^{\gp \gp}_k(t)\tau^2)} \\
%&\big(e^{i2\pi (\phi_k(t+\tau)-\phi_k(t)-\phi_k^\gp(t) \tau- \frac 12\phi^{\gp \gp}_k(t)\tau^2)}-1\big)\Big\}.\\
%\end{aligned}
%\end{equation*}
Denote
\begin{equation}
	\label{def_MSSO_m1}
\begin{array}{l}
	%\mathfrak{S}(t,\eta,\lambda) =
	 \fr_x(t,\eta,\lambda) = \int_{\RR} x_{\rm m}(t, \tau) \mathcal{K}_\gs(\tau,\eta,\lambda)d\tau=\int_{\RR} \frac 1 {\gs} g\big(\frac \tau {\gs}\big)  x_{\rm m}(t, \tau)  e^{-i2\pi\tau\eta-i\pi \lambda \tau^2} d\tau.
 \end{array}
	\end{equation}
Then we have
\begin{equation}\label{def_MSSO_m2}
\begin{array}{l}
	 \fr_x(t,\eta,\lambda) =\sum_{k=0}^K x_k(t) \wb g\big( \gs(\eta -\phi^\gp_k(t)), \gs^2(\gl - \phi^{\gp \gp}_k(t))\big).
	\end{array}
 \end{equation}

In following, $\gs=\frac {c_0}{\ga^2}$ as in Theorem \ref{theo:LQFT}.  In addition, since in general $\ga$ is small,  we assume $\gs\ge 1$ for simplicity of presentation. Furthermore,
since $x, t$ will be fixed throughout the proof, % and the statement of the theorem,
in the following we use  $\mathfrak{S}(\eta,\lambda)$, $\mathfrak{R}(\eta,\lambda)$,  $\mathcal{G}$ and $\mathcal{G}_{\ell}$ to  denote $\mathfrak{S}_x(t, \eta, \lambda)$, $\mathfrak{R}_x(t, \eta,\lambda)$, $\mathcal{G}(t)$ and $\mathcal{G}_{\ell}(t)$ respectively.  First we establish a few lemmas.

\begin{lem}
For $x(t)\in A_\ga$, let $x_{\rm m}(t, \tau)$ be the linear chirp approximation to $x(t+\tau)$ at time $t$ defined above. Then
\begin{equation}
\label{x_xm_error}
\begin{array}{l}
|x(t+\tau)-x_{\rm m}(t, \tau)|\le M\big(B_1 \ga^3|\tau| +\frac {\pi}3 B_2 \ga^7 |\tau|^3\big).
\end{array}
\end{equation}
\end{lem}

\begin{proof} The proof is straightforward. Indeed, by \eqref{cond_A} and \eqref{cond_phi},
\begin{equation*}
\begin{array}{l}
|x(t+\tau)-x_{\rm m}(t, \tau)|=|x_{\rm r}(t, \tau)|
\\ ~~~~~~~~~~~~~~~~~~~~~~~~\le \sum_{k=0}^K\Big\{ |A_k(t+\tau)-A_k(t)| + A_k(t) ~\big|i2\pi \big(\phi_k(t+\tau)-\phi_k(t)-\phi_k^\gp(t) \tau- \frac 12\phi^{\gp \gp}_k(t)\tau^2\big) \big|\Big\}
\\ ~~~~~~~~~~~~~~~~~~~~~~~~\le \sum_{k=0}^K\Big\{ A_k(t) B_1 \ga^3 |\tau|  + A_k(t)  2\pi \sup_{\xi \in \RR}  \frac 16 \big |\phi{'''}_k (\xi) \tau^3 \big|
\\ ~~~~~~~~~~~~~~~~~~~~~~~~\le M(t) B_1 \ga^3 |\tau|  +   \sum_{k=0}^K A_k(t) \frac \pi 3 B_2 \ga^7 |\tau|^3
\\~~~~~~~~~~~~~~~~~~~~~~~~=M\big(B_1 \ga^3|\tau| +\frac {\pi}3 B_2 \ga^7 |\tau|^3\big),
\end{array}
\end{equation*}
 as desired.
\end{proof}

\begin{lem}
For $x(t)\in A_\ga$, let $\fs(\eta, \gl)$ be its transform by CT and $\fr(\eta, \gl)$ be the approximation of $\fs(\eta, \gl)$ with linear chirps defined by \eqref{def_MSSO_m1}. Then
\begin{equation}\label{S_R_error}
\begin{array}{l}
\big|\fs(\eta, \gl)-\fr(\eta, \gl)\big|\le \ga M \big(B_1 c_0N +\frac {\pi}3 B_2  c_0^3 N ^3\big).
\end{array}
\end{equation}
\end{lem}
\begin{proof} By \eqref{x_xm_error} and the facts $g\ge 0$,  supp$g\subseteq [-N, N]$, we have
\begin{equation*}
\begin{array}{l}
\big|\fs(\eta, \gl)-\fr(\eta, \gl)\big|=\Big| \int_\RR (x(t+\tau)-x_{\rm m}(t, \tau))\frac 1{\gs}g(\frac \tau{\gs}) e^{-i2\pi (\eta \tau+\frac 12 \gl \tau^2)}d\tau\Big|
\\~~~~~~~~~~~~~~~~~~~~~~~\le\int_{-N\gs }^{N\gs}  M\big(B_1 \ga^3|\tau| +\frac {\pi}3 B_2 \ga^7 |\tau|^3\big)
\frac 1{\gs}g(\frac \tau{\gs}) d\tau
\\~~~~~~~~~~~~~~~~~~~~~~~\le M\big(B_1 \ga^3 N \gs +\frac {\pi}3 B_2 \ga^7  \gs^3 N^3\big)
\\~~~~~~~~~~~~~~~~~~~~~~~=\ga M \big(B_1 c_0N +\frac {\pi}3 B_2  c_0^3 N ^3\big).
\end{array}
\end{equation*}
\end{proof}

\begin{lem}\label{lem:nonoverlap}
Let $\cG_\ell, 0\le \ell \le K$ be the sets defined by \eqref{def_Gell}. If $\alpha\le \frac{\mu\sqrt{c_0 \gt}}{4 M L}$, then $\cG_\ell$ are nonoverlapping, namely, $\cG_\ell \cap \cG_k=\emptyset$ for $\ell\not=k$.
\end{lem}

\begin{proof} Assume $(\eta, \gl)\in \cG_\ell \cap \cG_k$ for some $\ell\not=k$.
By the definition of $\cG_\ell$, we have
\begin{equation*}
\begin{array}{l}
|\phi'_k(t)-\phi'_\ell(t)|+\rho |\phi''_k(t)-\phi''_\ell(t)|\\
\le |\phi'_k(t)-\eta|+\rho |\phi''_k(t)-\gl|+|\phi'_\ell(t)-\eta|+\rho |\phi''_\ell(t)-\gl|
\\\le \frac 1{\gs} \left( \gs |\phi'_k(t)-\eta|+\rho \gs^2 |\phi''_k(t)-\gl|\right)+\frac 1{\gs} \left( \gs |\phi'_\ell(t)-\eta|+\rho \gs^2 |\phi''_\ell(t)-\gl|\right) ~~ \hbox{(since $\gs\ge 1$)}
\\\le \frac 2\gs \Big (\frac {4ML}\mu \Big)^2 \le 2 \gt,
\end{array}
\end{equation*}

a contradiction to  the well-separated condition \eqref{def_sep_cond_cros}.
Thus the sets $\cG_\ell, 0\le \ell \le K$ are nonoverlapping.
\end{proof}

Since we assume $\gs\ge 1$, from \eqref{def_sep_cond_cros}, we have
\begin{equation}
 \label{extra_ineq}
 \begin{array}{l}
\gs |\phi'_k(t)-\phi'_\ell(t)|+\rho\;  \gs ^2 |\phi''_k(t)-\phi''_\ell(t)|\ge \gs \Big(|\phi'_k(t)-\phi'_\ell (t)|+\rho |\phi''_k(t)-\phi''_\ell(t)|
\ge 2\gs \gt.
\end{array}
 \end{equation}

 \begin{lem}
For $x(t)\in A_\ga$, let $\fr(\eta, \gl)$ be defined by \eqref{def_MSSO_m1}. If $\gs \gt \ge \big(\frac {4ML}\mu\big)^2$,  that is $\alpha\le \frac{\mu\sqrt{c_0 \gt}}{4 M L}$ when $\gs=\frac {c_0}{\ga^2}$,
then
\begin{equation}
\begin{array}{l}
\label{R_with_ell}
\big|\fr(\eta, \gl)-x_\ell(t) \wb g\big( \gs(\eta -\phi^\gp_\ell(t)), \gs^2(\gl - \phi^{\gp \gp}_\ell(t))\big) \big|\le \frac {M L}{\sqrt {\gs \gt}}, \forall (\eta, \gl)\in \cG_\ell,
\end{array}
\end{equation}
\begin{equation}
\begin{array}{l}
\label{R_with_ell1}
\big|\fr(\phi'_\ell(t), \phi''_\ell(t))-x_\ell(t)  \big|\le \frac {M L}{\sqrt {2 \gs \gt}}.
\end{array}
\end{equation}
\end{lem}

\begin{proof} By \eqref{def_MSSO_m2}, we have for any $(\eta, \gl)\in \cG_\ell$,
\begin{equation*}
 \begin{array}{l}
\big|\fr(\eta, \gl)-x_\ell(t) \wb g\big( \gs(\eta -\phi^\gp_\ell(t)), \gs^2(\gl - \phi^{\gp \gp}_\ell(t))\big) \big|
\\=\Big|\sum_{k\not= \ell } x_k(t) \wb g\big( \gs(\eta -\phi^\gp_k(t)), \gs^2(\gl - \phi^{\gp \gp}_k(t))\big)\Big|
\\ \le \sum_{k\not= \ell } A_k(t)
%\frac L
L \big( \gs|\eta -\phi^\gp_k(t)| +\rho \gs^2|\gl - \phi^{\gp \gp}_k(t) |\big)^{-\frac 12}\hbox{(by \eqref{inequality_g1})}
\\ \le \sum_{k\not= \ell } A_k(t)
L \big( \gs|\phi'_\ell(t) -\phi^\gp_k(t)| -\gs |\eta- \phi^{\gp}_\ell(t) |+ \rho \gs^2|\phi''_\ell(t) -\phi''_k(t)| -\rho \gs^2|\gl - \phi^{\gp \gp}_\ell(t) |\big)^{-\frac 12}
\\ \le \sum_{k\not= \ell } A_k(t)
L \big( 2\gs \gt- \big(\frac {4ML}\mu\big)^2\big)^{-\frac 12} \hbox{(by \eqref{extra_ineq} and \eqref{def_Gell})}
\\ \le \frac {ML}{\sqrt {\gs \gt}},
\end{array}
 \end{equation*}
since $\gs \gt \ge \big(\frac {4ML}\mu\big)^2$. This proves \eqref{R_with_ell}. The proof \eqref{R_with_ell1} is similar and the details are omitted.
\end{proof}

\bigskip

{\bf Proof of Theorem (a).} Let $(\eta, \gl)\in \cG$. Suppose $(\eta, \gl) \not \in G_\ell$ for any $\ell$.  Then
$$
\gs|\eta -\phi^\gp_k(t)| +\rho  \gs^2|\gl - \phi^{\gp \gp}_k(t) |> \big(\frac {4ML}\mu\big)^2.
$$
Hence,
\begin{equation*}
 \begin{array}{l}
\big|\fr(\eta, \gl)\big| =\Big|\sum_{k=0}^K x_k(t) \wb g\big( \gs(\eta -\phi^\gp_k(t)), \gs^2(\gl - \phi^{\gp \gp}_k(t))\big)\Big|
\\~~~~~~~~~~~\le \sum_{k=0}^K A_k(t)
\frac L{\big( \gs|\eta -\phi^\gp_k(t)| +\gs^2|\gl - \phi^{\gp \gp}_k(t) |\big)^{\frac 12}} ~~\hbox{(by \eqref{inequality_g1})}
\\ ~~~~~~~~~~~< \sum_{k=0 }^K \frac {A_k(t) L} {4ML/\mu}=\frac \mu4.
\end{array}
 \end{equation*}
This, together with \eqref{S_R_error}, implies
\begin{equation*}
 \begin{array}{l}
\big|\fs(\eta, \gl)\big| \le \big|\fs(\eta, \gl)-\fr(\eta, \gl)\big|+\big|\fr(\eta, \gl)\big|
\\~~~~~~~~~~~\le \ga M \big(B_1 c_0N +\frac {\pi}3 B_2  c_0^3 N ^3\big)+\frac \mu4
\\~~~~~~~~~~~\le  \frac \mu4 +\frac \mu 4=\frac \mu2,
\end{array}
 \end{equation*}
where the second last inequality follows \eqref{cond_alpha} on the condition for $\ga$. This leads to a contradiction that $(\eta, \gl)\in \cG$. Hence
%the assumption that
there must exist $\ell$ such that $(\eta, \gl)\in \cG_\ell$. Lemma \ref{lem:nonoverlap} has shown that  $\cG_\ell, 0\le \ell \le K$ are disjoint.
% follows from the fact that $\cG_k, 0\le k\le K$ are nonoverlapping as shown in Lemma \ref{lem:nonoverlap}.

Finally we show that each $\cG_\ell$ is non-empty. To this regard, we show $(\phi'_\ell(t), \phi''_\ell(t))\in \cG$. Indeed,
the following fact followed from \eqref{R_with_ell1}
\begin{equation*}
 \begin{array}{l}
\big|\fr(\phi'_\ell(t), \phi''_\ell(t))\big|\ge |x_\ell(t)|-\frac {M L}{\sqrt {2 \gs \gt}}\ge \mu-\frac \mu{4\sqrt 2}>\frac {3\mu}4
\end{array}
 \end{equation*}
implies
\begin{equation*}
 \begin{array}{l}
\big|\fs(\phi'_\ell(t), \phi''_\ell(t)) \big|\ge  \left|\fr(\phi'_\ell(t), \phi''_\ell(t))\right| -\big|\fs(\phi'_\ell(t), \phi''_\ell(t))-\fr(\phi'_\ell(t), \phi''_\ell(t))\big|
> \frac {3\mu}4 -\frac \mu 4=\frac \mu2.
\end{array}
 \end{equation*}
Hence $(\phi'_\ell(t), \phi''_\ell(t))\in \cG$.
\hfill $\square$

\bigskip
{\bf Proof of \eqref{abs_IA_est}.}  By the definition of $\widehat{\eta}_{\ell}$ and $\widehat{\lambda}_{\ell}$, \eqref{S_R_error} and \eqref{R_with_ell1}, we have
\begin{equation}\label{IA_est_ineq1}
\begin{array}{l}
|\fs(\widehat{\eta}_{\ell}, \widehat{\lambda}_{\ell})|\ge |\fs(\phi'_\ell(t), \phi''_\ell(t))|
\\~~~~~~~~~~~~~\ge |\fr(\phi'_\ell(t), \phi''_\ell(t))| - \ga M \big (c_0 N B_1+ \frac \pi 3 B_2 c_0^3 N^3\big)
\\~~~~~~~~~~~~~\ge A_\ell(t) - \frac {LM}{\sqrt{2 \gs \gt}}- \ga M \big (c_0 N B_1+ \frac \pi 3 B_2 c_0^3 N^3\big)
\\~~~~~~~~~~~~~\ge A_{\ell}(t)- \ga M \Big( \frac {L}{\sqrt{c_0 \gt}}+ c_0 N B_1+ \frac \pi 3  B_2 c_0^3 N^3\Big).
\end{array}
\end{equation}
On the other hand, by \eqref{S_R_error} and \eqref{R_with_ell}, we have

\begin{equation}\label{IA_est_ineq2}
\begin{array}{l}
|\fs(\widehat{\eta}_{\ell}, \widehat{\lambda}_{\ell})|\le |\fr(\widehat{\eta}_{\ell}, \widehat{\lambda}_{\ell})|+ \ga M \big (c_0 N B_1+ \frac \pi 3 B_2 c_0^3 N^3\big)
 \\~~~~~~~~~~~~~\le |x_\ell(t) \wb g\big( \gs(\wh \eta_\ell -\phi^\gp_\ell(t)), \gs^2(\wh \gl _\ell- \phi^{\gp \gp}_\ell(t))\big) \big| + \frac {LM}{\sqrt{\gs \gt}}+ \ga M \big (c_0 N B_1+ \frac \pi 3 B_2 c_0^3 N^3\big)
\\~~~~~~~~~~~~~\le A_\ell(t) + \ga M\Big(\frac {L}{\sqrt{c_0 \gt}}+ c_0 N B_1+ \frac \pi 3 B_2 c_0^3 N^3\Big)
 \hbox{(since $|\wb g(\eta, \gl)|\le 1$).}
\end{array}
\end{equation}
%\begin{eqnarray}
%\nonumber &&
%|\fs(\widehat{\eta}_{\ell}, \widehat{\lambda}_{\ell})|\le |\fr(\widehat{\eta}_{\ell}, \widehat{\lambda}_{\ell})|+ \ga M \big (c_0 N B_1+ \frac \pi 3 B_2 c_0^3 N^3\big) \\
%\nonumber && \le |x_\ell(t) \wb g\big( \gs(\wh \eta_\ell -\phi^\gp_\ell(t)), \gs^2(\wh \gl _\ell- \phi^{\gp \gp}_\ell(t))\big) \big| + \frac {LM}{\sqrt{\gs \gt}}
%+ \ga M \big (c_0 N B_1+ \frac \pi 3 B_2 c_0^3 N^3\big)\\
%&& \le A_\ell(t) + \ga M\Big(\frac {L}{\sqrt{c_0 \gt}}+ c_0 N B_1+ \frac \pi 3 B_2 c_0^3 N^3\Big) ~~ \hbox{(since $|\wb g(\eta, \gl)|\le 1$).}
%\label{IA_est_ineq2}
%%&&= A_{\ell}(t)+ \ga \Big( \frac {LM}{\sqrt{c_0 \gt}}+M c_0 N B_1+ \frac \pi 3 M B_2 c_0^3 N^3\Big).
%\end{eqnarray}
Relationship of  $|\fs(\widehat{\eta}_{\ell}, \widehat{\lambda}_{\ell})|$ and $A_\ell(t)$ in \eqref{IA_est_ineq1} and \eqref{IA_est_ineq2} leads to \eqref{abs_IA_est}.
\hfill $\square$

\bigskip
{\bf Proof of \eqref{phi_est}.}
Write $\fs(\widehat{\eta}_{\ell}, \widehat{\lambda}_{\ell})=|\fs(\widehat{\eta}_{\ell}, \widehat{\lambda}_{\ell})| e^{i2\pi \psi(t)}$ for some real-valued function $\psi(t)$. Since for any complex number $z_1, z_2$, $\big|| z_1|-|z_2|\big| \le |z_1-z_2|$, we have
\begin{equation*}
\begin{array}{l}
A_\ell(t) \big| |\wb g\big( \gs(\wh \eta_\ell -\phi^\gp_\ell(t)), \gs^2(\wh \gl_\ell - \phi^{\gp \gp}_\ell(t))\big)|-1\big|
\\\le
A_\ell(t) \big| \wb g\big( \gs(\wh \eta_\ell -\phi^\gp_\ell(t)), \gs^2(\wh \gl_\ell - \phi^{\gp \gp}_\ell(t))\big)e^{i2\pi \phi_\ell(t)}
- e^{-i2\pi \psi(t)}\big|
\\=\big| \wb g\big( \gs(\wh \eta_\ell -\phi^\gp_\ell(t)), \gs^2(\wh \gl_\ell - \phi^{\gp \gp}_\ell(t))\big)x_\ell(t)
- A_\ell(t)e^{-i2\pi \psi(t)}\big|
\\\le \big| \wb g\big( \gs(\wh \eta_\ell -\phi^\gp_\ell(t)), \gs^2(\wh \gl_\ell - \phi^{\gp \gp}_\ell(t))\big)x_\ell(t)- \fs(\wh{\eta}_{\ell}, \wh{\gl}_{\ell})\big|
 +\big| \fs(\wh{\eta}_{\ell}, \wh{\gl}_{\ell})-A_\ell(t)e^{-i2\pi \psi(t)}\big|
\\\le \big| \wb g\big( \gs(\wh \eta_\ell -\phi^\gp_\ell(t)), \gs^2(\wh \gl_\ell - \phi^{\gp \gp}_\ell(t))\big)x_\ell(t)- \fr(\wh{\eta}_{\ell}, \wh{\gl}_{\ell})\big|
+\big| \fs(\wh{\eta}_{\ell}, \wh{\gl}_{\ell})- \fr(\wh{\eta}_{\ell}, \wh{\gl}_{\ell})\big| +\big| |\fs(\wh{\eta}_{\ell}, \wh{\gl}_{\ell})|-A_\ell(t)\big|
\\\le \frac{ML}{\sqrt{\gs \gt}}+ \ga M \big(c_0 N B_1+ \frac \pi 3 M B_2 c_0^3 N^3\big)
+\ga M \big( \frac {L}{\sqrt{c_0 \gt}}+ c_0 N B_1+ \frac \pi 3  B_2 c_0^3 N^3\big)
\\= 2 \ga M \big( \frac {L}{\sqrt{c_0 \gt}}+ c_0 N B_1+ \frac \pi 3 B_2 c_0^3 N^3\big),
\end{array}
\end{equation*}
where the second inequality follows from \eqref{R_with_ell}, \eqref{S_R_error}, and \eqref{abs_IA_est}.
Thus $\wb g$ satisfies \eqref{inequality_g2} with
$\eta=\gs(\wh \eta_\ell -\phi^\gp_\ell(t)), \gl=\gs^2(\wh \gl_\ell - \phi^{\gp \gp}_\ell(t)$ and $\varepsilon=\ga$. Hence
 \eqref{phi_est} holds by the property (b) of the admissible window function $g$.
\hfill $\square$

\bigskip
{\bf Proof of \eqref{comp_est}.}  By \eqref{S_R_error} and  \eqref{R_with_ell1}, we have
\begin{equation*}
\begin{array}{l}
%\big|\mathfrak{S}(t,\hat{\eta}_\ell,\hat{\lambda}_\ell)-x_\ell (t)\big |\leq .
\big| \fs(\wh{\eta}_\ell, \wh{\gl}_\ell)-x_\ell (t)\big |
\\\leq \big| \fs(\wh{\eta}_\ell, \wh{\gl}_\ell)-\fs(\phi'_\ell(t), \phi''_\ell(t))\big|
+\big|\fs(\phi'_\ell(t), \phi''_\ell(t))-\fr(\phi'_\ell(t), \phi''_\ell(t))\big|
+\big| \fr(\phi'_\ell(t), \phi''_\ell(t))-x_\ell (t)\big |
\\\le \Big| \frac 1\gs \int_\RR g(\frac \tau \gs) \big( e^{-i2\pi\wh \eta_\ell\tau-i\pi \wh \gl_\ell \tau^2}- e^{-i2\pi\wh \phi'_\ell(t)\tau-i\pi \wh \phi''_\ell(t) \tau^2}\big)d\tau \Big|
+\ga M \big(c_0 N B_1+ \frac \pi 3 B_2 c_0^3 N^3\big) +\frac {LM}{\sqrt{\gs \gt}}
 \\\le \frac 1\gs \int_\RR g(\frac \tau \gs) \big| 2\pi\wh \eta_\ell \tau +\pi \wh \gl_\ell\tau^2- 2\pi\phi'_\ell(t)\tau -\pi \phi''_\ell(t)\tau^2 \big| d\tau
+\ga M \Big( \frac L{\sqrt{c_0 \gt}}+ c_0 N B_1+ \frac \pi 3  B_2 c_0^3 N^3\Big)
\\\le 2\pi |\wh \eta_\ell -\phi'_\ell(t)| \frac 1\gs \int_\RR g(\frac \tau \gs) |\tau| d\tau
+ \pi |\wh \gl_\ell -\phi''_\ell(t)| \frac 1\gs \int_\RR g(\frac \tau \gs) \tau^2 d\tau
+\ga M \Big( \frac {L}{\sqrt{c_0 \gt}}+ c_0 N B_1+ \frac \pi 3  B_2 c_0^3 N^3\Big)
\\ \le 2\pi |\wh \eta_\ell -\phi'_\ell(t) \gs N  + \pi |\wh \gl_\ell -\phi''_\ell(t)| \gs^2 N^2
+\ga M\Big( \frac {L}{\sqrt{c_0 \gt}}+c_0 N B_1+ \frac \pi 3  B_2 c_0^3 N^3\Big)
\\= o(1)+\ga M\Big( \frac {L}{\sqrt{c_0 \gt}}+ c_0 N B_1+ \frac \pi 3  B_2 c_0^3 N^3\Big),
\end{array}
\end{equation*}
where the last equality follows from \eqref{phi_est}. This completes the proof of \eqref{comp_est}.
\hfill $\square$

%\end{proof}

%\hfill$\blacksquare$

%ACKNOWLEDGMENT % OF SUPPORT AND DISCLAIMER:
%The authors thank the anonymous reviewers for their valuable comments.  The authors would like to thank Professor Hongbing Ji for  helpful discussions.

\section*{Acknowledgments}
The authors thank the anonymous reviewers for their valuable comments. 

%This work is partially supported by the Hong Kong Research Council, under Projects $\sharp$ 12300917 and $\sharp$ 12303218, and HKBU Grants $\sharp$ RC-ICRS/16-17/03 and $\sharp$ RC-FNRA-IG/18-19/SCI/01, and by the Simons Foundation, under grant $\sharp$ 353185.

This work was partially supported by 
the ARO under Grant $\sharp$ W911NF2110218,  HKBU Grant $\sharp$ RC-FNRA-IG/18-19/SCI/01, 
the Simons Foundation under Grant $\sharp$ 353185, 
the National Natural Science Foundation of China under Grants $\sharp$ 62071349.
%$\sharp$ 61972265  and  $\sharp$ 11871348

%\begin{thebibliography}{99}

\end{document}